\documentclass[11pt]{article}

\usepackage[numbers]{natbib}
\usepackage[left=1in,bottom=1in,right=1in,top=1in]{geometry}
\usepackage{authblk,float}
\usepackage[small,bf]{caption}
\usepackage{amsmath, amsthm, amssymb,multirow,hyperref}
\usepackage{footnote}
\usepackage{array}
\usepackage{graphics}
\usepackage{graphicx}
\usepackage{xargs}
\usepackage{amsfonts}
\usepackage{color}
\usepackage{epsfig,psfrag,url,verbatim}
\usepackage{mathtools, dsfont, stackrel, booktabs}
\usepackage{soul}
\usepackage[ruled,vlined,noend,noline]{algorithm2e}
\usepackage[parfill]{parskip}
\usepackage{subcaption}
\usepackage{titlesec}
\usepackage[doublespacing]{setspace}
\usepackage{enumitem}

\titlespacing*{\section}
{0pt}{1.ex plus .3ex minus .2ex}{1.ex plus .2ex}
\titlespacing*{\subsection}
{0pt}{1.ex plus .3ex minus .2ex}{1.ex plus .2ex}
\titlespacing*{\subsubsection}
{0pt}{1.ex plus .3ex minus .2ex}{1.ex plus .2ex}
\newtheorem{theorem}{Theorem}
\newtheorem{lemma}{Lemma}

\newtheorem{definition}{Definition}
\newtheorem{remark}{Remark}

\newtheorem{assumption}{Assumption}

\newcommand{\pparagraph}{\noindent\textbf}
\newcommand{\sbt}{\mathrm{s.t.}}

\newcommand\eqdef{\mathrel{\overset{\makebox[0pt]{\mbox{\normalfont\tiny\sffamily def}}}{=}}}

\DeclareMathOperator*{\argmin}{arg\,min}
\DeclareMathOperator*{\argmax}{arg\,max}

\DeclareMathOperator{\sign}{sign}

\newenvironment{myarray}[2][1]
  {\def\arraystretch{#1}\array{#2}}
  {\endarray}

\begin{document}
\title{Fast Best Subset Selection: Coordinate Descent and Local Combinatorial Optimization Algorithms}

\date{July, 2019}

\author{Hussein Hazimeh\thanks{Operations Research Center. Email: {\texttt{hazimeh@mit.edu}}} }
\author{
Rahul Mazumder\thanks{MIT Sloan School of Management, Operations Research Center and MIT Center for Statistics.    Email: {\texttt{rahulmaz@mit.edu}} }}  
 
\affil{Massachusetts Institute of Technology}

\maketitle

\begin{abstract}
The $L_0$-regularized least squares problem (aka \emph{best subsets}) is central to sparse statistical learning and has attracted significant attention across the wider statistics, machine learning, and optimization communities. Recent work has shown that modern mixed integer optimization (MIO) solvers can be used to address small to moderate instances of this problem. In spite of the usefulness of $L_0$-based estimators and generic MIO solvers, there is a steep computational price to pay when compared to popular sparse learning algorithms (e.g., based on $L_1$ regularization). In this paper, we aim to push the frontiers of computation for a family of $L_0$-regularized problems with additional convex penalties. We propose a new hierarchy of necessary optimality conditions for these problems.
We develop fast algorithms, based on coordinate descent and local combinatorial optimization, that are guaranteed to converge to solutions satisfying these optimality conditions. From a statistical viewpoint, an interesting story emerges. When the signal strength is high, our combinatorial optimization algorithms have an edge in challenging statistical settings. When the signal is lower, pure $L_0$ benefits from additional convex regularization. We empirically demonstrate that our family of $L_0$-based estimators can outperform the state-of-the-art sparse learning algorithms in terms of a combination of prediction, estimation, and variable selection metrics under various regimes (e.g., different signal strengths, feature correlations, number of samples and features). Our new open-source sparse learning toolkit \texttt{L0Learn} (available on CRAN and Github) reaches up to a three-fold speedup (with $p$ up to $10^6$) when compared to competing toolkits such as~\texttt{glmnet} and \texttt{ncvreg}. 
\end{abstract}

\setstretch{1.35}

\section{Introduction}
The ongoing surge in high-dimensional data has drawn a lot of attention to sparse learning across several scientific communities. Indeed, sparsity can be very effective in high-dimensional settings as it leads to compact models that can be easier to interpret~\cite{buhlmann2011statistics,hastie2015statistical}. 
We consider the usual linear regression setup with $y = X \beta + \epsilon$, where $y \in \mathbb{R}^n$ is the response, $X \in \mathbb{R}^{n \times p}$ is the model matrix, $\beta\in \mathbb{R}^p$ is the vector of regression coefficients, and $\epsilon \in \mathbb{R}^n$ is a noise vector. 
%We will assume that $y$ and the columns of $X$ are mean-centered and standardized to have unit $L_2$-norm.
We will assume that the columns of $X$ are standardized to have unit $L_2$-norm, and we ignore the intercept term to simplify the presentation. Our goal is to estimate $\beta$ under the assumption that it is sparse (i.e., has few nonzeros)---a common desiderata in the high-dimensional learning framework with $p \gg n$~\cite{buhlmann2011statistics,hastie2015statistical}.
A natural and direct way to 
obtain such a sparse estimator is by minimizing the 
least squares loss with an $L_0$-norm\footnote{The $L_0$-(pseudo) norm of $\beta$, i.e., $\|\beta\|_0$ counts the number of nonzeros in $\beta$.} penalty on $\beta$ \cite{miller2002subset}. Statistical (optimality) properties of this estimator have been extensively studied~\cite{greenshtein2006best,raskutti2011minimax,zhang2012general,zhang2014lower}.
{Many appealing alternative sparsity-inducing estimators have been proposed in the literature based on Lasso~\cite{tibshirani1996regression}, stepwise regression, continuous nonconvex regularization~\cite{hastie2015statistical}, etc---each with different operating characteristics. Our focus in this paper is on the algorithmic aspects of $L_0$-based estimators.} Recent work~\cite{onbestsubset, mazumder2017subset} has brought to light an intriguing phenomenon: in low signal-to-noise-ratio (SNR) regimes, the vanilla version of $L_0$ penalization suffers from overfitting. One way to mitigate this problem is by considering a larger family of estimators that includes (in addition to the $L_0$-penalty) an $L_1$ or $L_2$ norm regularization~\cite{mazumder2017subset}. In this paper, we consider the following extended family of 
$L_0$-based estimators, i.e., $L_0 L_q$ regularized regression problems of the form:
\begin{equation} \label{problem:intro}
\hat{\beta} \in \argmin_{\beta \in \mathbb{R}^p} ~~ \frac 12 \| y - X \beta \|_{2}^2 + \lambda_0 \| \beta \|_0 + \lambda_q \| \beta \|_{q}^{q},
\end{equation}
where $q \in  \{ 1,2 \}$ determines the type of the additional regularization (i.e., $L_1$ or $L_2$). The regularization parameter $\lambda_0$ controls the number of nonzeros (i.e., selected variables) in $\hat{\beta}$, and 
$\lambda_q$ controls the amount of shrinkage induced by $L_q$ regularization.
In many regimes (and under suitable choices of $\lambda_0,\lambda_q$), estimators from Problem \eqref{problem:intro} exhibit superior statistical properties (variable selection, prediction, and estimation) compared to computationally friendlier alternatives (e.g., based on Lasso or stepwise regression)---see for example,~\cite{bertsimas2015best,zhang2014lower,bertsimas2017sparse, mazumder2017subset,raskutti2011minimax,zhang2012general}. 
In spite of its potential usefulness, Problem~\eqref{problem:intro} is NP-hard~\cite{natarajan1995sparse} and poses computational challenges. Recent work by \cite{bertsimas2015best} has shown that high-quality solutions can be obtained for the cardinality-constrained least squares problem via mixed integer optimization (MIO), in the order of minutes when $p \sim 1000$. However, efficient solvers for the Lasso (e.g., \texttt{glmnet} \cite{glmnet}) can address much larger problems within a second. Our goal is to bridge this gap in computation time by developing fast solvers that can obtain high-quality (approximate) solutions to Problem \eqref{problem:intro} for large and challenging instances (e.g., $p \sim 10^6$ and small $n$). This will allow performing systematic large-scale experiments to gain a deeper understanding of the statistical properties of $L_0$-based estimators and their differences with the state of the art. Such an understanding is currently limited due to computational considerations.

Our approach is based on two complementary algorithms: \textbf{(i)} cyclic coordinate descent (CD) for quickly finding solutions to Problem \eqref{problem:intro}, and \textbf{(ii)} novel combinatorial search algorithms, which help improve solutions from (i). Particularly, the solutions obtained by (ii) cannot be improved by making small changes to their support. 
We establish novel convergence guarantees for our algorithms. 
We also address delicate implementation aspects of our algorithms and provide \texttt{L0Learn}: an open-source and efficient R/C++ toolkit available on CRAN at \url{https://CRAN.R-project.org/package=L0Learn} and on Github at \url{https://github.com/hazimehh/L0Learn}.

\pparagraph{Current Landscape and Related Work:} 
Our main focus is on the computational aspects of Problem~\eqref{problem:intro}.
We contextualize our contribution within the rather large and impressive literature on algorithms for sparse regression---see for e.g.,~\cite{BeckSparsityConstrained, bertsimas2015best} for an overview. We broadly categorize the main existing algorithms into two categories:
\begin{itemize}[leftmargin=*]
    \item \textbf{Proxy Algorithms and Heuristics}: Proxy algorithms use a proxy/surrogate to the $L_0$ norm (e.g., $L_1$ norm or non-convex penalties such as MCP and SCAD \cite{tibshirani1996regression,zhang2010nearly, fan2001variable}). Fast solvers have been devised for these proxies (e.g., \cite{glmnet, ncvreg,sparsenet})---they typically result in good solutions (though not optimal for non-convex problems). Another approach is to use heuristics to find approximate solutions to Problem~\eqref{problem:intro} with $\lambda_{q}=0$. Popular methods include: (greedy) stepwise regression~\cite{hastie2015statistical}, iterative hard thresholding (IHT)~\cite{blumensath2009-acha,bertsimas2015best},
   greedy CD~\cite{BeckSparsityConstrained} and randomized CD~\cite{randomCDL0}.
    
    \item \textbf{Exact Algorithms}: These approaches \emph{exactly} solve an optimization problem involving the $L_0$ norm. \cite{bertsimas2015best} use MIO to compute near-optimal solutions for least squares with a cardinality constraint for $p\approx 1000$. \cite{bertsimas2017sparse} propose a cutting plane method for a similar problem, which works well with mild sample correlations and a sufficiently large $n$. \cite{mazumder2015discrete} use mixed integer linear optimization for solving an $L_0$-variant of the Dantzig Selector.
\end{itemize}
In spite of their usefulness, exact algorithms are usually accompanied by a steep increase in computational cost,
placing them at a disadvantage compared to faster alternatives~\cite{onbestsubset}.
To this end, our approach borrows the computational strengths of the proxy algorithms while maintaining a notion of ``local combinatorial exactness''---i.e., making small perturbations to the support of the solution cannot improve its objective. Similar to the proxy algorithms, we employ cyclic CD as one of our main workhorses. We note that standard results on the convergence of cyclic CD \cite{Tseng01} do not apply for our problem, and one of our contributions is rigorously establishing its convergence. 
A novelty of our work is the use of local combinatorial search to obtain high quality solutions. Our attention to the delicate computational aspects make our proposed algorithms comparable (and at times faster) in speed to the fastest proxy algorithms (e.g., \texttt{glmnet} and \texttt{ncvreg}).

\pparagraph{Contributions:}
We summarize our key contributions below: 

\begin{enumerate}[leftmargin=*]
\item We introduce a new family of necessary optimality conditions for Problem~\eqref{problem:intro}, leading to a hierarchy of classes of local minima. Classes higher up in the hierarchy are of better quality.

\item We propose new algorithms based on cyclic CD and local combinatorial search to obtain these local minima. We present a novel convergence analysis of the algorithms. We formulate the local combinatorial search problems as structured MIO problems and develop efficient solvers for special cases. Our local search algorithms can run in seconds to minutes when $p$ is in the order of $10^3$ to $10^6$.

\item Our open-source R/C++ toolkit, \texttt{L0Learn}, often runs faster than state-of-the-art toolkits (e.g., \texttt{glmnet} and \texttt{ncvreg}). Typical speedups (of a version of our algorithm) range between $25\%$ to $300\%$ for $p$ up to $10^6$ and $n \approx 10^3$.

\item Experiments on real and synthetic datasets suggest that our algorithms do a good job in optimizing Problem~\eqref{problem:intro},
with solutions often found to be similar to that of exact MIO methods, but with significantly shorter run times. 
In terms of statistical performance, our algorithms are found to be superior in terms of a combination of metrics (estimation, prediction, and variable selection), compared to state-of-the-art methods for sparse learning. 
\end{enumerate}

\subsection{Notation}
We use the following notation throughout paper. We denote the set $\{1,2,\dots,p\}$ by $[p]$, the canonical basis for $\mathbb{R}^p$ by $e_1, \dots, e_p$, and the standard Euclidean norm by $\| \cdot \|$. Similarly, $\|\cdot\|_{q}$ denotes the standard $L_{q}$-norm with $q \in \{0,1,2,\infty\}.$
{For any $\theta \in \mathbb{R}^p$ and $i \in [p]$, we define $\widetilde{\theta}_i = \langle  y - \sum\nolimits_{j \neq i} X_j \theta_j, X_i \rangle $. For any vector $u \in \mathbb{R}^{k}$, we define $\sign(u)\in \mathbb{R}^{k}$ as a vector whose $i$th component is given by $\sign(u_i) = u_i/|u_i|$ if $u_i \neq 0$ and $\sign(u_i) \in [-1,1]$ if $u_i = 0$.}
We denote the support of $\beta \in \mathbb{R}^p$ by 
$\text{Supp}(\beta)=\{i: \beta_{i} \neq 0, i \in [p]\}$. For $S \subseteq [p]$, we let $\beta_S \in \mathbb{R}^{|S|}$ denote the subvector of $\beta$ with indices in $S$. Similarly, $X_S$ denotes the submatrix of $X$ with column indices $S$. 
We use $U^S$ to denote the $p \times p$ matrix whose $i$th column is $e_i$ if $i \in S$ and zero otherwise. Thus, $(U^S \beta)_i = \beta_i$ if $i \in S$ and $(U^S \beta)_i = 0$ if $i \notin S$.

Proofs of lemmas and theorems are included in the supplementary material.

\section{Necessary Optimality Conditions} \label{sec:optimalityconditions}

We present a family of necessary optimality conditions for Problem~\eqref{problem:intro}, leading to different classes of local minima\footnote{As we argue in Section~\ref{section:stationarysolutions} these also satisfy the usual notion of a local minimizer in nonlinear optimization.}.
Our methodology is centered around the following problem:
\begin{equation}
\label{problem:main}
\min_{\beta \in \mathbb{R}^p}~~F(\beta) \eqdef f(\beta) + \lambda_0 \|\beta\|_0,
\end{equation}
where $f(\beta)$ is the least squares term with additional convex regularizers:
\begin{align}
	f(\beta) \eqdef \frac{1}{2} \| y - X \beta \|^{2}  + \lambda_1 \|\beta\|_1 + \lambda_2 \|\beta\|^{2}_2.
\end{align}
We will use the shorthands: (i) $(L_0L_2)$ to denote Problem~\eqref{problem:main} with $\lambda_1=0$ and $\lambda_2>0$; (ii) $(L_0L_1)$ to denote Problem~\eqref{problem:main} with $\lambda_1>0$ and $\lambda_2 = 0$; and (iii) $(L_0)$ to denote Problem~\eqref{problem:main} with $\lambda_1=\lambda_2=0$. Unless specified otherwise, we will assume that $\lambda_0>0$.

Next, we present an overview of the different classes of local minima (minima for short) that we study---this is then followed by a more formal treatment.
\begin{itemize}[leftmargin=*]
    \item \textbf{Stationary Solutions}: Solutions where the directional derivative is non-negative in any direction.
    %%%(see Section \ref{section:stationarysolutions}).
    \item \textbf{Coordinate-wise (CW) Minima}: Solutions where optimizing w.r.t. one coordinate at a time (while keeping others fixed) cannot improve the objective.
    %%(see Section \ref{section:CWminima}).
    \item \textbf{PSI($k$) Minima}: These are stationary solutions where (i) removing any subset (of size at most $k$) from the support, (ii) adding any subset (of size at most $k$) to the support, and (iii) optimizing over the newly added subset, cannot improve the objective.
    %%(see Section \ref{section:psi}).
    \item \textbf{FSI($k$) Minima}: These are similar to PSI($k$) minima except that in step (iii), if we optimize over the whole new support, the objective does not improve.
    %%(see Section \ref{section:fsi}).
    \item \textbf{IHT Minima}: These are fixed points arising from the popular IHT algorithm.
    %%(see Section \ref{section:IHT}).
\end{itemize}
We also establish the following hierarchy among the different classes introduced above:
\begin{equation}\label{hier-equation-display}
\textsc{Hierarchy:}~~~~~~~~~\begin{aligned} \text{FSI$(k)$} \\
\text{Minima} \end{aligned} ~\subseteq~ \begin{aligned} \text{PSI$(k)$} \\
\text{Minima} \end{aligned} ~\subseteq~ \begin{aligned} \text{CW} \\ 
\text{Minima} \end{aligned} ~\subseteq~ \begin{aligned}
\text{IHT} \\ 
\text{Minima} \end{aligned} ~\subseteq~ \begin{aligned} \text{Stationary} \\
\text{Solutions} \end{aligned}
~~~~~~~~~~~~~~\end{equation}
In the above hierarchy, stationary solutions are the weakest. As we move from the right to left, the classes become smaller (i.e, satisfy more restrictive necessary optimality conditions) until reaching the most restrictive class: FSI$(k)$ minima. Moreover, for sufficiently large $k$, FSI($k)$ and PSI($k$) minima coincide with the class of global minimizers of Problem \eqref{problem:main}. We now present a formal treatment of the classes of minima introduced above.

\subsection{Stationary Solutions} \label{section:stationarysolutions}
For a function $g: \mathbb{R}^p \to \mathbb{R}$ and a vector $d \in \mathbb{R}^p$, we denote the (lower) directional derivative~\cite{bertsekas2016nonlinear} of $g$ at $\beta$ in the direction $d$ by:
%\begin{equation*}
$g'(\beta;d) \eqdef \liminf_{\alpha \downarrow 0} (g(\beta + \alpha d) - g(\beta))/\alpha$.
%\end{equation*}
Directional derivatives play an important role in describing necessary optimality conditions for continuous optimization problems \cite{bertsekas2016nonlinear}. 
Although $F(\beta)$ is not continuous, it is insightful to use the notion of a directional derivative to arrive at a basic definition of stationarity for Problem~\eqref{problem:main}.
\begin{definition}\label{stat-soln-def}{(Stationary Solution)}
	A vector $\beta^{*} \in \mathbb{R}^p$ is a stationary solution for Problem (\ref{problem:main}) if for every direction vector $d \in \mathbb{R}^p$, the lower directional derivative satisfies: $F'(\beta^{*};d) \geq 0$.
\end{definition}
\textcolor{black}{Let $\nabla f(\beta) \in \mathbb{R}^{p}$ denote a subgradient of $f(\beta)$. If $\beta$ has a support $S$, the notation $\nabla_S f(\beta)$ refers to the components of $\nabla f(\beta)$ restricted to $S$.} 
Lemma~\ref{lemma:stationarity} gives an alternate characterization of Definition~\ref{stat-soln-def}.
%characterization of $F'(\beta;d) \geq0$ for all $d \in \mathbb{R}^{p}$.
\begin{lemma}\label{lemma:stationarity}
Let $\beta^{*} \in \mathbb{R}^{p}$ with support $S$. $\beta^{*}$ is a stationary solution for Problem~\eqref{problem:main} iff $\nabla_S f(\beta^{*}) = 0$.
\end{lemma}

Note that $\nabla_S f(\beta^{*}) = 0$ can be explicitly written as:
\begin{align}
\label{eq:stationary}
	\beta^{*}_{i} = \sign(\widetilde{\beta}^{*}_i) \frac{|\widetilde{\beta}^{*}_i| - \lambda_1}{1+2\lambda_2} \ ~~\text{ and }~~ \ 
	|\widetilde{\beta}^{*}_i| > \lambda_1~~~~\text{for all $i \in \text{Supp}(\beta^{*})$},
\end{align}
where we recall that $\widetilde{\beta}^{*}_i \eqdef \langle   y - \sum\nolimits_{j \neq i} X_j \beta^{*}_j, X_i \rangle$. Characterization (\ref{eq:stationary}) suggests that a stationary solution $\beta^{*}$ does not depend on $\lambda_0$ and does not impose any restriction on the coordinates outside the support. {Moreover, it can be readily verified that a stationary solution to Problem \eqref{problem:main} satisfies the traditional definition of a local minimum in nonlinear optimization, i.e., if $\beta^{*}$ is a stationary solution, then there exists a $\delta > 0$ such that $F(\beta^{*}) \leq F(\beta)$ for any $\beta$ satisfying $\| \beta - \beta^{*} \| < \delta$.}

\subsection{Coordinate-wise (CW) Minima} \label{section:CWminima}
We consider a 
class of stationary solutions inspired by 
coordinate-wise algorithms~\cite{BeckSparsityConstrained,bertsekas2016nonlinear,Tseng01}.

\begin{definition}{(CW Minimum)}
	\label{def:cwminima}
	A vector $\beta^{*} \in \mathbb{R}^p$ is a CW minimum for Problem~\eqref{problem:main} if for every $i \in [p]$, $\beta^{*}_i$ is a minimizer of $F(\beta^{*})$ w.r.t. the $i$th coordinate (with others held fixed), i.e., 
	\begin{align}
	\label{eq:cwmindef1}
	\beta^{*}_i \in \argmin_{\beta_i \in \mathbb{R}} F(\beta^{*}_1,\dots, \beta^{*}_{i-1}, \beta_{i}, \beta^{*}_{i+1}, \dots, \beta^{*}_{p} ).
	\end{align}
\end{definition}
As every column of $X$ has unit $L_{2}$-norm,
$\beta^{*}_i$ is given by the following thresholding operator $\widetilde{T}$:
\begin{align}
\label{eq:thresholdingmap}
\widetilde{T}(\widetilde{\beta}^{*}_i, \lambda_0, \lambda_1, \lambda_2) \eqdef \argmin_{\beta_i \in \mathbb{R}} \left \{ \frac{1 + 2\lambda_2}{2} \Big(\beta_i - \frac{\widetilde{\beta}^{*}_i}{1+2\lambda_2} \Big)^{2} + \lambda_1 |\beta_i| + \lambda_0 \mathds{1}[\beta_i \neq 0]  \right \},
\end{align}
where $\{\lambda_{i}\}_{0}^{2}$ and $\widetilde{\beta}^{*}_{i}$ are fixed, and the set $\widetilde{T}(\widetilde{\beta}^{*}_i, \lambda_0, \lambda_1, \lambda_2)$ is described below.

\begin{lemma}\label{lemma:univariate}
	Let $\widetilde{T}$ be the thresholding operator defined in (\ref{eq:thresholdingmap}). Then,
	\begin{align*}
			\widetilde{T}(\widetilde{\beta}^{*}_i, \lambda_0, \lambda_1, \lambda_2) = \begin{cases}
				\Big \{ \sign(\widetilde{\beta}^{*}_i) \frac{|\widetilde{\beta}^{*}_i| - \lambda_1}{1+2\lambda_2} \Big \} & \text{ if } \frac{|\widetilde{\beta}^{*}_i| - \lambda_1}{1+2\lambda_2} > \sqrt{2\lambda_0 \over 1+2 \lambda_{2}} \\
				\{ 0 \} & \text{ if } \frac{|\widetilde{\beta}^{*}_i| - \lambda_1}{1+2\lambda_2} < \sqrt{2\lambda_0 \over 1+2 \lambda_{2}} \\
				\Big \{ 0 , \sign(\widetilde{\beta}^{*}_i) \frac{|\widetilde{\beta}^{*}_i| - \lambda_1}{1+2\lambda_2} \Big \} & \text{ if } \frac{|\widetilde{\beta}^{*}_i| - \lambda_1}{1+2\lambda_2} = \sqrt{2\lambda_0 \over 1+2 \lambda_{2}}. 
				\end{cases}
	\end{align*}
\end{lemma}

Lemma~\ref{lemma:CW} presents an alternative characterization of CW minima. 

\begin{lemma}
\label{lemma:CW}
A vector $\beta^{*} \in \mathbb{R}^{p}$ is a CW minimum iff 
\begin{equation} \label{lemma:CW-eqn}
\begin{myarray}[1.5]{l c r}
		&\beta^{*}_i = \sign(\widetilde{\beta}^{*}_i) \frac{|\widetilde{\beta}^{*}_i| - \lambda_1}{1+2\lambda_2} \ \text{ and } \ 
			|\beta^{*}_i| \geq \sqrt{\frac{2 \lambda_0}{1+2\lambda_2}}, & \text{for every $i\in \text{Supp}(\beta^{*})$}  \\
	\text{and~~~~}    &\frac{|\widetilde{\beta}^{*}_i| - \lambda_1}{1+2\lambda_2} \leq \sqrt{\frac{2 \lambda_0}{1+2\lambda_2}}& \text{for every $i\notin \text{Supp}(\beta^{*})$}. 		
	\end{myarray}
    \end{equation}
	\end{lemma}

Comparing \eqref{lemma:CW-eqn} to (\ref{eq:stationary}), we see that the class of stationary solutions contains the class of CW minima, and the containment is strict (in general).
%%%for a general $X$ and $y$. 

\subsection{Swap Inescapable Minima}\label{section:swapinescapable}
We now introduce stationary solutions that further refine the class of CW minima, using notions from local combinatorial optimization. Given a CW minimum $\beta^{*}$, one might obtain a better solution by the following ``swapping'' operation: we set some nonzeros in $\beta^*$ to zero and allow some entries from outside the support of $\beta^*$ to be nonzero. Then, we 
optimize over the new support using one of the following rules: 
(a) {\em Partial Optimization}: we optimize only w.r.t. the coordinates added from outside the support or (b) {\em Full Optimization}: we optimize w.r.t. all the coordinates in the new support.
This may lead to a solution with a smaller objective value.
If the current solution cannot be improved using the swapping operation, we call $\beta^*$ a \emph{Swap Inescapable} minimum. Our proposal is inspired by the work of~\cite{BeckSparsityConstrained} for the cardinality-constrained problem---where, the authors suggest a special case of partial swap optimization involving {\em one} coordinate.
However, the problem studied here is different: we consider $L_0$-penalization (versus an $L_0$ constraint) and a non-smooth $f(\beta)$. Furthermore, we allow multiple coordinates to be swapped at once via 
partial or full optimization. 

\pparagraph{Partial Swap Inescapable (PSI) Minima:} 
We formally define {\em Partial Swap Inescapable} (PSI) minima, arising 
from the partial optimization step outlined above.
Recall that for any $L\subseteq [p]$, the $i$th coordinate of the vector $(U^L \beta)$ is $\beta_i$ if $i \in L$ and zero otherwise.

\begin{definition}{(PSI Minima)}
	\label{def:psi}
	Let $k$ be a positive integer. A vector $\beta^{*}$ with support $S$ is a PSI minimum of order $k$, denoted by PSI$(k)$, if it is a stationary solution and for every $S_1 \subseteq S$, $S_2 \subseteq  S^c$, with $|S_1| \leq k$, $|S_2| \leq k$, the following holds
	\begin{equation*}
	F(\beta^{*}) \leq \min\limits_{\beta_{S_2}} F(\beta^{*} - U^{S_1} \beta^{*} + U^{S_2} \beta).
	\end{equation*} 
\end{definition}
The following lemma characterizes PSI minima of order one, aka PSI($1$).
\begin{lemma} \label{lemma:psi1}
	A vector $\beta^{*} \in \mathbb{R}^{p}$ is a PSI($1$) minimum iff
	\begin{equation*}
    \begin{myarray}[1.5]{l c r}
&		\beta^{*}_i = \sign(\widetilde{\beta}^{*}_i) \frac{|\widetilde{\beta}^{*}_i| - \lambda_1}{1+2\lambda_2} \ \text{ and } \ 
		|\beta^{*}_i| \geq \max \Bigg \{ \sqrt{\frac{2 \lambda_0}{1+2\lambda_2} }
		, \max\limits_{j \notin \text{Supp}(\beta^{*})} \frac{ |\widetilde{\beta}^{*}_{ij}| - \lambda_1}{1+2\lambda_2}	\Bigg \},
        & \text{for $i\in \text{Supp}(\beta^{*})$} \\
\text{and}~~~&			\frac{|\widetilde{\beta}^{*}_i| - \lambda_1}{1+2\lambda_2} \leq \sqrt{\frac{2 \lambda_0}{1+2\lambda_2}},  &\text{for $i\notin \text{Supp}(\beta^{*})$}
	    \end{myarray}
    \end{equation*}
    where $\widetilde{\beta}^{*}_{ij} = \langle y - \sum_{l\neq i,j} X_l \beta_l, X_j \rangle$.
\end{lemma}
Lemmas~\ref{lemma:CW} and~\ref{lemma:psi1} suggest that PSI($1$) minima impose additional restrictions on the magnitude of nonzero coefficients,  when compared to CW minima.
The class of CW minima contains PSI($k$) minima for any $k$. Furthermore, as $k$ increases, the class of PSI($k$) minima becomes smaller---{till it coincides with the class of global minimizers of Problem~\eqref{problem:main}}. 

\pparagraph{Full Swap Inescapable (FSI) Minima:}
We formally define {\em Full Swap Inescapable} (FSI) minima, arising 
from the full optimization step outlined above.
\smallskip
\begin{definition}{(FSI Minima)}
	\label{def:fsi}
	Let $k$ be a positive integer. A vector $\beta^{*}$ with support $S$ is a FSI minimum of order $k$, denoted by FSI$(k)$, if for every $S_1 \subseteq S$ and $S_2 \subseteq  S^c$, such that $|S_1| \leq k$ and $|S_2| \leq k$, the following holds
	\begin{equation*}
	F(\beta^{*}) \leq \min_{\beta_{(S \setminus S_1) \cup S_2}} F(\beta^{*} - U^{S_1} \beta^{*} + U^{(S \setminus S_1) \cup S_2} \beta).
	\end{equation*}
\end{definition}
We note that for a fixed $k$, the class of PSI($k$) minima contains FSI($k$) minima, justifying a part of the hierarchy displayed in~\eqref{hier-equation-display}.
As $k$ increases, the class of FSI($k$) minima becomes smaller till it coincides with the set of global minimizers of Problem~\eqref{problem:main}. 
Sections~\ref{section:psialgorithm} and \ref{section:fsialgorithm} introduce algorithms to obtain PSI($k$) minima and FSI($k$) minima, respectively.

\subsection{Stationarity Motivated by Iterative Hard Thresholding (IHT)}
\label{section:IHT}
Proximal gradient algorithms such as IHT are popularly used for $L_0$-penalized least squares problems~\cite{blumensath2009-acha}. It is insightful to consider the class  of stationary solutions associated with IHT and study how they compare to CW minima. 
Let $f_d (\beta):= \frac{1}{2} \|y - X \beta\|^2 + \lambda_2 \|\beta\|^2$. The gradient of 
$f_{d}(\beta)$ is Lipschitz continuous with parameter $L$ (say), i.e., $\| \nabla  f_d (\beta) - \nabla f_d (\alpha)\| \leq L \| \beta - \alpha \|$ for all $\beta, \alpha \in \mathbb{R}^p$.
IHT applied to Problem \eqref{problem:main} performs the following updates:
\begin{align}\label{iht-updates-1}
	\beta^{k+1} \in \argmin_{\beta \in \mathbb{R}^p} \left\{ \frac{1}{2 \tau} \| \beta - (\beta^k - \tau \nabla f_d (\beta^k)) \|^2 + \lambda_1 \|\beta\|_1 + \lambda_0 \|\beta\|_0 \right\},
\end{align}
where $\tau > 0$ is a constant step size. We say that $\alpha \in \mathbb{R}^{p}$ is a fixed point of update~\eqref{iht-updates-1} if $\beta^k=\alpha$ leads to $\beta^{k+1}=\alpha$. This suggests another notion of 
stationarity (see Definition~\ref{defn-iht-min}) for Problem~\eqref{problem:main}.
To this end, consider Theorem~\ref{theorem:IHT} establishing the convergence of $\beta^{k}$ to a fixed point of update~\eqref{iht-updates-1}. 

\begin{theorem} 
\label{theorem:IHT}
	Let $L$ be defined as above. The sequence $\{\beta^k\}$ defined in~\eqref{iht-updates-1} converges to a fixed point $\beta^*$ of update~\eqref{iht-updates-1} for any $\tau < \frac{1}{L}$. Note that $\beta^{*}$ is a fixed point iff 
	\begin{equation}\label{iht-lemma-minima}
    \begin{myarray}[1.5]{ccc}
&			 \beta^{*}_i = \sign(\widetilde{\beta}^{*}_i) \frac{|\widetilde{\beta}^{*}_i| - \lambda_1}{1+2\lambda_2} \ \text{ and } \ 
			|\beta^{*}_i| \geq \sqrt{2 \lambda_0 \tau }&
\text{for $i\in \text{Supp}(\beta^{*})$} \\ 
\text{and}~~~&\frac{|\widetilde{\beta}^{*}_i| - \lambda_1}{1+2\lambda_2} \leq \sqrt{\frac{2 \lambda_0}{(1+2\lambda_2)^2 \tau}}& \text{	for $i\notin \text{Supp}(\beta^{*})$}
\end{myarray}
\end{equation}
\end{theorem}
\smallskip
\begin{definition}\label{defn-iht-min}
A vector $\beta^*$ is an IHT minimum for Problem~\eqref{problem:main} if it satisfies~\eqref{iht-lemma-minima} for $\tau < \frac{1}{L}$.
\end{definition}

The following remark shows that the class of 
IHT minima contains the family of CW minima. 

\smallskip

\begin{remark} \label{remark:IHT}
Let $M$ be the largest eigenvalue of $X^T X$ and take $L = M + 2\lambda_2$.
By Theorem~\ref{theorem:IHT},  
any $\tau < \frac{1}{M + 2\lambda}$ ensures the convergence of updates~\eqref{iht-updates-1}. Since the columns of $X$ are normalized, we have $M \geq 1$. 
Comparing~\eqref{iht-lemma-minima} 
to~\eqref{lemma:CW-eqn} 
we see that the class of IHT minima includes CW minima. Usually, 
for high-dimensional problems, $M \gg 1$---making the class of IHT minima much larger than CW minima 
(see Section~\ref{section:cdcomparison} for numerical examples).
\end{remark}

We have now explained the full hierarchy among the classes of local minima presented in~\eqref{hier-equation-display}. Section~\ref{section:algos} discusses algorithms to obtain solutions that belong to these classes. 

\section{Algorithms}
\label{section:algos}
Section~\ref{section:CD} presents a cyclic CD algorithm which converges to CW minima, and Section~\ref{sec:local-algos} discusses local combinatorial optimization to obtain PSI$(k)$/FSI$(k)$ minima.

\subsection{Cyclic Coordinate Descent}
\label{section:CD}

Our main workhorse is cyclic CD~\cite{bertsekas2016nonlinear} with full minimization in every coordinate---we also include some additional tweaks for reasons we discuss shortly. With initialization $\beta^0$, we update the first coordinate (with others fixed) to get $\beta^1$, and continue the updates as per a cyclic rule.
Let $\beta^{k}$ denote the solution obtained after performing $k$ coordinate updates. 
Then, $\beta^{k+1}$ is obtained by updating the $i$th coordinate (with others held fixed) via:
 \begin{align} \label{eq:argminF}
\beta^{k+1}_{i} \in \argmin_{\beta_i \in \mathbb{R}} F(\beta_1^{k},\dots, \beta_{i-1}^{k}, \beta_{i}, \beta_{i+1}^{k}, \dots, \beta_{p}^{k} ),
\end{align}
where $i = (k+1)~\text{mod}~(p+1)$. Recall that the operator $\widetilde{T}(\widetilde{\beta}_i, \lambda_0, \lambda_1, \lambda_2)$ (defined in~\eqref{eq:thresholdingmap}) describes solutions of Problem~\eqref{eq:argminF}. Specifically, it returns two solutions when $\tfrac{|\widetilde{\beta_i}| - \lambda_1}{1+2\lambda_2} = \sqrt{2\lambda_0 \over 1+2 \lambda_{2}}$. In such a case, we consistently choose one of these solutions\footnote{This convention is used for a technical reason in the context of our proof of Theorem~\ref{theorem:convergence}.}, namely the nonzero solution. Thus, we use the new operator (note the use of $T$ instead of $\widetilde{T}$):
\begin{align}
	\label{eq:thresholding}
	T(\widetilde{\beta}_i, \lambda_0, \lambda_1, \lambda_2) \eqdef \begin{cases}
				 \sign(\widetilde{\beta_i}) \frac{|\widetilde{\beta_i}| - \lambda_1}{1+2\lambda_2} & \text{ if } \frac{|\widetilde{\beta_i}| - \lambda_1}{1+2\lambda_2} \geq \sqrt{2\lambda_0 \over 1+2 \lambda_{2}} \\
			    0  & \text{ if } \frac{|\widetilde{\beta_i}| - \lambda_1}{1+2\lambda_2} < \sqrt{2\lambda_0 \over 1+2 \lambda_{2}}
			    \end{cases}
\end{align}
for update (\ref{eq:argminF}). In addition to the above modification, we introduce ``spacer steps'' that are occasionally performed during the course of the algorithm to \emph{stabilize} its behavior\footnote{The spacer steps are introduced for a technical reason, and our proof of convergence of CD relies on this to ensure the stationarity of the algorithm's limit points.}---spacer steps are commonly used in the context of continuous nonlinear optimization problems (e.g., see~\cite{bertsekas2016nonlinear}).  We perform a spacer step as follows. 
Let $C$ be an a-priori fixed positive integer. We keep track of the supports encountered so far, and when a certain support $S$ (say) appears for $Cp$-many times, we perform one pass of cyclic CD to minimize the continuous function $\beta_{S} \mapsto f(\beta_S)$. This entails updating every coordinate in $S$ via the operator: $T(\widetilde{\beta}_i, 0, \lambda_1, \lambda_2)$ (see Subroutine~1). 

Algorithm~1 summarizes the above procedure. Count[S] is an associative array that stores the number of times a support $S$ appears---it takes $S$ as a key and returns 
the number of times $S$ has appeared so far. Count[S] is initialized to zero for any $S$ that appears for the first time during the course of the algorithm.
Note that in the worst case,
storing Count[] may require an exponential (in $p$) amount of memory. However, in practice, only one or few supports appear for $Cp$ many times and need to be maintained in Count[].

\begin{algorithm}[!h]
	\DontPrintSemicolon
	\SetKwInOut{Input}{Input}\SetKwInOut{Output}{Output}
	\Input{Initial Solution $\beta^{0}$, Positive Integer $C$}
	$k \gets 0$ \\
	\While{Not Converged}
	{
		\For{$i$ in $1$ to $p$}
		{
			$\beta^{k+1} \gets \beta^{k}$ \\
			$\beta^{k+1}_{i} \gets \argmin_{\beta_i \in \mathbb{R}} F(\beta_1^k,\dots, \beta_{i}, \dots, \beta_{p}^k )$ using (\ref{eq:thresholding}) \tcp*{Non-spacer Step}
			$k \gets k+1$ \\
			Count[Supp($\beta^{k}$)] $\gets$ Count[Supp($\beta^{k}$)] + 1 \\
			\If{Count[Supp($\beta^{k}$)] = $C p$}
				{
					$\beta^{k+1} \gets \text{SpacerStep}(\beta^{k})$ \tcp*{Spacer Step}
					Count[Supp($\beta^{k}$)] = 0 \\
					$k \gets k+1$
				}		
		}
	}
	\caption{Coordinate Descent with Spacer Steps (CDSS)}
	\label{alg:CD}
\end{algorithm}
\setcounter{algocf}{0}
\begin{algorithm}[!htbp]
	\SetAlgorithmName{Subroutine}{Subroutine}{Subroutine}
	\DontPrintSemicolon
	\SetKwInOut{Input}{Input}\SetKwInOut{Output}{Output}
	\Input{$\beta$}
					\For{$i$ in Supp$(\beta)$}
						{
							$\beta_{i} \gets \argmin_{\beta_i \in \mathbb{R}} f(\beta_1,\dots,  \beta_{i}, \dots, \beta_{p} )$ using (\ref{eq:thresholding}) with $\lambda_0 = 0$ \\	
						}
	\Return ($\beta$)
	\caption{SpacerStep($\beta$)}
	\label{alg:spacerstep}
\end{algorithm}

\pparagraph{Why Cyclic CD?} Cyclic CD has been practically shown to be among the fastest algorithms for Lasso~\cite{glmnet} and continuous non-convex regularizers
(e.g., MCP, SCAD, etc)~\cite{sparsenet, ncvreg}. Coordinate updates in cyclic CD have low cost and exploit sparsity via sparse residual updates and active set convergence \cite{glmnet}. This makes it well-suited for high-dimensional problems with $p \gg n$ and $p$ of the order of tens-of-thousands to millions. 
On the other hand, methods requiring evaluation of the full gradient (e.g., proximal gradient descent, greedy coordinate descent, etc) can have difficulty in scaling with $p$ \cite{nesterov2012efficiency}. 
For example, proximal gradient descent methods do not exploit sparsity-based structure as well as CD-based methods~\cite{glmnet,nesterov2012efficiency}.
We also note that based on our experiments, random CD (proposed by \cite{randomCDL0} for a problem similar to ours) exhibits slower convergence in practice (see Section \ref{section:cdcomparison})---see also related discussions in~\cite{BeckConvergence} for convex problems. We have also observed empirically that cyclic CD has an edge over competing algorithms, both in terms of optimization objective (see Section~\ref{section:cdcomparison}) and statistical performance (see Sections~\ref{section:PhaseTransitionsN}, \ref{section:PhaseTransitionsSNR}, and \ref{section:high-dim-exp}).

\subsubsection{Convergence Analysis}
We analyze the convergence behavior of Algorithm~1---in particular, we prove a new result establishing convergence to a CW minimum (the limit point depends upon the initialization). Moreover, we show that the linear rate of convergence of CD which holds for minimization of a smooth convex function~\cite{BeckConvergence} extends to our non-convex problem (in an asymptotic sense). We note that if we avoid full minimization and use a conservative step size, the proofs of convergence become straightforward by virtue of a sufficient decrease of the objective value after every coordinate update\footnote{This observation also appears in establishing convergence of IHT-type algorithms---see for example~\cite{BeckSparsityConstrained, PenalizedIHT,bertsimas2015best}.}. 
However, using CD with a conservative step size for Problem~\ref{problem:main} can have a detrimental effect on the solution quality. By examining the fixed points, it can be shown that a conservative step size leads to a class of stationary solutions that contains CW minima. 
Cyclic CD has been studied in earlier work with non-convex but continuous regularizers~\cite{sparsenet,ncvreg,Tseng01} with a least-squares data fidelity term. However, to our knowledge, a convergence analysis of cyclic CD for Problem~\eqref{problem:main} is novel. 

We present a few lemmas describing the behavior of Algorithm \ref{alg:CD} (some additional technical lemmas are in the supplementary).
Theorem \ref{theorem:convergence} establishes convergence and Theorem \ref{theorem:convergencerate} presents an asymptotic linear rate of convergence 
of Algorithm \ref{alg:CD}.

The following lemma states that Algorithm~\ref{alg:CD} is a descent algorithm. 
\begin{lemma} \label{lemma:descent}
Algorithm \ref{alg:CD} is a descent algorithm and $F(\beta^k) \downarrow F^{*}$ for some $F^* \geq 0$.
\end{lemma}

For the remainder of the section, we will make the following minor assumptions to establish the convergence of Algorithm~\ref{alg:CD} for the $(L_0)$ and $(L_0 L_1)$ problems. These assumptions are not needed for problem $(L_0 L_2)$.

\begin{assumption}
	\label{assumption:Xlinindp}
	Let $m = \min \{n,p \}$. Every set of $m$ columns of $X$ are linearly independent.
\end{assumption}
\begin{assumption}\label{assumption:initial}
(Initialization) If $p > n$, we assume 
that the initial estimate $\beta^{0}$ satisfies 
\begin{itemize}
\item In the $(L_0)$ problem: $F(\beta^{0}) \leq \lambda_0 n$.
\item  In the $(L_0 L_1)$ problem: $F(\beta^{0}) \leq f(\beta^{\ell_1}) + \lambda n$ where
	$
	f(\beta^{\ell_1}) = \min_{\beta} \frac{1}{2} \| y - X \beta \|^{2}  + \lambda_1 \|\beta\|_1.
	$
\end{itemize}
\end{assumption}
The following remark demonstrates that Assumption~2 is rather minor.
\smallskip
\begin{remark}
Suppose $p > n$ and Assumption 1 holds. For the $(L_0)$ problem, let $S \subseteq [p]$ such that $|S|=n$. If $\beta^0$ is defined such that $\beta^0_S$ is the least squares solution on the support $S$ with $\beta^0_{S^c} = 0$; 
then $F(\beta^0)= \lambda_0 n$ (since the least squares loss is zero).
This satisfies Assumption~2. For the $(L_0 L_1)$ problem, we note that there always exists an optimal lasso solution $\hat{\beta}$ such that $\|\hat{\beta} \|_0 \leq n$ (e.g., see \cite{LassoUniqueness}). Therefore, $\hat{\beta}$ satisfies Assumption~2.
\end{remark}

In what follows, we assume that Assumptions~\ref{assumption:Xlinindp} and~2 hold for the $(L_0)$ and $(L_0 L_1)$ problems. 
Lemma~\ref{lemma:bdsuppsize} shows that in the $(L_0)$ and $(L_0 L_1)$ problems, the support size of any $\beta^k$ obtained by Algorithm \ref{alg:CD} cannot exceed $\min\{n,p\}.$
\smallskip
\begin{lemma}
\label{lemma:bdsuppsize}

For the $(L_0)$ and $(L_0 L_1)$ problems, $\{\beta^k\}$ satisfies $\|\beta^k\|_0 \leq \min \{n,p \}$ for all $k$.
\end{lemma}

The following theorem establishes the convergence of Algorithm~1. 
\smallskip
\begin{theorem}
\label{theorem:convergence}

The following holds for Algorithm~1:
\begin{enumerate}
    \item The support of $\{\beta^{k} \}$ stabilizes after a finite number of iterations, i.e., there exists an integer $m$ and a support $S$ such that $\text{Supp}(\beta^{k}) = S$ for all $k \geq m$.
    \item The sequence $\{\beta^{k} \}$ converges to a CW minimum $B$ with $\text{Supp}(B) = S$.
\end{enumerate}
\end{theorem}

Theorem~\ref{theorem:convergencerate} presents an asymptotic linear rate of convergence for Algorithm~1: we make use of part 1 of Theorem~\ref{theorem:convergence} and the convergence rate of cyclic CD (for smooth and strongly convex functions)~\cite{BeckConvergence}.
Note that in Theorem~\ref{theorem:convergencerate}, \textsl{full cycle} refers to a single pass of vanilla CD over all the coordinates in $S$, 
and $\beta^K$ refers to the iterate generated after performing $K$ full cycles of CD. 

\begin{theorem}
\label{theorem:convergencerate} [Adaptation of~\cite{BeckConvergence}, Theorem 3.9]
Let $\{ \beta^K \}$ be the full-cycle iterates generated by Algorithm~\ref{alg:CD} and $B$ be the limit with support $S$. Let $m_S$ and $M_S$ denote the smallest and largest eigenvalues of $X_S^T X_S$, respectively. Then, there is an integer $N$ such that for all $K \geq N$ the following holds:
\begin{align}
F(\beta^{K+1}) - F(B) \leq \left(1 - \frac{m_S + 2\lambda_2}{2(1+2\lambda_2) \left( 1 + {|S| (\frac{M_S+2\lambda_2}{1+2\lambda_2}})^{2}  \right)} \right) \left(F(\beta^{K}) - F(B)\right).
\end{align}
\end{theorem}

\subsection{Local Combinatorial Optimization Algorithms}\label{sec:local-algos}
We present algorithms to obtain solutions belonging to the classes of Swap Inescapable minima introduced in Section \ref{section:swapinescapable}.

\subsubsection{Algorithms for PSI minima} \label{section:psialgorithm}
We introduce an iterative algorithm that leads to a PSI$(k)$ minimum. In the $\ell$th iteration, the algorithm performs two steps: (1) runs Algorithm \ref{alg:CD} to get a CW minimum $\beta^{\ell}$, and (2) searches for a ``descent move'' by solving the following combinatorial optimization problem:
\begin{equation}
	\label{eq:psi}
%	\min_{ \substack{\beta,\ S_1 \in S,\ S_2 \in S^{c} \\ |S_1| \leq k,\ |S_2| \leq k}} F(\beta^{l} - U^{S_1} \beta^{l} + U^{S_2} \beta).
\min_{\beta, S_1, S_2} ~~ F(\beta^{\ell} - U^{S_1} \beta^{\ell} + U^{S_2} \beta) ~~~ \sbt~~~~ S_1 \subseteq S,\ S_2 \subseteq S^{c}, |S_1| \leq k,\ |S_2| \leq k,
\end{equation}
where $S = \text{Supp}(\beta^{\ell})$. Note that if there is a feasible solution $\hat{\beta}$ to Problem~\eqref{eq:psi} satisfying $F(\hat{\beta}) < F(\beta^{\ell})$, then $\hat{\beta}$ may not be a CW minimum. In this case, Algorithm~1 can be initialized with $\hat{\beta}$ to obtain a better solution for Problem~\eqref{problem:main}.
Otherwise, if such a $\hat{\beta}$ does not exist, then $\beta^{\ell}$ is a PSI$(k)$ minimum (by Definition \ref{def:psi}). Algorithm~\ref{alg:mipcd} (aka CD-PSI$(k)$) summarizes the  algorithm.
\begin{comment}
\begin{algorithm}[htbp]
	\DontPrintSemicolon
	%\SetAlgoLined
	\SetKwInOut{Input}{Input}\SetKwInOut{Output}{Output}
	\Input{CW minimum $\beta^{0}$}
	$l \gets 0$\\
	\Repeat
		{	
			\If {(\ref{eq:psi}) initialized with $\beta^{l}$ has a feasible solution $\hat{\beta}^{l}$ satisfying $F(\hat{\beta}^l) < F(\beta^l)$ }{ 
			$\beta^{l+1} \gets$ Output of Algorithm \ref{alg:CD} initialized with $\hat{\beta}^l$}
			\Else{\textbf{Stop}}
			$l \gets l + 1$ \\
		}
	\Return{$\beta^{l}$}
	\caption{CD-PSI($k$)}
	\label{alg:mipcd}
\end{algorithm}
\end{comment}
\begin{algorithm}[htbp]
	\DontPrintSemicolon
	\SetKwInOut{Input}{Input}\SetKwInOut{Output}{Output}
	$\hat{\beta}^0 \gets \beta^{0}$ \\
	\For{$\ell = 0, 1, \dots$}
		{	$\beta^{\ell+1} \gets$ Output of Algorithm \ref{alg:CD} initialized with $\hat{\beta}^{\ell}$  \\
			\If {Problem~(\ref{eq:psi}) has a feasible solution $\hat{\beta}$ satisfying $F(\hat{\beta}) < F(\beta^{\ell+1})$ }{$\hat{\beta}^{\ell+1} \gets \hat{\beta}$}
			\Else{\textbf{Terminate}}
		}
	\caption{CD-PSI($k$)}
	\label{alg:mipcd}
\end{algorithm}

\smallskip

\begin{theorem} \label{theorem:cd-kswaps}
	Let $\{ \beta^{\ell} \}$ be the sequence of iterates generated by Algorithm \ref{alg:mipcd}. 
    For the $(L_0)$ and $(L_0 L_1)$ problems, suppose that Assumptions \ref{assumption:Xlinindp} and~2 hold. Then, Algorithm \ref{alg:mipcd} terminates in a finite number of iterations and the output is a PSI$(k)$ minimum.
\end{theorem}
As indicated in Theorem \ref{theorem:cd-kswaps}, Algorithm~\ref{alg:mipcd} terminates in a finite number of iterations (which depends upon the number of swaps that improve the objective value).
In our experiments, Algorithm~\ref{alg:mipcd} typically terminates in less than 20 iterations. Each iteration involves a call to Algorithm 1 (which is quite fast) followed by solving a feasibility problem (i.e., finding $\hat{\beta}$ as described in Algorithm~\ref{alg:mipcd}). As we discuss next, for moderate $p$ (e.g., $10^3$ to $10^4$), this feasibility problem can be handled efficiently using MIO solvers. When $k=1$, we propose a specialized algorithm for solving this feasibility problem that can easily scale to problems with $p=10^6$.

\pparagraph{MIO formulation for Problem~\eqref{eq:psi}:}
Problem~(\ref{eq:psi}) admits an MIO formulation given by:
\begin{subequations}
\label{mio:psi}
\begin{align}
   	 \min_{\theta, \beta, z}~~~~& f(\theta) + \lambda_0 \sum\limits_{i \in [p]} z_i \label{eq:psiF}\\
\sbt~~~    & \theta = \beta^{\ell} - \sum_{i \in S} e_i \beta^{\ell}_i (1 - z_i) + \sum_{i \in S^{C}} e_i \beta_i \label{eq:psitheta} \\
    & -{\mathcal M} z_i \leq \beta_i \leq {\mathcal M} z_i, \ \  \ \forall i \in S^{c} \label{eq:psibetabd}\\
    & \sum_{i \in S^{c}} z_i \leq k \label{eq:psilessk}\\
    & \sum_{i \in S} z_i \geq  |S| - k \label{eq:psigreaterk}\\
    & \beta_i \in \mathbb{R}, \ \  \ \forall i \in S^{c} \label{eq:psibeta}\\
    & z_i \in \{0,1\}, \ \  \ \forall i \in [p], \label{eq:psiz}
\end{align}
\end{subequations}
where the optimization variables are $\theta \in \mathbb{R}^{p}$,
$\beta_{i}, i \in S^c$, and $z \in \{0, 1\}^{p}$.
In formulation~\eqref{mio:psi}, $S = \text{Supp}(\beta^{\ell})$, where $\beta^\ell$ is fixed, and ${\mathcal M}$ is a Big-M parameter (a-priori specified) controlling the 
$L_{\infty}$-norm of $\beta_{S^c}$. Any sufficiently large value of $\mathcal{ M }$ will lead to a solution for Problem~\eqref{eq:psi}; however, a tight choice for $\mathcal{ M }$ affects the run time of the MIO solver---see~\cite{bertsimas2015best} for additional details.
We note that the $\|\theta\|_1$ term included in $f(\theta)$ 
can be expressed via linear inequalities using auxiliary variables. Thus, Problem~\eqref{mio:psi} is a mixed integer quadratic optimization (MIQO) problem.

We now explain the constraints in Problem~\eqref{mio:psi} and how they relate to Problem~\eqref{eq:psi}.
To this end, let $S_1$ and $S_2$ be subsets defined in \eqref{eq:psi}. 
Let $\theta = \beta^{\ell} - U^{S_1} \beta^{\ell} + U^{S_2} \beta$ and this relation is expressed in~\eqref{eq:psitheta}.
Let us consider any binary variable $z_{i}$ where $i \in S$. 
If $z_{i} = 0$ then $\beta^{\ell}_{i}$ is removed from $S$, and we have $\theta_{i} = 0$ (see~\eqref{eq:psitheta}).
If $z_{i}=1$, then $\beta^{\ell}_{i}$ is not removed from $\theta$, and we have
$\theta_{i} =\beta^{\ell}_{i} \neq 0$ (see~\eqref{eq:psitheta}).
Note that $|S_{1}| = \sum_{i \in S} (1-z_{i}) = |S| - \sum_{i \in S} z_{i}.$ The condition $|S_{1}| \leq k$, is thus encoded in the constraint $\sum_{i \in S} z_{i} \geq |S| - k$ in~\eqref{eq:psigreaterk}.
Thus we have that $ \| \theta_{S} \|_{0} = \sum_{i \in S} z_{i}$.

Now consider any binary variable $z_{i}$ where $i \in S^c$. If $z_{i}=1$, then by~\eqref{eq:psibetabd} we observe that $\beta_{i}$ is free to vary in $[-\mathcal M, \mathcal M]$.
This implies that $\theta_{i} = \beta_{i}$. If $z_{i} =0$ then $\theta_{i}=\beta_{i}=0$. Note $\sum_{i \in S^c} z_{i} = |S_{2}|$, and the constraint $|S_{2}| \leq k$ is expressed via $\sum_{i \in S^c} z_{i} \leq k$ in~\eqref{eq:psilessk}.
It also follows that $ \| \theta_{S^c} \|_{0} = \sum_{i \in S^c} z_{i}$.
Finally, we note that the function appearing in the objective~\eqref{eq:psiF} is $F(\theta)$, since $\lambda_0 \sum_{i \in [p]} z_{i} = \lambda_0 \| \theta \|_0$. 

\smallskip

\begin{remark} \label{remark:psi}
Problem~\eqref{mio:psi} has a smaller (combinatorial) search space compared to an MIO formulation for the full Problem~\eqref{problem:main}---solving~\eqref{mio:psi} for small values of $k$ is usually much faster than Problem~\eqref{problem:main}.
Furthermore, an MIO framework can quickly deliver a feasible solution to Problem~\eqref{mio:psi} with a smaller objective than the current solution---this is usually much faster than establishing optimality via dual bounds. Note that the MIO-framework can also certify (via dual bounds) if there is no feasible solution with a strictly smaller objective value.
\end{remark}

Section~\ref{section:experiments} presents examples where Problem~\eqref{mio:psi} leads to higher quality solutions---from both the optimization and statistical performance viewpoints. We now present an efficient algorithm for solving the special case of Problem~\eqref{eq:psi} with $k=1$.

\pparagraph{An efficient algorithm for computing PSI($1$) minima.}
Subroutine 2 presents an algorithm for Problem~\eqref{eq:psi} with $k=1$. That is, we search for a feasible solution $\hat{\beta}$ of Problem~(\ref{eq:psi}) satisfying $F(\hat{\beta}) < F(\beta^{\ell})$.
\setcounter{algocf}{2}
\begin{algorithm}[!htbp]
	\DontPrintSemicolon
	%\SetAlgoLined
	\SetAlgorithmName{Subroutine}{Subroutine}{Subroutine}
	\SetKwInOut{Input}{Input}\SetKwInOut{Output}{Output}
	%\Input{CW minimum $\beta$}
	$S \gets \text{Supp}(\beta^{\ell})$ \\
	\For{$i \in S$}
	{
		\For{$j \in S^c$}
		{	\vspace{-0.5cm}
		\begin{flalign}
			& v^{*}_j \gets \argmin\limits_{v_j \in \mathbb{R}} F(\beta^{\ell} - e_i  \beta^{\ell}_i + e_j v_j) && \label{code::argminF} \\
			& F^{*}_j \gets F(\beta^{\ell} - e_i  \beta^{\ell}_i + e_j v^{*}_j) && \label{code:fstarj}
		\end{flalign}
		}
	\vspace{-0.8cm}
	\begin{flalign} \vartheta \gets \argmin\limits_{j \in S^c} F^{*}_j && \label{code:kgets} \end{flalign}  \\
	\vspace{-0.4cm}
	\If{$F^{*}_\vartheta < F(\beta^{\ell})$}
		{\vspace{-0.5cm} 
		\begin{flalign} 
			& \hat{\beta} \gets \beta^{\ell} - e_i  \beta^{\ell}_i + e_\vartheta v^{*}_\vartheta && \label{code:betaupated}  \\
			& \textbf{Terminate} \nonumber && 
		\end{flalign} \vspace{-0.8cm} 
		}
	}
    %\Return ($\beta$)
	\caption*{Vanilla Implementation of Problem~\eqref{eq:psi} with $k=1$.}
	\label{subroutine:findswap}
\end{algorithm}
The two for loops in Subroutine 2 can run for a total of $(p - \|\beta^{\ell}\|_0) \|\beta^{\ell} \|_0$ iterations, where every iteration requires $O(n)$ operations to perform the minimization in (\ref{code::argminF}) and evaluate the new objective in (\ref{code:fstarj}). Therefore, Subroutine 2 entails an overall cost of $O \Big( n (p - \|\beta^{\ell} \|_0) \|\beta^{\ell} \|_0 \Big)$. However, we show below that a careful implementation can reduce the cost by a factor of $n$; leading to a cost of $O \Big((p - \|\beta^{\ell} \|_0) \|\beta^{\ell} \|_0 \Big)$-many operations.

A solution $v^{*}_j$ of Problem (\ref{code::argminF}) is given by
\begin{align} \label{eq:vstarj}
v^{*}_j = 	\begin{cases}
\sign(\bar{\beta}_j) \frac{|\bar{\beta}_j| - \lambda_1}{1+2\lambda_2} & \text{ if~~~ } \frac{|\bar{\beta}_j| - \lambda_1}{1+2\lambda_2} \geq \sqrt{2\lambda_0 \over 1+2 \lambda_{2}} \\
0 & ~\text{otherwise}
%0 & \text{ if~~~ } \frac{|\bar{\beta}_j| - \lambda_1}{1+2\lambda_2} < \sqrt{2\lambda_0 \over 1+2 \lambda_{2}} \\
\end{cases}
\end{align}
where
\begin{equation}\label{resid-1-1}
\bar{\beta}_j  = \langle r + X_i \beta^{\ell}_i , X_j\rangle = \langle r , X_j\rangle + \langle X_i , X_j\rangle \beta^{\ell}_i,
\end{equation}
and $r = y - X \beta^{\ell}$.  We note that in 
Algorithm~\ref{alg:mipcd}, solving (\ref{eq:psi}) is directly preceded by a call to Algorithm \ref{alg:CD}. The quantities $\langle r , X_j\rangle$ and $ \langle X_i , X_j \rangle$ appearing on the right-hand side of~\eqref{resid-1-1} can be stored during the call to Algorithm~1 (these two quantities are computed by CD as part of the `covariance updates'---see Section \ref{section:RegPath} for details).
By reusing these two stored quantities, we can compute every $\bar{\beta}_j$ (and consequently $v^{*}_j$) in $O(1)$ arithmetic operations.

Furthermore, the following equivalent representations hold:
\begin{align}
& \argmin_{j \in S^c} F^{*}_j \iff \argmax_{j \in S^c} |v^{*}_j| \\
& F^{*}_\vartheta < F(\beta^{l}) \iff |v^{*}_\vartheta| > |\beta^{l}_i|.
\end{align}
Thus, we can avoid the computation of the objective $F^{*}_j$ in (\ref{code:fstarj}) and replace (\ref{code:kgets}) with $\vartheta \gets \argmax_{j \in S^c} |v^{*}_j|$. Furthermore, we can replace $F^{*}_\vartheta < F(\beta^{\ell})$ (before equation (\ref{code:betaupated})) with $|v^{*}_\vartheta| > |\beta^{l}_i|$. We summarize these changes in Subroutine~3, which is the efficient counterpart of Subroutine 2.
\setcounter{algocf}{3}
\begin{algorithm}[!htbp]
	\DontPrintSemicolon
	%\SetAlgoLined
	\SetAlgorithmName{Subroutine}{Subroutine}{Subroutine}
	\SetKwInOut{Input}{Input}\SetKwInOut{Output}{Output}
	%\Input{CW minimum $\beta$}
	$S \gets \text{Supp}(\beta^{\ell})$ \\
	\For{$i \in S$}
	{
		\For{$j \in S^c$}
		{	%\vspace{-0.5cm}
        Compute $v^{*}_j$ in $O(1)$ using \eqref{eq:vstarj} \\
        $ \ \ $
		}
	\vspace{-0.8cm}
	\begin{flalign*} \vartheta \gets \argmax_{j \in S^c} |v^{*}_j| && \end{flalign*}  \\
	\vspace{-0.4cm}
	\If{$|v^{*}_\vartheta| > |\beta^{\ell}_i|$}
		{\vspace{-0.5cm} 
		\begin{flalign*} 
			& \hat{\beta} \gets \beta^{\ell} - e_i  \beta^{\ell}_i + e_\vartheta v^{*}_\vartheta &&   \\
			& \textbf{Terminate} \nonumber && 
		\end{flalign*} \vspace{-0.8cm} 
		}
	}
    %\Return ($\beta$)
	\caption*{Efficient Implementation of Problem~\eqref{eq:psi} with $k=1$.}
	\label{subroutine:findswapefficient}
\end{algorithm}
Note that Subroutine~3 has a cost of $O \Big((p - \|\beta^{l}\|_0) \|\beta^{l}\|_0 \Big)$ operations.

\smallskip

\begin{remark}
Since CD-PSI($1$) (Algorithm~\ref{alg:mipcd} with $k=1$) is computationally efficient, in Algorithm~\ref{alg:mipcd} (with $k>1$),
CD-PSI($1$) may be used to replace Algorithm~\ref{alg:CD}. 
In our numerical experiments, this is found to work well in terms of lower run times and also in obtaining higher-quality solutions (in terms of objective values). This modification also guarantees convergence to a PSI($k$) minimum (as the proof of Theorem \ref{theorem:cd-kswaps} still applies to this modified version).
\end{remark}

\subsubsection{Algorithm for FSI minima} \label{section:fsialgorithm}
To obtain a FSI($k$) minimum, 
Problem~(\ref{eq:psi}) 
needs to be modified---we replace optimization w.r.t. the variable $U^{S_2} \beta$ by that of $U^{(S \setminus S_1) \cup S_2} \beta$. This leads to the following problem:
\begin{equation}	\label{eq:fsi}
\begin{aligned}
	\min_{\beta, S_1, S_2}~~~ F(\beta^{\ell} - U^{S_1} \beta^{\ell} + U^{(S \setminus S_1) \cup S_2} \beta) ~~~\sbt~~~~ 
    \ S_1 \subseteq S,\ S_2 \subseteq S^{c},  |S_1| \leq k,\ |S_2| \leq k,
\end{aligned}
\end{equation}
where $S = \text{Supp}(\beta^{l})$. Similarly, Algorithm~\ref{alg:mipcd} gets modified by considering Problem~\eqref{eq:fsi} instead of Problem~\eqref{eq:psi}. By the same argument used in the proof of Theorem \ref{theorem:cd-kswaps}, this modification guarantees that Algorithm \ref{alg:mipcd} converges in a finite number of iterations to a FSI($k$) minimum. 
\begin{comment}
Problem~(\ref{eq:fsi}) can be expressed as  an MIQO problem. To this end, Problem~(\ref{mio:psi}) needs to be modified in lines~\eqref{eq:psibetabd},~\eqref{eq:psitheta} 
and~\eqref{eq:psibeta}  with the following constraints:
\begin{equation*}
\begin{myarray}[1.5]{l}
 \theta = \beta^{\ell} - \sum_{i \in S} e_i \beta^{\ell}_i (1 - z_i) + \sum_{i \in S^{C}} e_i \beta_i  + \sum_{i \in S} e_i \beta_i \\
  -{\mathcal M} z_i \leq \beta_i \leq {\mathcal M} z_i,~~~~~~ \ \   i \in S \cup S^c = [p] \\
 \beta_{i} \in \Re, ~~~~~ i \in S \cup S^c = [p]
 \end{myarray}
\end{equation*}
\end{comment}

We present an MIO formulation for Problem~(\ref{eq:fsi}). We write $\theta = \beta^{\ell} - U^{S_1} \beta^{\ell} + U^{(S \setminus S_1) \cup S_2} \beta$ and use a binary variable $w_i, i \in S$ to indicate whether $i \in S_1$ or not: we set $w_{i} = 0$ iff $i \in S_{1}$. We use another binary variable 
$z_i, i \in [p]$ to indicate the number of nonzeros in $\theta$, i.e., 
$z_{i}=0 \Rightarrow \theta_i=0$. This leads to the following MIO problem:
\begin{subequations}
\label{mio:fsi}
\begin{align}
   	 \min_{\theta, w, z}~~~~& f(\theta) + \lambda_0 \sum\limits_{i \in [p]} z_i \label{eq:fsiF}\\
    & -{\mathcal M} z_i \leq \theta_i \leq {\mathcal M} z_i, \ \  \ \forall i \in [p] \label{eq:fsibetabd}\\
    & z_i \leq w_i,  \ \  \ \forall i \in S \label{eq:fsiwlz}\\
    & \sum_{i \in S^{c}} z_i \leq k \label{eq:fsilessk}\\
    & \sum_{i \in S} w_i \geq  |S| - k \label{eq:fsigreaterk}\\
    & \theta_i \in \mathbb{R}, \ \  \ \forall i \in [p] \label{eq:fsibeta}\\
    & z_i \in \{0,1\}, \ \  \ \forall i \in [p] \label{eq:fsiz} \\
    & w_i \in \{0,1\} \ \  \ \forall i \in S. \label{eq:fsiw}
\end{align}
\end{subequations}
In~\eqref{eq:fsibetabd}, for every $i \in [p]$, the binary variable $z_i=1$ if $i \in \text{Supp}(\beta)$, and $\mathcal{M}$ is a sufficiently large constant (similar to that in \eqref{mio:psi}).
The second term in the objective~\eqref{eq:fsiF} stands for $\lambda_0\sum_{i} z_{i}=\lambda_0 \| \theta \|_0$. In \eqref{eq:fsiwlz} we enforce the condition that if $w_{i}=0$ then $z_{i}=0$ implying that $\theta_{i}=0$; If $w_{i}=1$ then 
$i \notin S_{1}$. Coordinates in $S \setminus S_1$ (i.e., those with $w_i = 1$) are free to be inside or outside of $\text{Supp}(\theta)$. 
We enforce $|S_1| \leq k$ and $|S_2| \leq k$ via \eqref{eq:fsigreaterk} and \eqref{eq:fsilessk}, respectively. 
\smallskip
\begin{remark}
Formulation~\eqref{mio:fsi} has a larger search space compared to formulation \eqref{mio:psi} of PSI minima, due to the additional number of continuous variables. While this leads to 
increased run times compared to Problem~\eqref{mio:psi}, it can still be solved faster than an MIO formulation for the full Problem~\eqref{problem:main} (for the same reasons as in Remark \ref{remark:psi}).
\end{remark}
In Section~\ref{swap-inescp-minima}, we present experiments where we compare the quality of FSI($k$) minima, for different values of $k$, to the other classes of minima.
\section{Efficient Computation of the Regularization Path}
\label{section:RegPath}
We designed \texttt{L0Learn}\footnote{Available on CRAN at \url{https://CRAN.R-project.org/package=L0Learn} and on github at \url{https://github.com/hazimehh/L0Learn}}: an extensible C++ toolkit with an R interface that implements most of the algorithms discussed in this paper. 
{Our toolkit achieves lower running times\footnote{Problem~\eqref{problem:main} usually leads to solutions with fewer nonzeros compared to Lasso and MCP penalized regression. This also contributes to reduced run times.} compared to other popular sparse learning toolkits (e.g., \texttt{glmnet} and \texttt{ncvreg}) by utilizing a series of computational tricks and heuristics. These include an adaptive grid of tuning parameters, continuation, active-set updates, greedy cyclic ordering of coordinates, correlation screening, and a careful accounting of floating point operations---some of these heuristics (as specified below) appear in prior work for deriving highly efficient algorithms for the Lasso (e.g., \texttt{glmnet})}. 
Below, we provide a detailed account of the aforementioned strategies that are also found to result in high quality solutions.

\pparagraph{Adaptive Selection of Tuning Parameters:} We use continuation on a grid of $\lambda_0$ values: $\lambda_{0}^{1} > \lambda_{0}^{2} > \dots > \lambda_{0}^{m}$ and use the solution obtained from $\lambda_0^{k}$ as a warm start for $\lambda_0^{k+1}$. The choice of $\lambda_0^{i}$'s requires care so we present a new method to select this sequence. If two successive values of the $\lambda_0$ sequence are far apart, one might miss good solutions. However, if these successive values are too close, the corresponding solutions will be identical.
To avoid this problem, we derive conditions on the choice of $\lambda_0$ values which ensure that Algorithm~1 leads to different solutions. To this end, we present the following lemma, wherein we assume that $\lambda_1,\lambda_2$ are a-priori fixed.
\smallskip
\begin{lemma}
	\label{lemma:nextlambda}
Suppose $\beta^{(i)}$ is the output of Algorithm~1 with $\lambda_0 = \lambda_0^i$. Let $S = \text{Supp}(\beta^{(i)})$, $ r = y -  X \beta^{(i)}$ denote the residual, and
  \begin{align} \label{eq:Mi}
   M^{i} =  \frac{1}{2(1+2\lambda_2)} \max_{j \in S^c} \Big ( |\langle r , X_j \rangle|  - \lambda_1 \Big)^2.
  \end{align}
  Then, running Algorithm \ref{alg:CD} for $\lambda_0^{i+1}<\lambda_0^{i}$ initialized at $\beta^{(i)}$ leads to a solution $\beta^{(i+1)}$ 
  satisfying: $\beta^{(i+1)} \neq \beta^{(i)}$ if 
  $\lambda^{i+1}_0 < M^{i}$, and $\beta^{(i+1)} = \beta^{(i)}$ if $\lambda^{i+1}_0 \in (M^{i} , \lambda^i_0 ]$.
\end{lemma}

Lemma~\ref{lemma:nextlambda} suggests a simple scheme to 
compute a sequence $\{\lambda_{0}^{i}\}$ that avoids duplicate solutions. Suppose we have computed the regularization path up to $\lambda_0 = \lambda_0^{i}$, then
$\lambda_0^{i+1}$ can be computed as 
$\lambda_0^{i+1} = \alpha M^i$, where $\alpha$ is a fixed scalar in $(0,1)$. Moreover, we note that $M^i$ (defined in (\ref{eq:Mi})) can be computed without explicitly calculating $\langle r , X_i \rangle$ for every $i \in S^c$, as these dot products can be maintained in memory while running Algorithm \ref{alg:CD} with $\lambda_0=\lambda_0^{i}$. Therefore, computing $M^i$, and consequently $\lambda_0^{i+1}$, requires only $O(|S^c|)$ operations.

\pparagraph{(Partially) Greedy Cyclic Order:} Suppose Algorithm \ref{alg:CD} is initialized with a solution $\beta^{0}$ and let $r^0 = y - X\beta^{0}$. Before running Algorithm \ref{alg:CD}, we reorder the coordinates based on sorting the quantities $|\langle r^0 , X_i \rangle|$ for $i \in [p]$ in descending order\footnote{Since the columns of $X$ have unit $L_2$-norm, updating index $\argmax_{i} |\langle r^0 , X_i \rangle|$ will lead to the maximal decrease in the objective function.}. {In practice, we perform \textsl{partial sorting}, in which only the top $t$ (e.g., $t=5000$ and $p$ is much larger) coordinates are sorted, while the rest maintain their initial order. This is typically faster and equally effective compared to sorting all coordinates.} 
Note that this ordering is performed \emph{once} before the start of Algorithm~1---this is different from greedy CD that finds the maximal correlation at every coordinate update. 
Our experiments in Section \ref{section:cdcomparison} indicate that (partially) greedy cyclic order performs significantly better than a vanilla cyclic order or random order. 

\pparagraph{Correlation Screening:}
When using continuation, we perform a screening method inspired by ~\cite{tibshirani2012strong}\footnote{Our approach differs from~\cite{tibshirani2012strong} who derive screening rules for convex problems.}. We restrict the updates of Algorithm \ref{alg:CD} to the support of the warm start in addition to a small portion (e.g., top $10^3$) of other coordinates that are highly correlated with the current residuals---these coordinates are readily available as a byproduct of the (partially) greedy cyclic ordering rule, described above. After convergence on the screened support, we check if any coordinate from outside the support violates the conditions of a CW minimum and rerun the algorithm if needed. Typically, the solution obtained from the screened set turns out to be a CW minimum, and only one pass is done over all the coordinates. 

\pparagraph{Active Set Updates:} 
As in prior work~\cite{glmnet}, active set methods are found to be very useful in our context as well.
Empirically, the iterates generated by Algorithm \ref{alg:CD} can typically achieve support stabilization in less than $10$ full cycles\footnote{Recall, one full cycle refers to updating all the $p$ coordinates in a cyclic order.}. This is further supported by Theorem \ref{theorem:convergence}, which establishes finite-time stabilization of the support. 
If the support does not change across multiple consecutive full cycles, we restrict the updates of Algorithm \ref{alg:CD} to the current support. After convergence on this support, we check whether any coordinate from outside the support violates the conditions of a CW minimum and rerun the algorithm if needed.

\pparagraph{Fast Coordinate Updates:} Following~\cite{glmnet}, we present techniques for efficiently computing the coordinate updates---these are also found to be useful for our implementation of PSI($1$).

Let $\beta^{k}$ be the current iterate in Algorithm~1 and $r^k$ be the residuals. To compute $\widetilde{\beta_i} = \langle   r^{k} , X_i \rangle + \beta^{k}_i$, one can use one of the following rules that exploit sparsity \cite{glmnet}: 
(i)~\emph{Residual Updates:} 
We maintain the residuals $r^{k}$ throughout the algorithm and compute $\widetilde{\beta}_{i}$ 
using $O(n)$ operations.
Once $\beta^{k+1}_i$ is computed,  we update 
$r^{k+1}$ 
with cost $O(n)$ operations. Since $\beta^k$ is sparse, many of the coordinates remain at zero during the algorithm implying that $r^{k+1} = r^{k}$.
(ii)~\emph{Covariance Updates:} This appears in~\cite{glmnet} for updating $\widetilde{\beta_{i}}$ without using the precomputed $r^{k}$.
Note that Algorithm~1 precomputes $\langle y , X_i \rangle$ for all $i \in [p]$. 
If a coordinate $\ell$ enters the support for the first time, 
we compute and store the covariance terms: $\langle X_\ell, X_j \rangle$ for all $j \in [p]$ with cost $O(np)$. In iteration $k+1$, we compute $\widetilde{\beta}_i$ 
using these covariances and exploiting the sparsity of $\beta^k$---this costs $O(\|\beta^k\|_0)$ operations. This costs less than computing $\widetilde{\beta_i}$ using rule (i) if $\|\beta^k\|_0 < n$.

Scheme (ii) is useful when the supports encountered by Algorithm~1 are small w.r.t. $n$. It is also useful for an efficient implementation of the PSI$(1)$ algorithm as it stores the dot products required in (\ref{resid-1-1}). However, when the supports encountered by CD are relatively large compared to $n$ (e.g., 10\% of $n$), then Scheme (i) can become significantly faster since the dot product computations can be accelerated using calls to the Basic Linear Algebra Subprograms (BLAS).

\section{Computational Experiments} \label{section:experiments}
In this section, we investigate both the optimization and statistical performances of our proposed algorithms and compare them to other popular sparse learning algorithms.
For convenience, we provide a road map of this section.  Section~\ref{section:cdcomparison} compares the optimization performance of our proposed algorithms and other variants of CD and IHT. Section~\ref{section:PhaseTransitionsN} empirically studies the statistical properties of estimators available from our proposed algorithms versus others for varying sample sizes. Section~\ref{section:PhaseTransitionsSNR} provides a similar study for varying SNR. Section~\ref{swap-inescp-minima} performs an in-depth investigation among the PSI$(k)$/FSI$(k)$ algorithms, for different values of $k$. Section~\ref{section:high-dim-exp} presents timing and statistical performance comparisons on some large-scale instances, including real datasets. Additional experiments are placed in the supplementary.

\subsection{Experimental Setup} \label{section:expsetup}

\pparagraph{Data Generation:} We consider a series of experiments on synthetic datasets for a wide range of problem sizes and designs. We generate a multivariate Gaussian 
data matrix $X_{n \times p} \sim \text{MVN}(0, \Sigma).$ 
We use a sparse coefficient vector $\beta^{\dagger}$ with $k^{\dagger}$ equi-spaced nonzero entries, each set to 1. We then generate the response vector $y = X \beta^{\dagger} + \epsilon$, where $\epsilon_i \stackrel{\text{iid}}{\sim} \text{N}(0,\sigma^2)$ is independent of $X$. We define the signal-to-noise ratio (SNR) by
$
\text{SNR} = \tfrac{\text{Var}(X \beta^{\dagger})}{\text{Var}(\epsilon)} =  \tfrac{\beta^{\dagger T} \Sigma \beta^{\dagger}}{\sigma^2}.$
An alternative to the SNR is the ``proportion of variance explained'' or $R^2$ for the {\em true model}: $R^2 = 1 - \frac{\text{Var}(y - X \beta^{\dagger})}{\text{Var}(y)} = \frac{\text{SNR}}{\text{SNR} + 1}$.

We consider the following instances of $\Sigma:=((\sigma_{ij}))$:
\begin{itemize}
\itemsep -0.2em 
\item \textbf{Constant Correlation}: We set $\sigma_{ij} = \rho$ for every $i \neq j$ and $\sigma_{ii} = 1$ for all $i \in [p]$. 
\item \textbf{Exponential Correlation}: We set $\sigma_{ij} = \rho^{|i - j|}$ for all $i,j$,
with the convention $0^0 =1$.
\end{itemize}
We select the tuning parameters by minimizing the prediction error on a separate validation set, which is generated under the fixed design setting as $y' = X \beta^{\dagger} + \epsilon'$, where $\epsilon'_i \stackrel{\text{iid}}{\sim} \text{N}(0,\sigma^2)$. In the supplementary, we also include alternative validation strategies. 
Particularly, we include the results of the experiments in Sections \ref{section:PhaseTransitionsN} and \ref{section:PhaseTransitionsSNR} based on both oracle and random design tuning (following~\cite{onbestsubset}). The results are found to be quite similar.

\pparagraph{Competing Algorithms and Parameter Tuning:}
In addition to our proposed algorithms, we compare the following state-of-the-art methods in the experiments:
\begin{itemize}[leftmargin=*]
\itemsep 0em 
\item \textbf{Lasso}: We use our own implementation and in the figures we denote it by ``$L_1$''. 
\item \textbf{Relaxed Lasso}: We use the Relaxed Lasso version suggested in~\cite{onbestsubset}, defined as $\beta^{\text{relaxed}} = \gamma \beta^{\text{lasso}} + (1-\gamma) \beta^{\text{LS}}$, where $\beta^{\text{lasso}}$ is the Lasso estimate, 
the nonzero components of $\beta^{\text{LS}}$ are given by the least squares solution on the support of $\beta^{\text{lasso}}$, and $\gamma \in [0,1]$ is a second tuning parameter (in addition to the tuning parameter for the Lasso). We use our own implementation for relaxed lasso and denote it by ``L1Relaxed''.
\item \textbf{Elastic Net}: This uses a combination of the $L_1$ and $L_2$ regularization \cite{Zou05regularizationand}. We use the implementation of \texttt{glmnet} and refer to it as ``$L_1 L_2$''.

\item \textbf{MCP}: This is the MCP penalty of~ \cite{zhang2010nearly}. We use the implementation provided in \texttt{ncvreg} \cite{ncvreg}.
\item \textbf{Forward Stepwise}: We use the implementation of~\cite{onbestsubset}, and denote it by ``FStepwise''.

\item \textbf{IHT}: We use our implementation for IHT for the $(L_0)$ problem with a step size of $M^{-1}$, where $M$ is defined in Remark \ref{remark:IHT}. 
\end{itemize}
For all the methods involving one tuning parameter, we tune over a grid of 100 parameter values, except for forward stepwise selection which we allow to run for up to 250 steps.
For the methods with two parameters, we tune over a 2-dimensional grid of $100 \times 100$ values. For our algorithms, the tuning parameter $\lambda_0$ is generated as per Section \ref{section:RegPath}. For the $(L_0 L_2)$ problem, we sweep $\lambda_2$ between $0.1 \lambda_2^{*}$ and $10 \lambda_2^{*}$, where $\lambda_2^{*}$ is the optimal regularization parameter for ridge regression (based on validation over $100$ grid points between $0.0001$ and $1000$) . For $(L_0 L_1)$, Lasso, Relaxed Lasso, and Elastic Net, we sweep $\lambda_1$ from $\|X^T y\|_{\infty}$ down to $0.0001 \times \|X^T y\|_{\infty}$. For Relaxed Lasso and Elastic Net, we sweep their second parameter between $0$ and $1$. For MCP, the range of the first parameter $\lambda$ is chosen by \texttt{ncvreg}, and we sweep the second parameter $\gamma$ between $1.5$ and $25$.

\pparagraph{Performance Measures:}
We use several metrics to evaluate the quality of an estimator $\hat{\beta}$ (say). In addition to the objective value, number of true positives (TP), false positives (FP), and support size, we consider the following measures:
\begin{itemize}[leftmargin=*]
\itemsep 0em 
\item \textbf{Prediction Error}: This is the same measure used in \cite{bertsimas2015best} and is defined by $ \| X \hat{\beta} - X \beta^{\dagger} \|^2 / \| X \beta^{\dagger}\|^2. $
The prediction error of a perfect model is $0$ and that of the null model ($\hat{\beta} = 0$) is $1$.

\item \textbf{$L_{\infty}$ Norm}: This is the 
estimation error measured in terms of the 
$L_{\infty}$-norm: $\|\hat{\beta} - \beta^{\dagger} \|_{\infty}$.

\item \textbf{Full support recovery}: We study if the support of $\beta^\dagger$ is completely recovered by $\hat{\beta}$, i.e., 
$\mathsf{1}[\text{Supp}(\beta^\dagger) = \text{Supp}(\hat{\beta})]$ --- we look at the average of this quantity across multiple replications, leading to an estimate for the probability of full support recovery. 
\end{itemize}

We would like to point out that SNR alone does not dictate how difficult the underlying statistical problem is (e.g., in terms of variable selection, estimation, or prediction error). The situation is rather subtle: in addition to SNR, the matrix $X$ and choices of $n,p,k^{\dagger}$ influence the statistical performance. For example, consider two instances with $p=100$ and $p=10^5$ where we set $k^{\dagger}=10,\Sigma=I,n=100,\text{SNR}=1$. For $p=100$, it is possible to obtain a model with estimation error smaller than the null model (with high probability), but this may not be possible for $p=10^5$. Similarly, for the case of constant correlation,
a problem with $\rho=0, \text{SNR}=1$ may be statistically easier 
than one with $\rho=0.5,\text{SNR}=100$ (with suitable choices of $n,p,k^\dagger$). 
The experiments that follow are intended to provide 
a solid understanding of how support recovery, sparsity, estimation error, and prediction error vary under different problem settings. We also seek to understand (i) when our algorithms can achieve full support recovery---a topic of considerable importance in high-dimensional statistics~\cite{david2017high_mod,wainwright2009information}; and (ii) when pure $L_0$ starts to overfit (a deeper statistical understanding of this seems to be in a nascent stage).

\subsection{Comparison among CD variants and IHT: Optimization performance}
\label{section:cdcomparison}

We investigate the optimization performance of the different algorithms for the $(L_0)$ problem. Particularly, we study the objective values and the number of iterations till convergence for IHT and the following variants of CD:
\begin{itemize}[leftmargin=*]
\itemsep -0.4em 
\item \textbf{Cyclic CD}: This is Algorithm \ref{alg:CD} with default cyclic order.
\item \textbf{Random CD}: This is a randomized version of CD, where the coordinates to be updated are chosen uniformly at random from $[p]$. This has been considered in \cite{randomCDL0}.
\item \textbf{Greedy Cyclic CD}: This is our proposed Algorithm \ref{alg:CD} with a partially greedy cyclic ordering of coordinates, described in Section \ref{section:RegPath}.
\end{itemize}
We generated a dataset with Exponential Correlation, $\rho = 0.5$, $n = 500$, $p = 2000$, SNR $= 10$, and a support size $k^{\dagger} = 100$. We generated 50 random initializations each with a support size of $100$, where the nonzero indices are selected uniformly at random in $1$ to $p$ and assigned values that are drawn from Uniform$(0,1)$. For every initialization, we ran Cyclic CD, Greedy Cyclic CD, and IHT and recorded the value of the objective function of the solution along with the number of iterations (here, one full pass over all $p$ coordinates is defined as one iteration) till convergence. For Random CD, we ran the algorithm 10 times for every initialization and averaged the objective values and number of iterations. For all the algorithms above, we declare convergence when the relative change in the objective is $<10^{-7}$. Figure \ref{fig:CDVariants} presents the results: the objective values resulting from Greedy Cyclic CD are significantly lower than the other methods; on average we have roughly a 12\% improvement in the objective from Random CD and 55\% improvement over IHT.
\begin{figure}[h!]
\centering 
\includegraphics[scale=0.5]{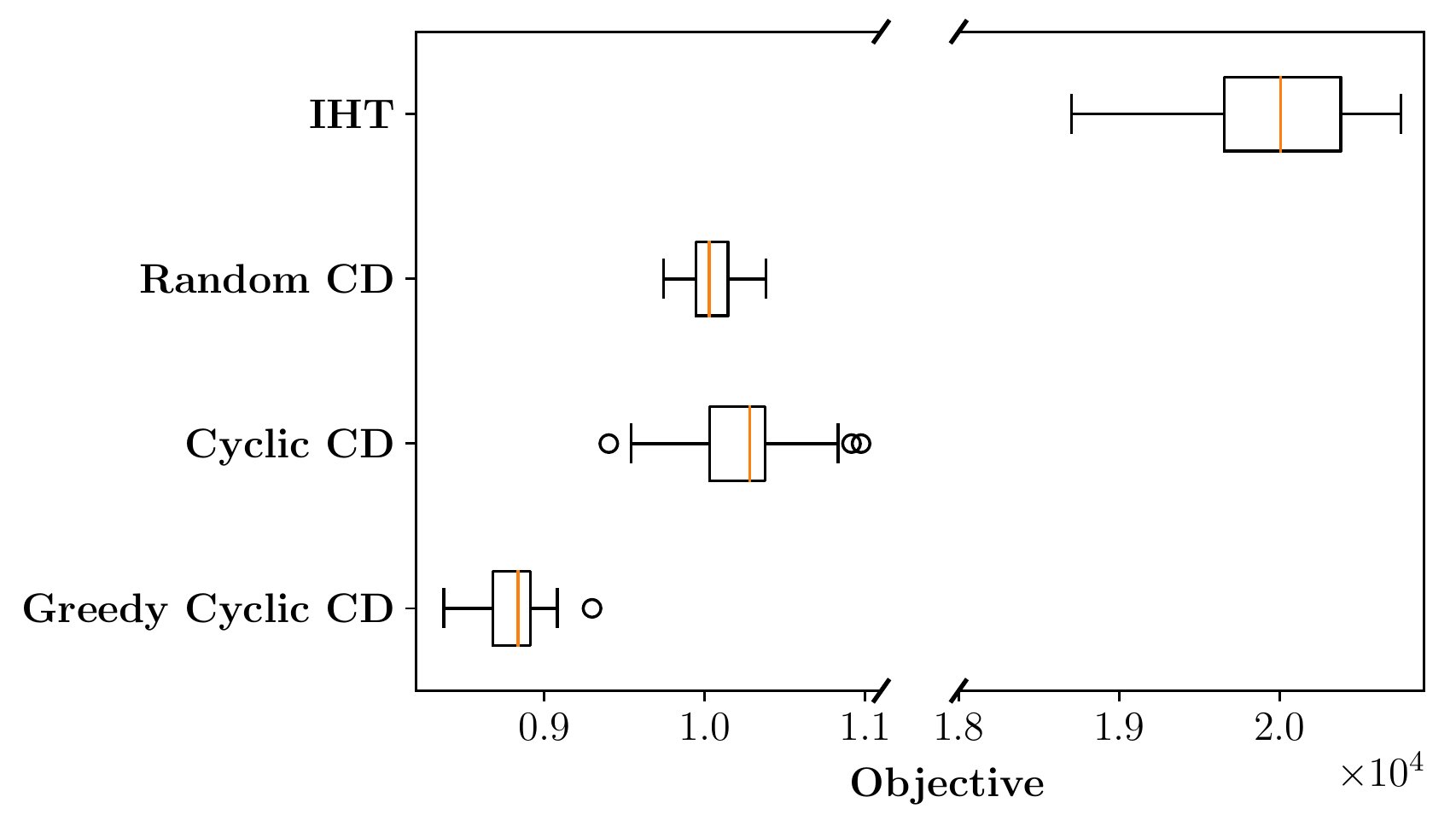}
\includegraphics[scale=0.512, trim={4.7cm 0 0 0},clip]{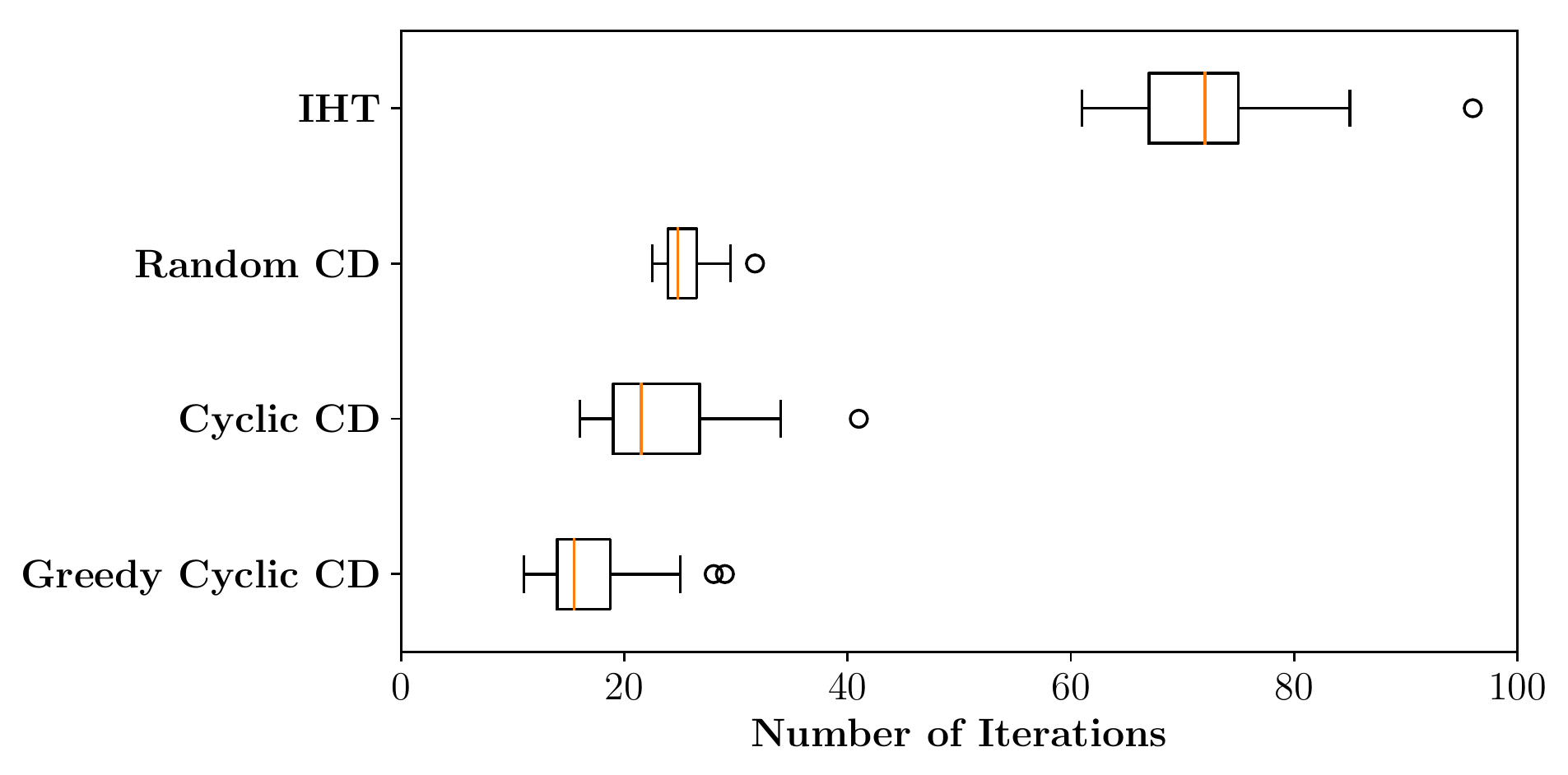}
\caption{Box plots showing the distribution of the objective values and the number of iterations (here, one full pass over all $p$ coordinates is defined as one iteration) till convergence for different variants of CD and IHT, for each algorithm we used 50 random initializations (as described in the text). The ticks of the box plots represent 1.5 times the inter-quartile range. }
\label{fig:CDVariants}
\end{figure}
This finding can be partially explained in the light of our discussion in Section \ref{section:IHT}, where we observed that the Lipschitz constant $L$ controls the quality of the solutions returned by IHT. In this high-dimensional setting, $L \approx 11$ which is far from $1$, and thus IHT can get stuck in relatively weak local minima. The number of iterations till convergence is also in favor of Greedy Cyclic CD, which requires roughly 28\% less iterations than Random CD and 75\% less iterations than IHT.

\subsection{Statistical Performance for Varying Sample Sizes}
\label{section:PhaseTransitionsN}
We study how the different statistical metrics change with the number of samples, while the other factors ($p,k^\dagger,\text{SNR}, \Sigma$) are fixed. We seek to empirically validate our hypothesis that:  {\em Under difficult statistical settings (e.g., high correlation or a small value of $n$), advanced optimization techniques such as combinatorial search, can lead to significantly improved statistical performance.}

We consider Algorithms 1 and 2 (with $k=1$) for the $(L_0)$, $(L_0 L_1)$, and $(L_0 L_2)$ problems; in addition to the competing penalties discussed in Section \ref{section:expsetup}. 
In Experiments~1 and 2 (below), we swept $n$ between $100$ and $1000$, and for every value of $n$, we generated 20 random training and validation datasets having Exponential Correlation, $p = 1000$, $k^{\dagger} = 20$, and SNR=5. All the results we report here are based on validation-set tuning. In the supplementary, we include the results for oracle and random design tuning.
\subsubsection{Experiment 1: High Correlation}
Here we choose $\rho = 0.9$ (exponential correlation)---this is a difficult problem due to the high correlations among features in the sample. Figure \ref{fig:NSweep-Exp09} summarizes the results. In the top panel of Figure~2 we show the results of Algorithm~\ref{alg:mipcd} applied to the $(L_0)$ and $(L_0 L_2)$ problems versus the other competing algorithms. In the bottom panel, we present a detailed comparison among Algorithms~1 and~2 for all the three problems: $(L_0)$, $(L_0 L_1)$ and $(L_0 L_2)$. 

\begin{figure}[h!]
\centering
Exponential Correlation, $\rho = 0.9$, $p = 1000$, $k^{\dagger} = 20$, SNR $=5$
\includegraphics[scale=0.5]{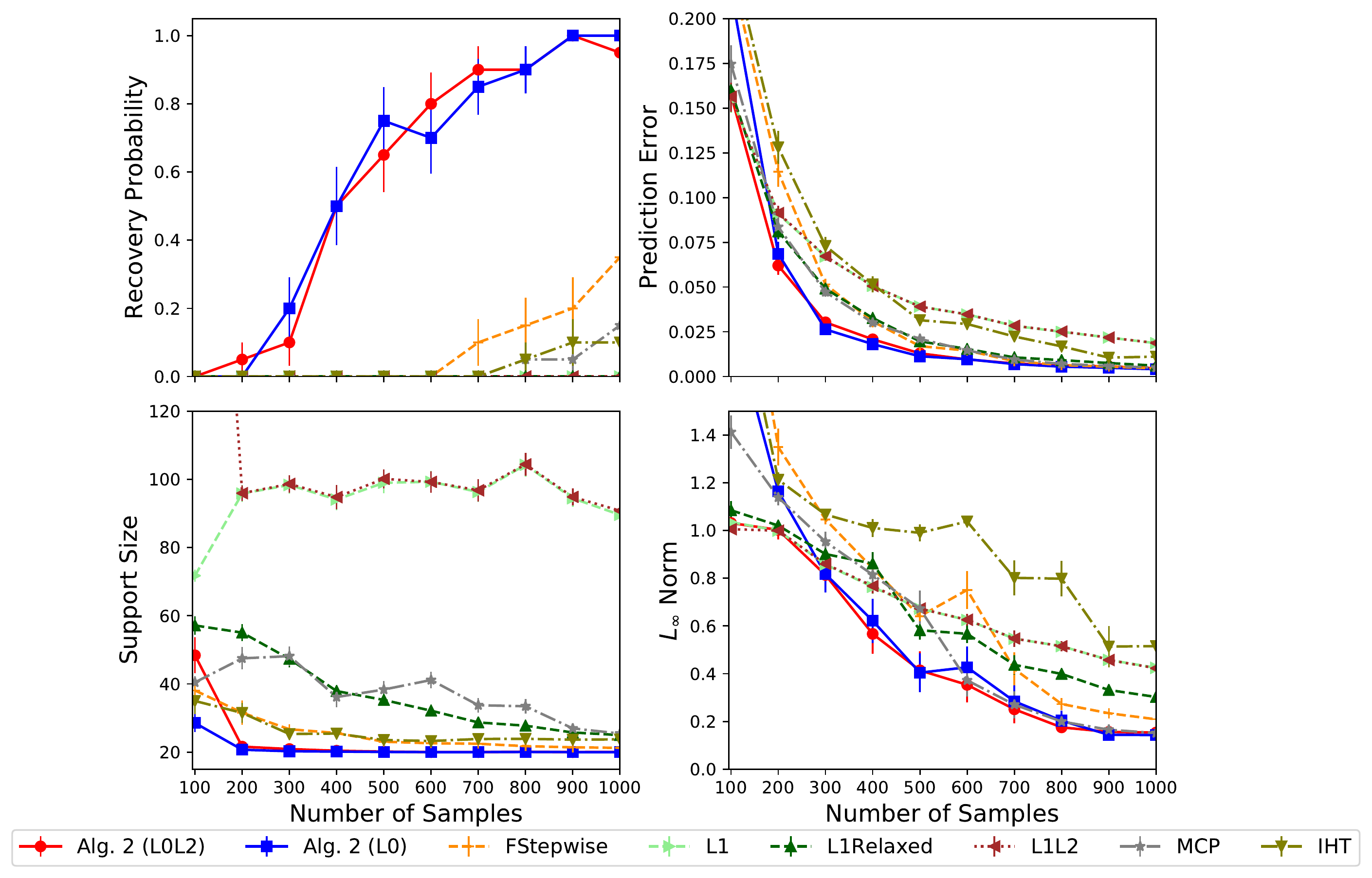}
\includegraphics[scale=0.5]{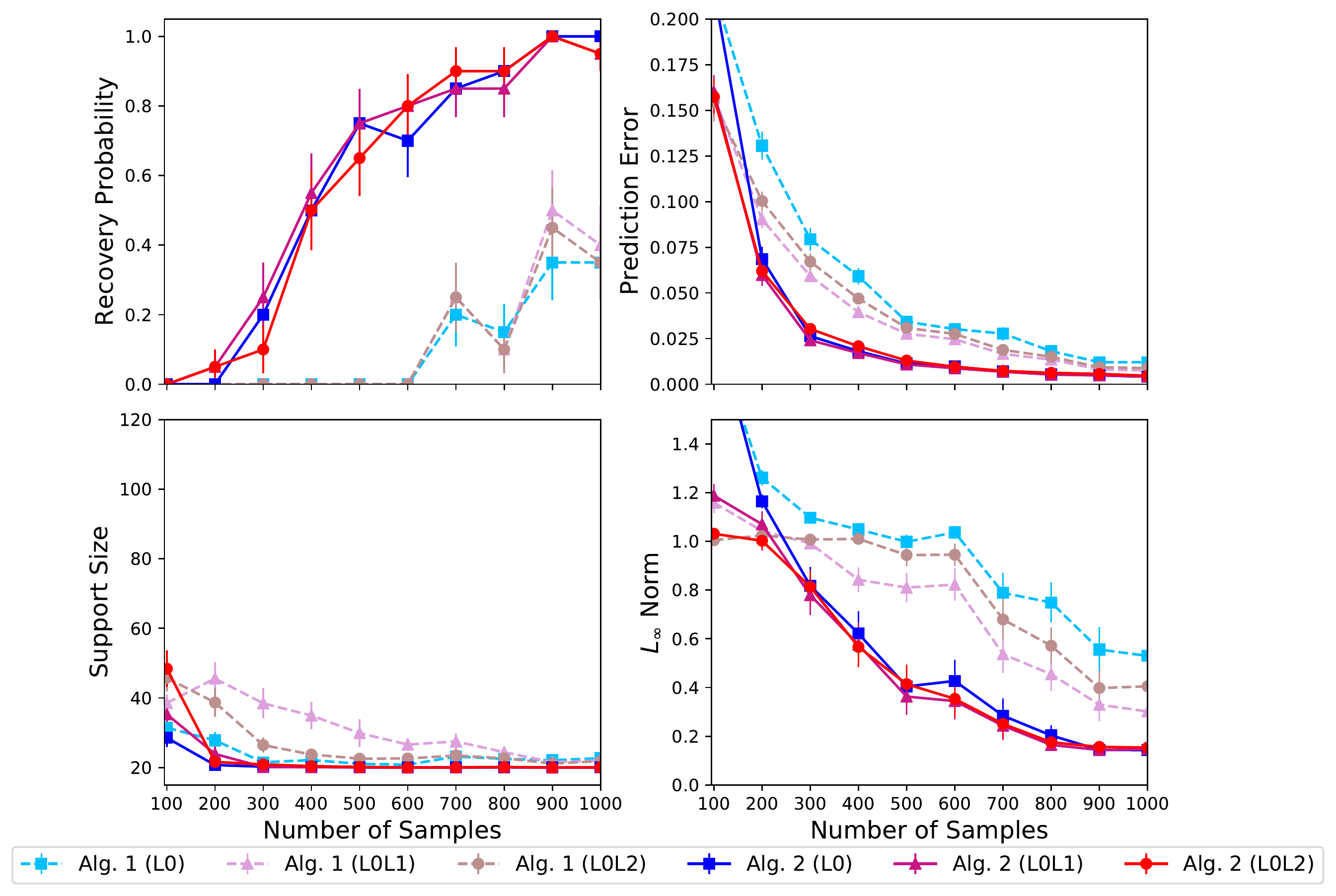}
\caption{{\small{Performance measures as the number of samples $n$ varies between $100$ and $1000$. The top figure compares two of our methods (Algorithm~\ref{alg:mipcd} for the $(L_0)$ and $(L_{0}L_{2})$ problems) with other state-of-the-art algorithms. The bottom figure compares Algorithms 1 and 2 for all three problems. Algorithm~\ref{alg:mipcd} performs significantly better than Algorithm 1 and the other methods, since it does a better job at optimization. }}}
\label{fig:NSweep-Exp09}
\end{figure}

From the top panel (Figure \ref{fig:NSweep-Exp09}), we can see that Algorithm~\ref{alg:mipcd} applied to the $(L_0 L_2)$ problem overall achieves the best performance in terms of different measures across all values of $n$. Algorithm~\ref{alg:mipcd} $(L_0)$ and Algorithm~\ref{alg:mipcd} $(L_0L_2)$
perform similarly for $n \geq 300$. The probability of full recovery for Algorithm~\ref{alg:mipcd} increases with $n$ and becomes $1$ at around $n=900$---note that the slight variation in the recovery probability values for our methods are solely due to the validation procedure (the oracle tuning presented in the supplementary eliminates this wiggly behavior). 
Lasso, Relaxed Lasso, and Elastic Net do not achieve full support recovery for any $n$---the corresponding lines are aligned with the horizontal axis. The $L_1$-based methods also lead to large support sizes---as expected, $L_1 L_2$ leads to supports that are denser than Lasso~\cite{Zou05regularizationand}.
{Due to shrinkage, tuning parameter selection based on prediction error makes the Lasso select models with many nonzero coefficients---this leads to sub-optimal variable selection.
The Relaxed Lasso attempts to undo the shrinkage effect of the Lasso leading to models with fewer nonzeros. In addition, we note that shrinkage of Lasso also interferes with variable selection---a shortcoming that is inherited by the Relaxed Lasso---as seen in the panel displaying recovery probability.} 

Moreover, MCP, FStepwise, and IHT have a probability of recovery around $0.3$ even when $n = p = 1000$---suggesting that they fail to do correct support recovery in this regime. A similar phenomenon occurs for the prediction error and the $L_{\infty}$ norm, where Algorithm~\ref{alg:mipcd} is seen to dominate.

The bottom figure shows that Algorithm~\ref{alg:mipcd} can significantly outperform Algorithm 1 (which performs no swaps). It seems that in this highly correlated setting, our local combinatorial optimization procedures have an edge in performance.

\subsubsection{Experiment 2: Mild Correlation}
In this experiment, we keep the same setup as the previous experiment, but we reduce the correlation parameter $\rho$ to $0.5$. In Figure \ref{fig:NSweep-Exp05}, we show the results for Algorithm 1 applied to $(L_0)$ and Algorithm~\ref{alg:mipcd} applied to ($L_0 L_2$) versus the other competing methods. We note that our other algorithms have a similar profile (we do not include their plots for space constraints). 
\begin{figure}[h!]
\centering
{\sf {Exponential Correlation, $\rho = 0.5$, $p = 1000$, $k^{\dagger} = 20$, SNR $=5$ }}
\includegraphics[scale=0.55]{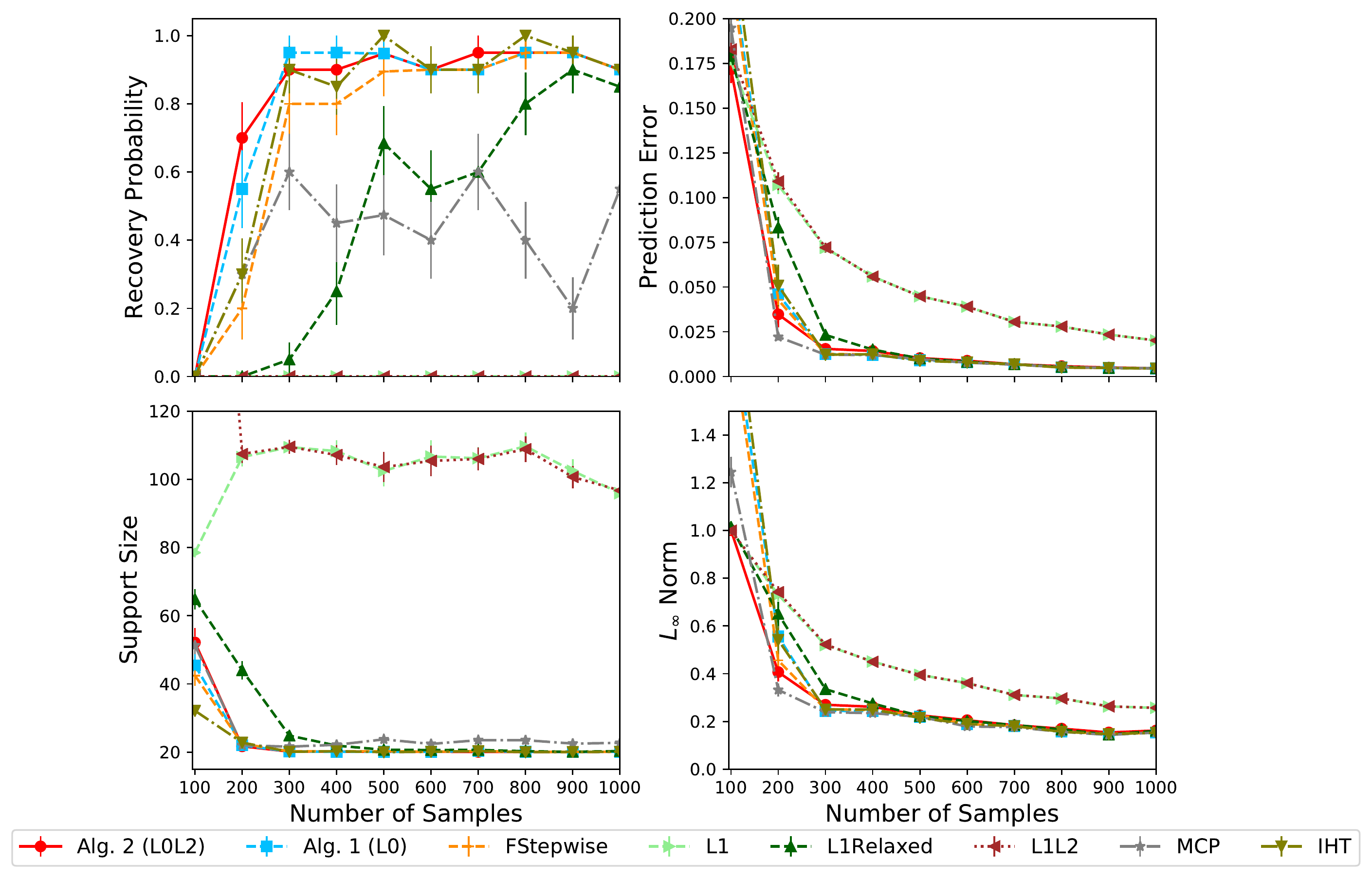}
\caption{\small{Performance measures as the number of samples $n$ varies between $100$ and $1000$. The figure compares two of our methods (Algorithm~\ref{alg:mipcd} $(L_0 L_2)$ and Algorithm 1 for $(L_0)$) with other state-of-the-art algorithms. Algorithms 1 and 2 perform similarly in this case (in contrast to the highly correlated setting in Figure \ref{fig:NSweep-Exp09}). Adding $L_1$ or $L_2$ regularization to $(L_0)$ does not help in this case.}}
\label{fig:NSweep-Exp05}
\end{figure}
This setup is easier from a statistical perspective, when compared to Experiment 1 where $\rho = 0.9$. Thus, we expect all the methods to perform better (overall) and display a phase-transition (for recovery probability) at a smaller sample size (compared to Experiment 1). Indeed, as shown in Figure~\ref{fig:NSweep-Exp05}, Algorithm~\ref{alg:CD}~($L_0$) and Algorithm~\ref{alg:mipcd}~($L_0 L_2$) have roughly the same profiles, and they outperform the other methods; they fully recover the true support using roughly $300$ samples. The swap variants of our methods in this case do not seem to lead to significant improvements over the non-swap variants, and this further supports our hypothesis: when the statistical problem is relatively easy, Algorithm~\ref{alg:CD} works quite well---more advanced optimization (e.g., using swaps) do not seem necessary. MCP and FStepwise also exhibit good performance, but they start doing full support recovery for much larger values of $n$; and MCP does not seem to be robust. Lasso in this case never recovers the true support, and this property is inherited by Relaxed Lasso which requires at least 900 samples to match our methods in terms of support recovery.
\subsection{Statistical Performance for Varying SNR} \label{section:PhaseTransitionsSNR}
We present two experiments studying the role of varying SNR values on the different performance measures. In each experiment, we vary the SNR between $0.01$ and $100$.
For every SNR value, we generated 20 random datasets over which we averaged the results. We observe that for low SNR values, ridge regression ($L_2$) works very well in terms of prediction performance. Hence we include $L_2$ in our results. {We do not include the results for IHT in this case as its run times are substantially longer compared to competing algorithms}. The results are based on validation-set tuning, and those for oracle and random design tuning are included in the supplementary.

\subsubsection{Experiment 1: Constant Correlation} \label{section:exp1constcorr}
We generated datasets with constant correlation, $\rho = 0.4$, $n = 1000$, $p = 2000$, and $k^{\dagger} = 50$. We report the results for Algorithm~\ref{alg:mipcd} applied to the $(L_0)$, and $(L_0 L_2)$ problems along with all the other state-of-the-art algorithms in Figure \ref{fig:SNRSweep-C}.

\begin{figure}[h!]
\centering
Constant Correlation, $\rho = 0.4$, $n=1000$, $p = 2000$, $k^{\dagger} = 50$
\includegraphics[scale=0.5]{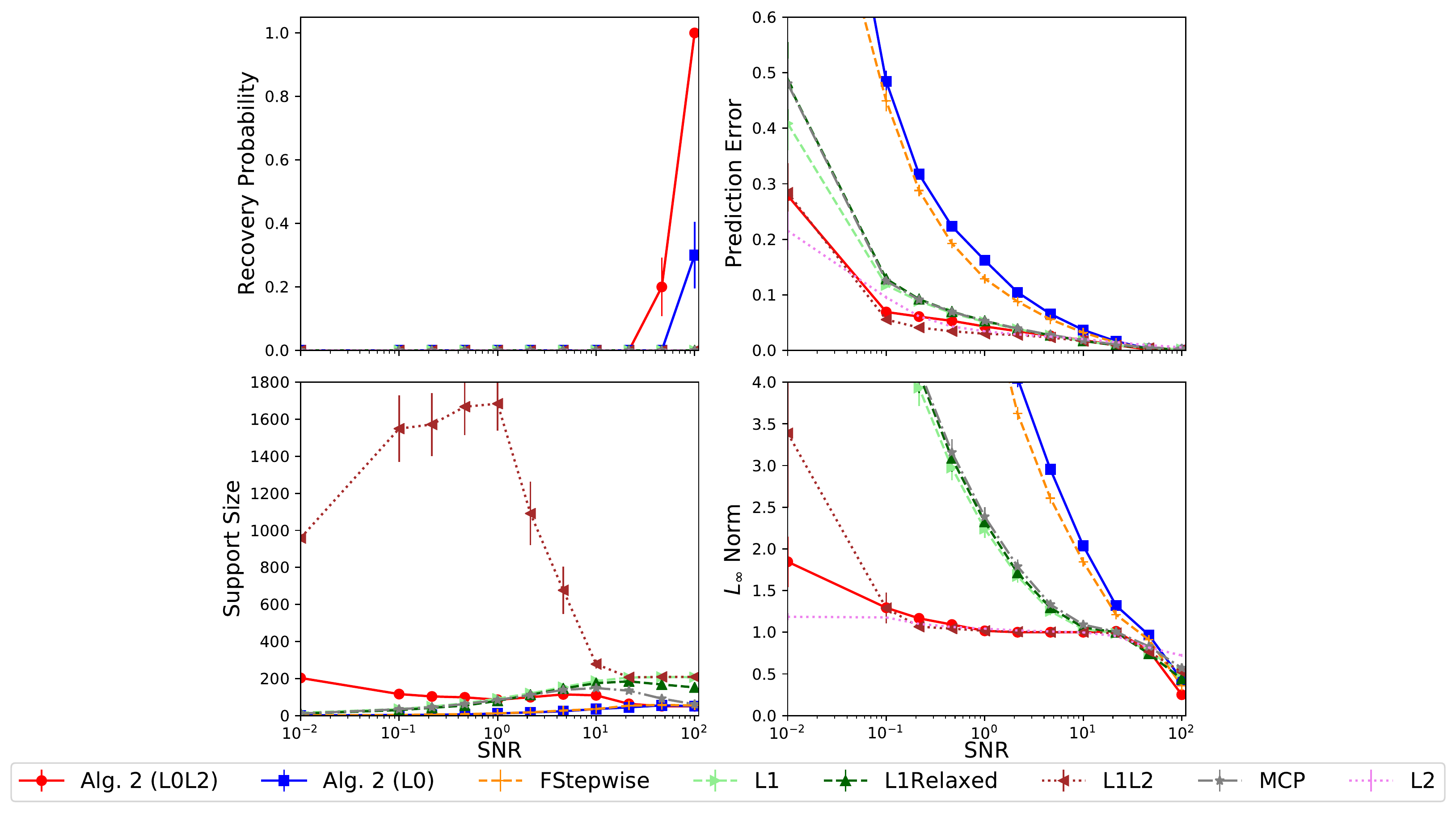}
\caption{\small{Performance measures as the signal-to-noise ratio (SNR) is varied between 0.01 and 100. The figure compares two of our methods (Algorithm~\ref{alg:mipcd} applied to the $(L_0)$ and $(L_0 L_2)$ problems) with other state-of-the-art algorithms. For low SNR levels, $(L_0 L_2)$ performs much better than $(L_0)$ (as the latter overfits in these settings).}}
\label{fig:SNRSweep-C}
\end{figure}
Figure \ref{fig:SNRSweep-C} suggests that full support recovery is difficult for all methods (suggesting that constant correlation leads to a difficult problem). At SNR $= 100$, ($L_0 L_2$) fully recovers the support while ($L_0$) has a recovery probability $\sim$ 0.3---this suggests that the additional $L_2$-regularization aids the optimization performance of Algorithm~\ref{alg:mipcd}. However, none of the other considered methods does full recovery, even for high SNR. Algorithm~\ref{alg:mipcd} ($L_0 L_2$) generally exhibits excellent performance in terms of all measures across the whole SNR range. Pure ($L_0$) tends to overfit quickly (as SNR becomes smaller) due to the lack of shrinkage~\cite{mazumder2017subset} and it selects small support sizes (like ``FStepwise''). Additional shrinkage ($L_0L_2$) seems to help alleviate this problem. The Elastic Net ($L_1 L_2$) performs similarly to ($L_0 L_2$) in terms of prediction error, but at the cost of very dense supports---in fact its support size can reach up to $90\%$ of $p$ for SNR $\sim 1$, which is undesirable from the viewpoint of having a parsimonious model. We note that for low SNR values, ($L_0 L_2$)'s prediction error is comparable to that of $L_2$ (for SNR=0.01, $L_2$ has the best predictive performance)---however,  ($L_0 L_2$) leads to much sparser models and hence has an advantage.   

\subsubsection{Experiment 2: Exponential Correlation}
We generated datasets having exponential correlation with $\rho = 0.5$, $n = 1000$, $p = 5000$, and $k^{\dagger} = 50$. We report the results in Figure \ref{fig:SNRSweep-Exp}.
We observe that this setup is relatively easier (from a statistical viewpoint) than the constant correlation experiment in Section~\ref{section:exp1constcorr}. Thus, we observe less significant differences among the algorithms, when compared to the first experiment (see Figures \ref{fig:SNRSweep-C} and \ref{fig:SNRSweep-Exp}). Algorithm~\ref{alg:mipcd} ($L_0 L_2$) again seems to dominate across different measures and for all SNR values. Algorithm~\ref{alg:mipcd} ($L_0$) and Algorithm~\ref{alg:mipcd} ($L_0 L_2$) exhibit similar performance for high SNR. For low SNR, Algorithm~\ref{alg:mipcd} ($L_0 L_2$), $L_2$, and $L_1 L_2$ have the best predictive performance, though ($L_0 L_2$) leads to the most compact models.  We note that even in this relatively easy case, Lasso and Elastic Net never fully recover the support---MCP and Relaxed Lasso also suffer in terms of full support recovery. 

\begin{figure}[h!] 
\centering
Exponential Correlation, $\rho = 0.5$, $n=1000$, $p = 5000$, $k^{\dagger} = 50$
\includegraphics[scale=0.5]{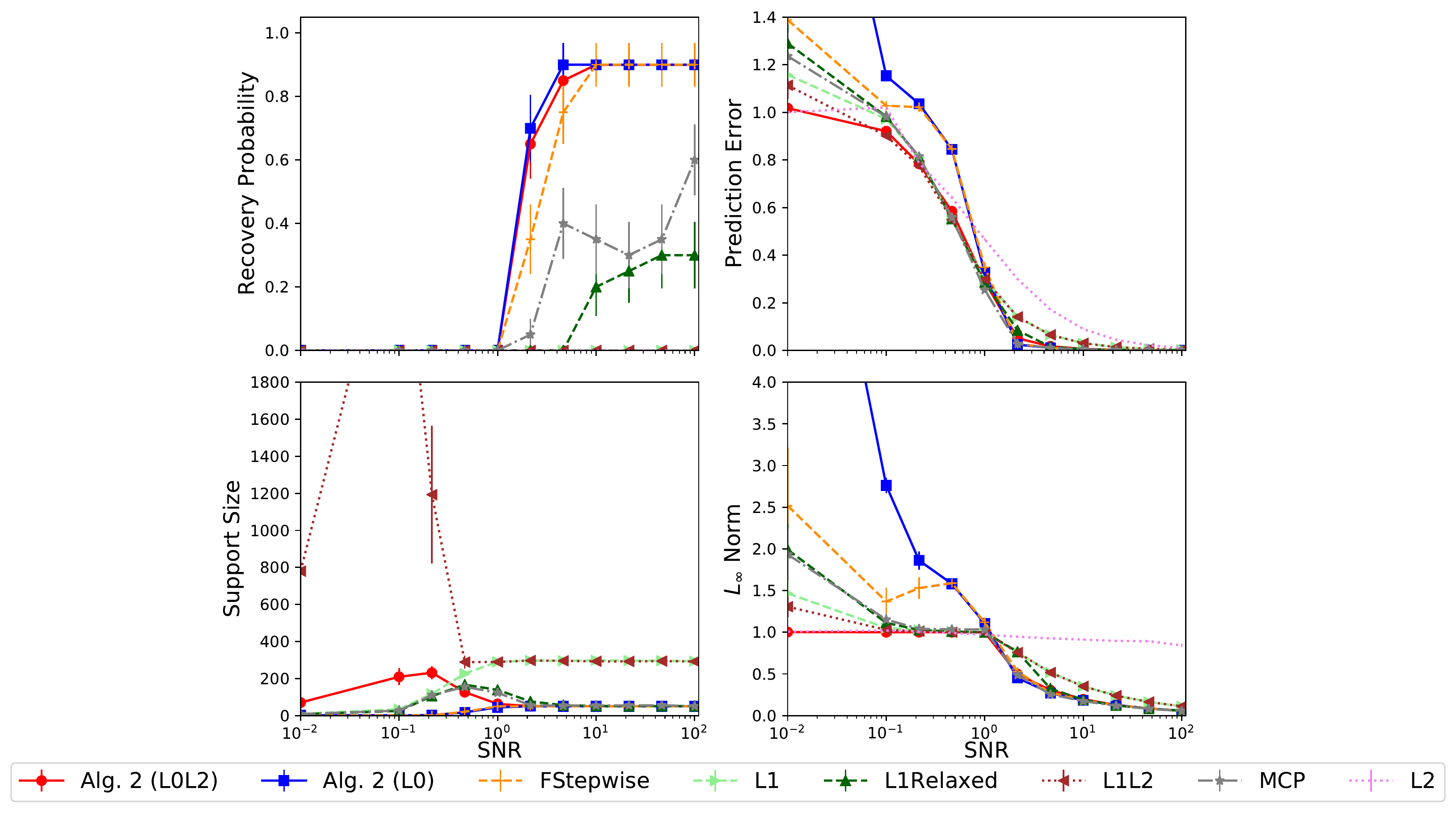}
\caption{\small{Performance measures as the signal-to-noise ratio (SNR) is varied between 0.01 and 100. The figure compares two of our methods (Algorithm~\ref{alg:mipcd} applied to the $(L_0)$ and $(L_0 L_2)$ problems) with other state-of-the-art algorithms. For SNR$\leq $1, $(L_0 L_2)$ performs significantly better than $(L_0)$---this performance improvement vanishes for larger SNR values.}}
\label{fig:SNRSweep-Exp}
\end{figure}

\subsection{Comparing PSI$(k)$ versus FSI$(k)$} \label{swap-inescp-minima}
Here, we examine the differences among the various classes of minima introduced in the paper, i.e.,  CW, PSI$(k)$ and FSI$(k)$ minima, for the ($L_0$) problem. To understand the differences, we consider a relatively difficult setting with Constant Correlation where $\rho = 0.9$, $n = 250$, $p = 1000$, and $k^{\dagger} = 25$. We set SNR$=300$ to allow for full support recovery. We generated 10 random training datasets under this setting and ran: Algorithm~1; and the PSI and FSI variants of Algorithm~\ref{alg:mipcd} for $k \in \{1,2,5\}$. All algorithms were initialized with a vector of zeros. For Algorithm~\ref{alg:mipcd} we used Gurobi (v7.5) to solve the MIQO subproblems (\ref{mio:psi}) and (\ref{mio:fsi}) when $k>1$.

Figure \ref{fig:localminima} presents box plots showing the distribution of objective values, true positives, and false positives recorded for each of the algorithms and 10 datasets. PSI($1$) and FSI($1$) minima lead to a significant reduction in the objective, when compared to Algorithm~1 (which results in CW minima). We do observe further reductions as $k$ increases, but the gains are less pronounced.
In this case, CW minima contains on average a large number of false positives ($> 35$) and few true positives---this is perhaps due to high correlations among all features, which makes the optimization task arguably very challenging. Both PSI and FSI minima increase the number of true positives significantly. A closer inspection shows that FSI minima do a better job in having fewer false positives, when compared to PSI minima. This comes at the cost of solving relatively more difficult optimization problems, but within reasonable computation times.

\begin{figure}[h!]
\centering
Constant Correlation, $\rho = 0.9$, $n=250$, $p = 1000$, $k^{\dagger} = 25$, SNR $=300$
\includegraphics[scale=0.46]{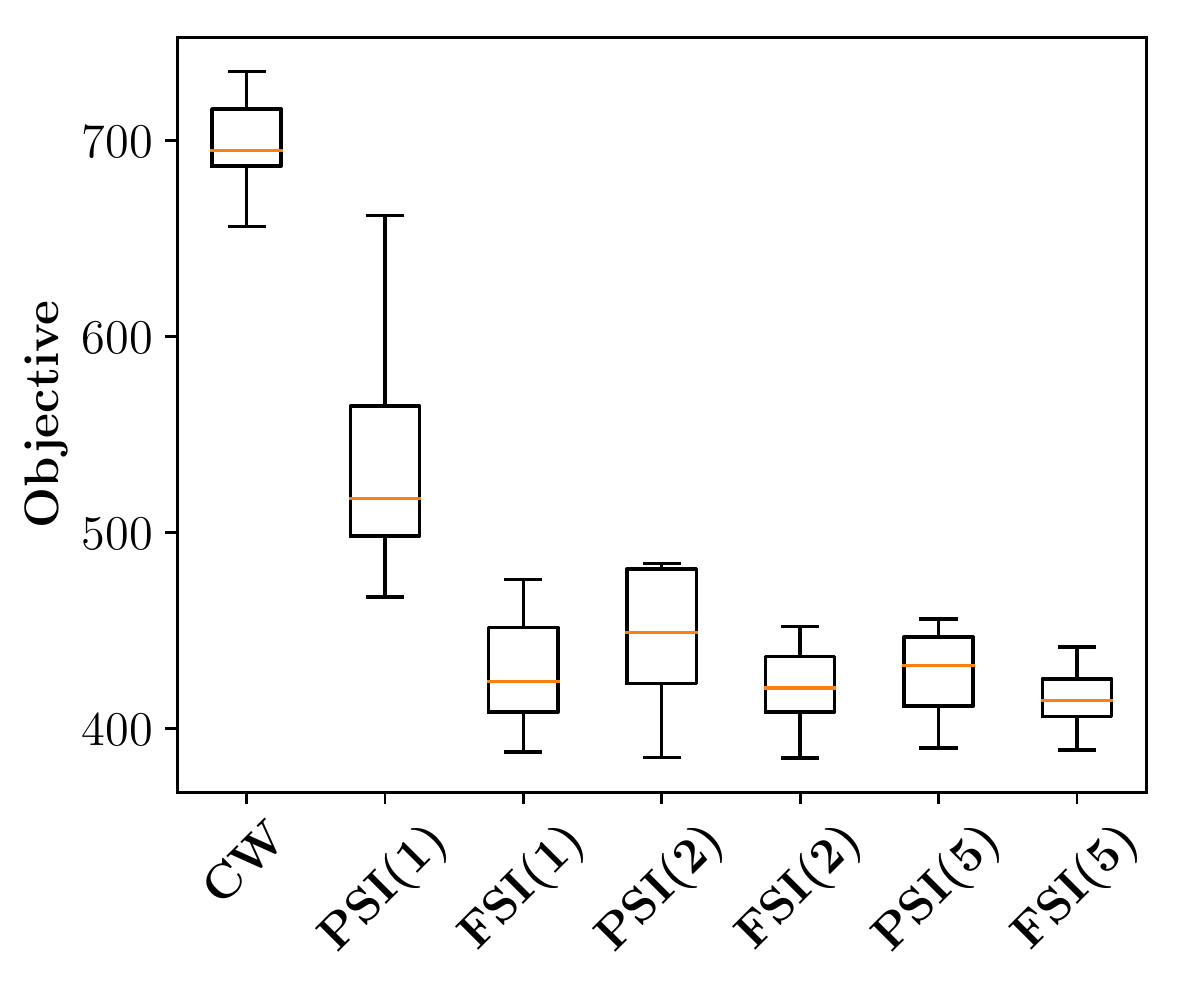} 
\hspace{-0.3cm}
\includegraphics[scale=0.46]{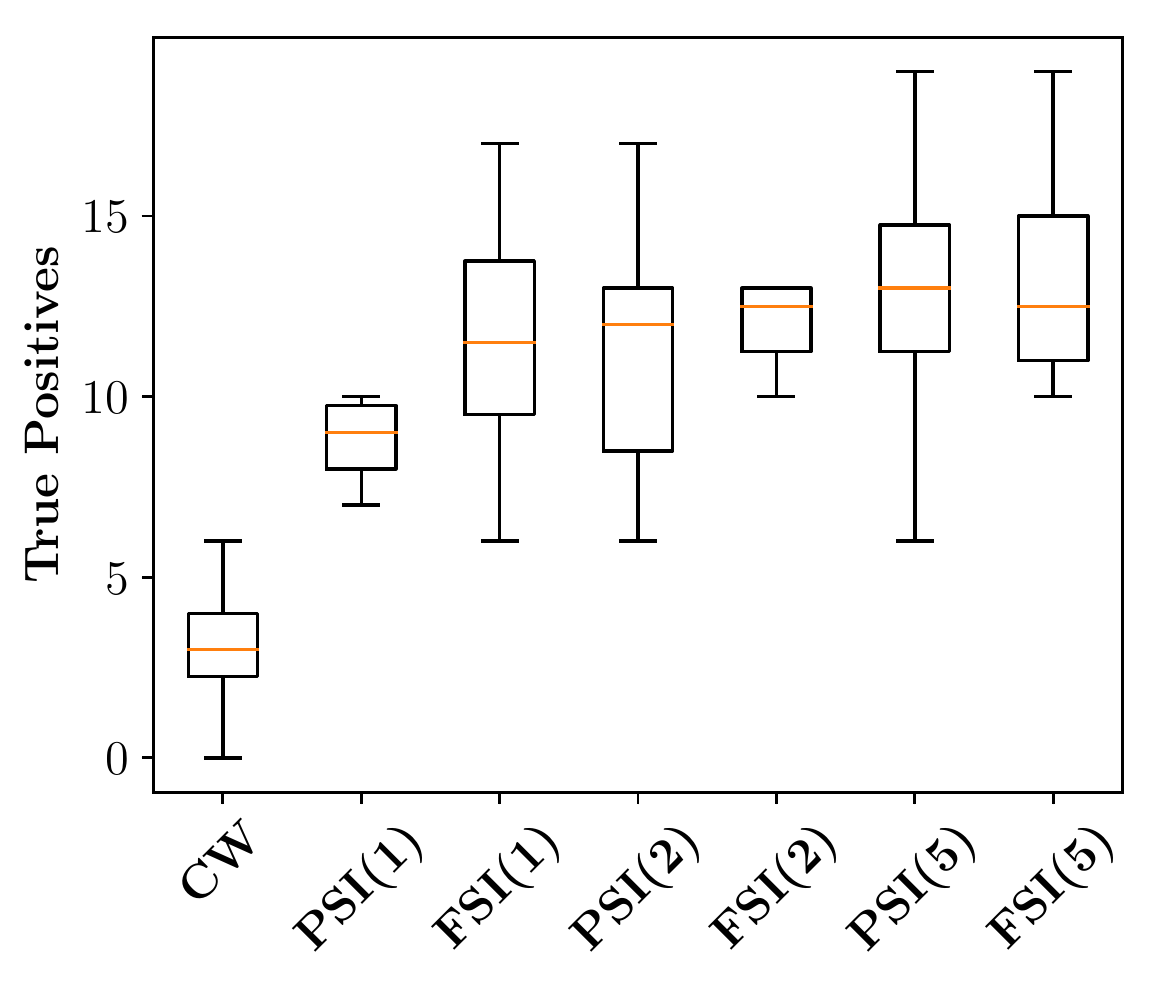}
\hspace{-0.3cm}
\includegraphics[scale=0.46]{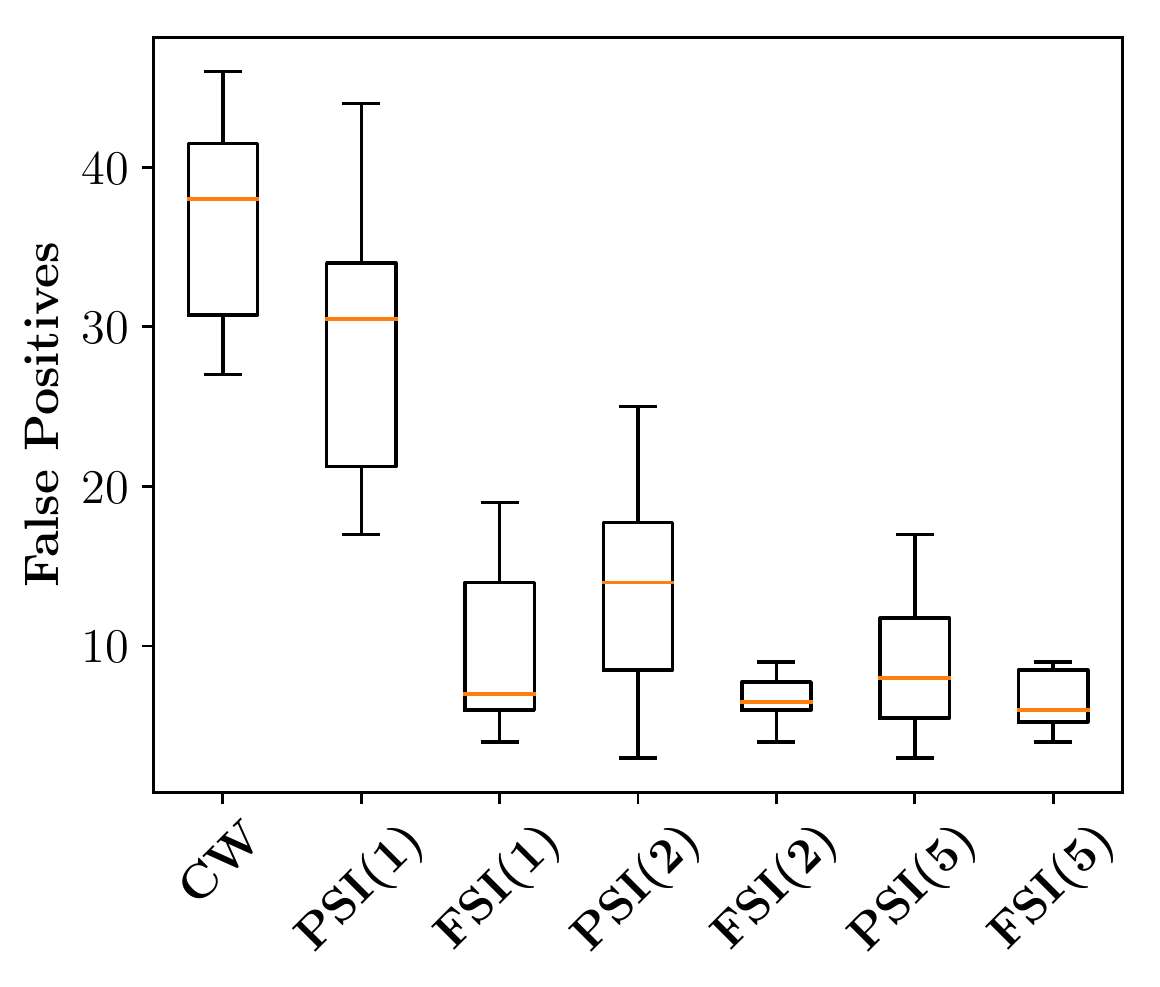}
\caption{Box plots showing the distribution of objective values, number of true positives, and number of false positives for the different classes of local minima.}
\label{fig:localminima}
\end{figure}
In the supplementary, we present an experiment studying the evolution of the intermediate solutions before Algorithm~\ref{alg:mipcd} reaches an FSI($k$) minimum---we observe that CD is effective at increasing the true positives, while local combinatorial search significantly reduces the false positives.

\subsection{Large High-dimensional Experiments} \label{section:high-dim-exp}
\pparagraph{Synthetic Experiments:} Here, we investigate the performance of the different algorithms when $p \gg n$. We ran two experiments with a large number of features under the following settings:
\begin{itemize}
\item \textbf{Setting 1}: Exponential Correlation, $\rho = 0.5$, $n = 1000$, $p=10^5/2$, $k^{\dagger} = 100$, and SNR $=10$
\item \textbf{Setting 2}: Constant Correlation, $\rho = 0.3$, $n = 1000$, $p = 10^{5}$, $k^{\dagger} = 50$, and SNR $=100$
\end{itemize}
Every experiment is performed with 10 replications, and the results are averaged. We report the results for Settings 1 and 2 in Table~\ref{table:exp}. 
\begin{table}[h!]
{\scalebox{.91}{\begin{tabular}{ll}
\centering
\hspace{+2cm} \textbf{Setting 1} $(n=1000,p=10^5/2, \rho=0.5)$ & \hspace{-7cm} \textbf{Setting 2} $(n=1000,p=10^5,\rho=0.3)$ \\
\begin{tabular}{@{}lcccc|}
\toprule
\textbf{Method} & $\| \beta \|_0$ & \textbf{TP} & \textbf{FP} & \textbf{PE}$\times 10^{2}$ \\
 \midrule
Alg 2 ($L_0$)         & $ 160 \pm 24 $               & $ 79 \pm 9$                   & $ 81 \pm 33$                     & $ 5 \pm 1.6$                   \\
Alg \ref{alg:CD} ($L_0 L_2$)            & $ \bf 100 \pm 0 $               & $ \bf 100 \pm 0 $                    & $ \bf 0 \pm 0 $                      & $ \bf 0.97 \pm 0.05$                \\
Alg \ref{alg:CD} ($L_0 L_1$)            & $ \bf 100 \pm 0 $              & $ \bf 100 \pm 0 $                    & $ \bf 0 \pm 0 $                      & $ 1 \pm 0.05$                 \\
L1              & $ 808 \pm 7 $                 & $ 95 \pm 1$                    & $ 712 \pm 7$                   & $ 7.9 \pm 0.17 $                \\
L1Relaxed       & $ 602  \pm 40$                 & $ 95 \pm 1$                    & $ 508 \pm 41$                    & $ 7.9 \pm 0.19$                  \\
MCP             & $ 102 \pm 1 $                 & $ \bf 100 \pm 0 $                    & $ 2.3 \pm 1$                      & $ \bf 0.97 \pm 0.05$                 \\ 
FStepwise       & $ 216 \pm 17$          & $ 64 \pm 7$                    & $ 152 \pm 23$                    & $ 8.9 \pm 1.3$                     \\\bottomrule
\end{tabular}
\begin{tabular}{cccc@{}}
\toprule
 $\| \beta \|_0$ & \textbf{TP} & \textbf{FP} & \textbf{PE}$\times 10^{3}$ \\
  \midrule
          $ 69 \pm 18$ &            $ 47 \pm 3$ &             $ 22 \pm 22$ &           $1.6 \pm 1$  \\
          $\bf 50 \pm 0$ &             $\bf 50 \pm 0$ &              $\bf 0 \pm 0$ &           $ \bf 0.5 \pm 0.02$ \\
        $\bf 50 \pm 0$ &             $\bf 50 \pm 0$ &              $\bf 0 \pm 0$ &           $ \bf 0.5 \pm 0.02$   \\
          $478 \pm 11$ &            $\bf 50 \pm 0$ &            $ 428 \pm 11 $ &           $ 4.7 \pm 0.1$ \\
         $ 385 \pm 12$ &            $ 50 \pm 0.2$ &            $ 335 \pm 13$ &          $ 4.4 \pm 0.2$ \\
     $65 \pm 3$ &             $\bf 50 \pm 0$ &             $15 \pm 3$ &           $ 3.5 \pm 0.13$     \\ 
           $75 \pm 2$ &            $ \bf 50 \pm 0 $ &             $ 25 \pm 2$ &           $ 1.1\pm 0.07$ \\
     \bottomrule
\end{tabular}
\end{tabular}}}
\caption{Performance measures for the different algorithms under Settings 1 and 2. TP, FP, and PE denote the True Positives, False Positives, and Prediction Error, respectively. The standard error of the mean is reported next to every value.}
\label{table:exp}
\end{table}

In Table \ref{table:exp}, Algorithm \ref{alg:CD} for the $(L_0 L_1)$ and $(L_0 L_2)$ problems fully recovers the true support and attains the lowest prediction error. None of the other methods was able to do full support recovery; Lasso and Relaxed Lasso capture most of the true positives but include a very large number of false positives. MCP comes in the middle between $(L_0 L_1)$/$(L_0 L_2)$ and Lasso---it captures all the true positives and includes few false positives. We also note that in such high SNR settings, we do not expect shrinkage (arising from the $L_1$/$L_2$ penalties) to lead to major statistical improvements. Thus, the difference in performance between $(L_0)$ and $(L_0 L_1)$/$(L_0 L_2)$ seems to be due to the continuous regularizers that help in optimization.

\pparagraph{Timings and Out-of-sample Performance:} 
We ran Algorithm~1 using our toolkit \texttt{L0Learn} and compared the running time and predictive performance versus \texttt{glmnet} and \texttt{ncvreg}, on a variety of real and synthetic datasets. For the real datasets, there is no ground truth---we study predictive performance vis-a-vis model sparsity. We note that \texttt{L0Learn}, 
\texttt{glmnet}, and \texttt{ncvreg} are solving different optimization problems---the run times provided herein are meant to demonstrate that a main workhorse for our proposed framework is competitive when compared to efficient state-of-the-art implementations for sparse learning.  
Below we provide some details about the datasets:
\begin{itemize}[leftmargin=*]
\item \textbf{House Prices}: $p = 104,000$ and $n = 200$. We added pairwise interactions to the popular Boston House Prices dataset \cite{bostonhouseprices} to get 104 features. Then, we added random ``probes'' (aka noisy features) by appending to the data matrix 1000 random permutations of every column. The validation and testing sets have 100 and 206 samples, respectively. 
\item \textbf{Amazon Reviews}: $p = 17,580$ and $n = 2500$. We used the Amazon Grocery and Gourmet Food dataset \cite{amazonreviews} to predict the helpfulness of every review (based on its text). Specifically, we calculated the helpfulness of every review as the ratio of the number of up votes to that of down votes, and we obtained $X$ by using Scikit-learn's TF-IDF transformer (while removing stopwords). The validation and testing sets have 500 and 1868 samples, respectively. We also created an augmented version of this dataset where we added random probes by appending to the data matrix 9 random permutations of every column to get $p = 174,755$.
\item \textbf{US Census}: $p = 55,537$ and $n=5000$. We used 37 features extracted from the 2016 US Census Planning Database to predict the mail-return rate\footnote{We thank Dr. Emanuel Ben David, US Census Bureau for help on preparing this dataset.} \cite{uscensus}. We appended the data matrix with 1500 random permutations of every column, and we randomly sampled 15,000 rows, evenly distributed between the training, testing, and validation sets.
\item \textbf{Gaussian 1M}: $p=10^6$ and $n=200$. We generated a synthetic dataset with independent standard normal entries. We set $k^{\dagger} = 20$, SNR=10, and performed validation and testing  as described in Section~\ref{section:expsetup}.
\end{itemize}
For all real datasets, we tuned and tested on separate validation and testing sets. The timings were performed on a machine with an i7-4800MQ CPU and 16GB RAM running Ubuntu 16.04 and OpenBLAS 0.2.20. For all methods, we report the training time required to obtain a grid of 100 solutions. For $(L_0 L_2)$, $(L_0 L_1)$, and MCP, we provide the time for a fixed $\lambda_2$, $\lambda_1$, and $\gamma$, respectively (these parameters have been set to the optimal values obtained via validation set tuning over 10 values of the tuning parameter). Table~\ref{table:timings} presents run times for all the four methods.

The results presented in Table~\ref{table:timings} show the following: \texttt{L0Learn} is faster than \texttt{glmnet} and \texttt{ncvreg} on all the considered datasets, e.g., more than twice as fast on the Amazon Reviews dataset. The speed-ups can be attributed to the careful design of \texttt{L0Learn} (as described in Section~\ref{section:RegPath}) and due to the nature of $L_0$ regularization which generally selects sparser supports than those obtained by $L_1$ or MCP regularization. Moreover, \texttt{L0Learn}, for both the ($L_0 L_2$) and ($L_0 L_1$) problems, provides much sparser supports and competitive testing MSE compared to the other toolkits. Finally, we note that prediction errors for our methods can be potentially improved by using Algorithm~\ref{alg:mipcd}, at the cost of slightly increased computation times.
\begin{table}[h]
\centering
\scalebox{0.95}{\begin{tabular}{ll}
\hspace{2.5cm} \textbf{Amazon Reviews} & \hspace{1cm} \textbf{Amazon Reviews (+Probes)} \\
\hspace{2.0cm} \text{($p = 17,580$, $n = 2500$)} & \hspace{2cm} \text{($p = 174,755$, $n = 2500$)} \\
\begin{tabular}{lccc} 
\toprule
\textbf{Toolkit} & \multicolumn{1}{l}{\textbf{Time}} & \multicolumn{1}{l}{\textbf{MSE$\scriptscriptstyle \times 10^2$}} & \multicolumn{1}{l}{$\|\beta \|_0$} \\ \midrule
glmnet (L1)     &  7.3                               & 4.82                            & 542                             \\
L0Learn ($L_0 L_2$)   & 3.3                               & \bf 4.77                            & \bf 159                              \\
L0Learn ($L_0 L_1$)   & \bf 2.8                               & 4.79                            & 173                              \\
ncvreg (MCP)     & 10.9                              & 6.71                            & 1484                             \\ \bottomrule
\end{tabular} &
\begin{tabular}{lccc}
\toprule
\textbf{Toolkit} & \multicolumn{1}{l}{\textbf{Time}} & \multicolumn{1}{l}{\textbf{MSE$\scriptscriptstyle \times 10^{2}$}} & \multicolumn{1}{l}{$\|\beta \|_0$} \\ \midrule
glmnet (L1)     & 49.4                               & \bf 5.11                             & 256                          \\
L0Learn ($L_0 L_2$)   &  31.7                               & 5.18                             & 37                               \\
L0Learn ($L_0 L_1$)   &  \bf 29.5                              & 5.20                             & \bf 36                              \\
ncvreg (MCP)     &67.3                               & 5.33                             & 318                              \\ \bottomrule
\end{tabular} \\ \\
\hspace{2.5cm} \textbf{US Census} & \hspace{2.5cm} \textbf{House Prices} \\
\hspace{1.5cm} \text{($p = 55,537$, $n = 5000$)} & \hspace{1.5cm} \text{($p = 104,000$, $n = 200$)} \\
\begin{tabular}{lccc}
\toprule
\textbf{Toolkit} & \multicolumn{1}{l}{\textbf{Time}} & \multicolumn{1}{l}{\textbf{MSE}} & \multicolumn{1}{l}{$\|\beta \|_0$} \\ \midrule
glmnet (L1)     & 28.7                              & 61.3                             & 222                              \\
L0Learn ($L_0 L_2$)   &  19.6                               & \bf 60.7                             & 15                               \\
L0Learn ($L_0 L_1$)   & \bf 19.5                               & 60.8                            & \bf 11                              \\
ncvreg (MCP)     & 32.7                              & 62.02                             & 16                              \\ \bottomrule
\end{tabular} &
\begin{tabular}{lccc}
\toprule
\textbf{Toolkit} & \multicolumn{1}{l}{\textbf{Time}} & \multicolumn{1}{l}{\textbf{MSE}} & \multicolumn{1}{l}{$\|\beta \|_0$} \\ \midrule
glmnet (L1)     & 2.3                               & 100                             & 112                              \\
L0Learn ($L_0 L_2$)   & \bf 1.8                               & \bf 94                             & \bf 59                               \\
 L0Learn ($L_0 L_1$)   & \bf 1.8                               & 104                             & 74                              \\
ncvreg (MCP)     &  3.9                               & 102                             & 140                              \\ \bottomrule 
\end{tabular} \\ \\
\hspace{2.5cm} \textbf{Gaussian 1M} \\
\hspace{2.5cm} \text{($p = 10^{6}$, $n = 200$)} \\
%\multicolumn{2}{c}{
\begin{tabular}{lccc}
\toprule
\textbf{Toolkit} & \multicolumn{1}{l}{\textbf{Time(s)}} & \multicolumn{1}{l}{\textbf{MSE}} & \multicolumn{1}{l}{$\|\beta \|_0$} \\ \midrule
glmnet (L1)     &        22.5                       &    \bf 4.55                         &    185                           \\
L0Learn ($L_0 L_2$)   &       \bf 16.5                        &     4.64                         &   \bf 11                           \\
L0Learn ($L_0 L_1$)   &         16.7                       &     5.12                        &     15                          \\
ncvreg (MCP)     &        36.5                       &           4.85                  &     147                       \\ \bottomrule
\end{tabular}
\end{tabular}}
\caption{Training time (in seconds), out-of-sample MSE, and the corresponding support sizes for a variety of high-dimensional datasets. The training time is for obtaining a regularization path with 100 solutions.}
\label{table:timings}
\end{table}
\section{Conclusion}
We proposed new algorithms for Problem~\eqref{problem:intro}, based on a combination of cyclic coordinate descent and local combinatorial search, and studied their convergence properties.
Our algorithms are inspired by a hierarchy of necessary optimality conditions for Problem~\eqref{problem:intro}, with solutions higher up the hierarchy being of higher quality.
In terms of optimization performance, Algorithm~1 leads to better solutions and is faster than IHT and random CD.
Our local optimization algorithms (Algorithm~\ref{alg:mipcd}) often lead to further improvements over Algorithm~1. In many difficult settings, solutions from Algorithm~\ref{alg:mipcd} match those of global MIO solvers for Problem~\eqref{problem:intro}, while running much faster.

Our algorithms shed interesting insights onto the statistical properties of high-dimensional regression---in terms of variable selection, 
%(support recovery and model parsimony),
estimation error, prediction error vis-a-vis 
problem parameters ($n,p,\text{SNR},\beta^\dagger$ and $\Sigma$). 
There is no overall winner among the vanilla versions of Lasso, stepwise, or $L_0$, across different settings---modifications such as Problem~\eqref{problem:intro} or Relaxed Lasso~\cite{onbestsubset} seem necessary.
{In low signal settings (e.g., low SNR or small $n$), where recovery (in terms of a small estimation error or full support recovery) seems impossible, one can hope to get a good predictive model that is also sparse. In these regimes, $(L_0 L_2)$, Elastic Net, and ridge typically achieve the best predictive performance, with $(L_0 L_2)$ selecting much smaller support sizes.} 
We observe that estimators arising from Problem~\eqref{problem:intro} typically outperform the state-of-the-art sparse learning algorithms in terms of a combination of metrics (prediction, variable selection and estimation), across a wide range of settings; and promise to be an appealing alternative to the Relaxed Lasso~\cite{onbestsubset}. Our proposed algorithms allow us to uncover regimes (previously unseen due to computational limitations) where there are important differences between $L_0$-based estimators and existing popular algorithms (based on $L_1$, stepwise selection, IHT, etc.). 
We provide an open-source implementation of the algorithms through our toolkit \texttt{L0Learn}, which achieves up to a 3x speed-up when compared to competing toolkits.

\subsection*{Acknowledgements}
The authors would like to thank the editors and 
reviewers for their insightful suggestions that led to an improvement of the paper. The authors acknowledge research funding from the Office of Naval Research ONR-N000141512342, ONR-N000141812298 (Young Investigator Award), the National Science Foundation (NSF-IIS-1718258), and MIT. 

%\subsection*{Author Biographies}
%\textbf{Rahul Mazumder}

%\textbf{Hussein Hazimeh} is a PhD candidate at the MIT Operations Research Center. His main research interests lie at the intersection of optimization and machine learning. In particular, he is working on developing scalable algorithms for sparse and interpretable learning. 

%N00014-18-1-2298

\bibliographystyle{plainnat_my}
\small{\bibliography{ref}}

\begin{appendix}

\newpage

\section{Appendix: Proofs and Technical Details}\label{sec:appendix}

\subsection{Proof of Lemma \ref{lemma:stationarity}}

\begin{proof}
For any $d \in \mathbb{R}^d$, we will show that $F^{'}(\beta;d)$ is given by:
	\begin{align}
		F^{'}(\beta;d) & = \begin{cases}
		\langle \nabla_S f(\beta), d_S \rangle & \text{ if } d_{S^c} = 0 \\
		\infty & \text{ o.w. }
		\end{cases}
		\label{eq:directionalderF}
	\end{align}
Let $d$ be an arbitrary vector in $\mathbb{R}^p$. Then,
	\begin{align*}
		F^{'}(\beta;d) & = \liminf_{\alpha \downarrow 0} \left\{\frac{F(\beta + \alpha d) - F(\beta)}{\alpha} \right\}  \\
%						   & = \liminf_{\alpha \downarrow 0} \Big \{ \frac{f(\beta + \alpha d) - f(\beta)}{\alpha} +  	\frac{\lambda_0 \|\beta + \alpha d\|_0 - \lambda_0 \|\beta\|_0}{\alpha} \Big \}\\
						   & =\liminf_{\alpha \downarrow 0} \Big \{ \underbrace{\frac{f(\beta + \alpha d) - f(\beta)}{\alpha}}_{\text{Term I}} + \underbrace{\lambda_0 \sum_{i \in S} \frac{ \|\beta_i + \alpha d_i \|_0 - 1}{\alpha}}_{\text{Term II}}  + \underbrace{ \lambda_0 \sum_{j \notin S} \frac{ \|\alpha d_j \|_0}{\alpha}}_{\text{Term III}} \Big \}
                                   	\end{align*}
First we note that $\lim_{\alpha \downarrow 0} \text{Term II} = 0$ since for any $i \in S$, $\|\beta_i + \alpha d_i\|_0 = 1$ for sufficiently small $\alpha$. Suppose $d_{S^c} = 0$. Then, the continuity of $f$ implies that $\lim_{\alpha \downarrow 0} \text{Term I} = f^{'}(\beta_S;d_S) = \langle \nabla_S f(\beta), d_S \rangle$, where the second equality follows by observing that 
$\beta_S \to f(\beta_S)$ is continuously differentiable (in the neighborhood of $\beta_S$). Also, Term III $= 0$. Therefore, we have: 
$$F^{'}(\beta;d) = \lim_{\alpha \downarrow 0} \text{Term I} + \lim_{\alpha \downarrow 0} \text{Term II} = \langle \nabla_S f(\beta), d_S \rangle.$$ 
We now consider the case when $d_{S^c} \neq 0$. In this case, $\lim_{\alpha \downarrow 0} \text{Term III} = \infty$; and since the limit of Term I is bounded, we  have $F^{'}(\beta;d) = \infty$. Thus, we have shown that \eqref{eq:directionalderF} holds. From~\eqref{eq:directionalderF}, we have $F^{'}(\beta;d) \geq 0$ for all $d$ iff $\nabla_S f(\beta) = 0$.
\end{proof}

\begin{comment}
\subsection{Proof of Remark \ref{remark:localmin}}
By the continuity of $f$, there exists a positive scalar $\delta$ and a non-empty ball $R = \{ \beta \in \mathbb{R}^p \  | \  ||\beta^{*} - \beta|| < \delta \}$ such that for every $\beta \in R$, we have $|f(\beta^{*}) - f(\beta)| < \lambda_0$. Let $S = \text{Supp}(\beta^{*})$. We assume w.l.o.g. that $\delta$ is small enough so that if $i \in S$, then for every $\beta \in R$, we have $i \in \text{Supp}(\beta)$. For any $\beta \in R$, if $\beta_i \neq 0$ for some $i \notin S$, we have $(||\beta^{*}||_0 - ||\beta||_0) \leq -1$. This implies
$$F(\beta^{*}) - F(\beta)  = f(\beta^{*}) - f(\beta) + \lambda_0 (||\beta^{*}||_0 - ||\beta||_0) \leq|f(\beta^{*}) - f(\beta)| + \lambda_0 (-1) < \lambda_0 - \lambda_0 = 0.$$
Otherwise, if $\text{Supp}(\beta) = S$, then the stationarity of $\beta^{*}$ and convexity of $f$ imply that $f(\beta^{*}) \leq f(\beta)$ and consequently $F(\beta^{*}) \leq F(\beta)$. Therefore, for any $\beta \in R$, we have $F(\beta^{*}) \leq F(\beta)$.
\end{comment}

\subsection{Proof of Lemma \ref{lemma:univariate}}
\begin{proof}
Let $g(u)$ denote the objective function minimized in (\ref{eq:thresholdingmap}), i.e., 
$$g(u) := \frac{1 + 2\lambda_2}{2} \Big(u - \frac{\widetilde{\beta}^{*}_i}{1+2\lambda_2} \Big)^{2} + \lambda_1 |u| + \lambda_0 \mathds{1}[u \neq 0].$$
If $|\widetilde{\beta}^{*}_i| > \lambda_1$, then 
$\min_{u \neq 0} g(u)$ is attained by
$\widehat{u} = \frac{\sign(\widetilde{\beta}^{*}_i)}{1+ 2\lambda_2} (|\widetilde{\beta}^{*}_i|- \lambda_1)$ (this is the well-known soft-thresholding operator).
Now, $g(\widehat{u}) < g(0)$ is equivalent to $\frac{|\widetilde{\beta}^{*}_i| - \lambda_1}{1+2\lambda_2} > \sqrt{2\lambda_0 \over 1+2 \lambda_{2}}$. Hence, $\widehat{u}$ is the minimizer of $g(u)$ when $\frac{|\widetilde{\beta}^{*}_i| - \lambda_1}{1+2\lambda_2} > \sqrt{2\lambda_0 \over 1+2 \lambda_{2}}$. Both $\widehat{u}$ and $0$ are minimizers of $g(u)$ if $\frac{|\widetilde{\beta}^{*}_i| - \lambda_1}{1+2\lambda_2} = \sqrt{2\lambda_0 \over 1+2 \lambda_{2}}$.
Finally, when $\frac{|\widetilde{\beta}^{*}_i| - \lambda_1}{1+2\lambda_2} < \sqrt{2\lambda_0 \over 1+2 \lambda_{2}}$, the function $g(u)$ is minimized at $u=0$. 
\end{proof}

\subsection{Proof of Theorem 1}
The proof of the theorem is similar to the proofs in~\cite{blumensath2009-acha, mazumder2017subset,bertsimas2015best} (which consider the cardinality constrained version of Problem~\eqref{problem:main}).

\subsection{Proof of Lemma \ref{lemma:descent}}
\begin{proof}
If $\beta^{k}$ is the result of a non-spacer step then $F(\beta^{k}) \leq F(\beta^{k-1})$ holds by definition. If $\beta^{k}$ is obtained after a spacer step, then $f(\beta^{k}) \leq f(\beta^{k-1})$. Since a spacer step cannot increase the support size of $\beta^{k-1}$, this implies that $\|\beta^{k}\|_0 \leq \|\beta^{k-1}\|_0$, and thus $F(\beta^k) \leq F(\beta^{k-1})$. Since $F(\beta^{k})$ is non-increasing and bounded below (by zero), it must converge to some $F^{*} \geq 0$.
\end{proof}

\subsection{Proof of Lemma \ref{lemma:bdsuppsize}}
\begin{proof}
The result holds trivially if $p \leq n$. Suppose $p > n$. In the $(L_0)$ problem, Assumption~2 states that $F(\beta^{0}) \leq \lambda_0 n$. Since 
Algorithm~1 is a descent method (by Lemma \ref{lemma:descent}) we have $F(\beta^k) \leq \lambda_0 n$ for every $k$, which implies $f(\beta^k) + \lambda_0 \|\beta^k\|_0 \leq \lambda_0 n$ and hence, $\lambda_0 \|\beta^k\|_0 \leq \lambda_0 n$. Therefore, $\|\beta^k\|_0 \leq n$ for all $k$.
Similarly, for the $(L_0 L_1)$ problem, Assumption~2 and the descent property imply $F(\beta^k) \leq f(\beta^{\ell_1}) + \lambda_0 n$ which can be equivalently written as $f(\beta^k) - f(\beta^{\ell_1}) \leq \lambda_0 ( n - \|\beta^k\|_0)$. But the optimality of the lasso solution implies $f(\beta^k) - f(\beta^{\ell_1}) \geq 0$, which leads to $\|\beta^k\|_0 \leq n$.
\end{proof}

\subsection{Proof of Theorem~\ref{theorem:convergence}}
Before presenting the proof of Theorem \ref{theorem:convergence}, we present some necessary lemmas. First, we recall how the iterates are indexed by Algorithm, 1. If $\beta^{l}$ is obtained after performing a spacer step, then $\beta^{l-1}$ corresponds to a non-spacer step---by this time, 
a certain support has occurred for $Cp$ times. 
Suppose, $\beta^k$ denotes the current value of $\beta$ in Algorithm~1. If the next step is a non-spacer step, then $\beta^{k+1}$ is obtained from $\beta^k$ by updating a single coordinate. Otherwise, if the next step is a spacer step, then \emph{all} the coordinates inside the support of $\beta^{k}$ will be updated to get $\beta^{k+1}$.

The following lemma shows that the sequence generated by Algorithm \ref{alg:CD} is bounded. 
\smallskip
\begin{lemma}
\label{lemma:boundedness}
%Under Assumption \ref{assumption:Xlinindp} for the $(L_0)$ problem and no assumptions for the $(L_0L_1)$ and $(L_0 L_2)$ problems, 
The sequence $\{ \beta^k \}$ is bounded.
%generated by Algorithm \ref{alg:CD} 
\begin{proof}
%For the $(L_0 L_1)$ and $(L_0 L_2)$ problems, the sequence $\{\beta^{k}\}$ is a subset of the level set
%$
%M = \{ \beta \in \mathbb{R}^p \ | \  F(\beta) \leq F(\beta^{0})  \},
%$
For the $(L_0 L_1)$ and $(L_0 L_2)$ problems, for all $k$, $\beta^{k}$ belongs to the level set $G=\{ \beta \in \mathbb{R}^p \ | \  F(\beta) \leq F(\beta^{0}) \}$
where $\beta^{0}$ is an initial solution. Since in both cases $F(\beta)$ is coercive, $G$ is bounded and therefore, $\{\beta^k \}$ is bounded.

We now study the ($L_0$) problem. 
Firstly, if $p \leq n$, then the objective function for the  ($L_0$) problem is coercive (under Assumption \ref{assumption:Xlinindp}), and the previous argument used for $(L_0 L_1) / (L_0 L_2)$ applies. Otherwise, suppose that $p > n$. Recall that from Lemma \ref{lemma:bdsuppsize}, we have $\|\beta^{k}\|_0 \leq n $ for all $ k \geq 0$; and from Assumption~2, we have $F(\beta^0) \leq \lambda_0 n$. In addition, by Lemma \ref{lemma:descent}, we have  $F(\beta^k) \leq \lambda_0 n$ for every $k$.  Therefore, it follows that $\beta^{k} \in A$ where,  
%%%defined by 
\begin{equation*}
    A =  \bigcup\limits_{S \subseteq [p], |S| \leq n} A_{S}, \ \text{ and } \ 
A_{S} = \{ \beta \in \mathbb{R}^p \ | \  \frac{1}{2} \|y-X_S\beta_S\|^{2} \leq \lambda n, \  \beta_{S^{c}} = 0  \}.
\end{equation*}
Note that in every $A_S$, the only components of $\beta$ that might be non-zero are in $\beta_S$. %%which belongs to a level set of the function.
By Assumption \ref{assumption:Xlinindp}, the level set 
$\{ \beta_{S} | \frac{1}{2} \|y-X_S\beta_S\|^{2} \leq \lambda n\} \subseteq \mathbb{R}^{|S|}$ is bounded,
%But Assumption \ref{assumption:Xlinindp} implies that the latter function is coercive, 
which implies that $A_S$ is bounded. Since $A$ is the union of a finite number of bounded sets, it is also bounded.
\end{proof}
\end{lemma}

The next lemma characterizes the limit points of Algorithm \ref{alg:CD}.
\smallskip
\begin{lemma}
\label{lemma:spacer}
%\textcolor{red}{Let be $\{ \beta^k \}$ be the sequence of iterates generated by Algorithm \ref{alg:CD}.  Suppose that the initial conditions and Assumption \ref{assumption:Xlinindp} hold for the $(L_0)$ and $(L_0 L_1)$ problems.} 
Let $S$ be a support that is generated infinitely often by the non-spacer steps, and let $\{ \beta^{l} \}_{l \in L}$ 
%%where $L \subseteq \mathbb{N}$ 
be the sequence of spacer steps generated by $S$. Then, the following hold true:
\begin{enumerate}
\item There exists an integer $N$ such that for all $l \in L$ and $l \geq N$ we have Supp$(\beta^{l}) = S$.
\item There exists a subsequence of $\{ \beta^{l} \}_{l \in L}$ that converges to a stationary solution $\beta^{*}$, where, 
$\beta^*_{S}$ is the unique minimizer of $\min_{\beta_S} f(\beta_S)$ and $\beta^*_{S^c}=0$.
\item Every subsequence of $\{\beta^{k} \}_{k \geq 0}$ with support $S$ converges to $\beta^{*}$ (as in Part 2, above).
\item $\beta^{*}$ satisfies $|\beta^{*}_j| \geq \sqrt{2\lambda_0 \over 1+2\lambda_2}$ for every $j \in S$. %COMMENT: HH: Why estimate?
\end{enumerate}
\begin{proof}
{\bf{Part 1.)}} Since the spacer steps optimize only over the coordinates in $S$, no element from outside $S$ can be added to the support by the spacer step. Thus, for every $l \in L$ we have $\text{Supp}(\beta^{l}) \subseteq S$. 
We now show that strict containment is not possible 
using the method of contradiction. To this end, suppose $\text{Supp}(\beta^{l}) \subsetneq S$ occurs infinitely often; and let us consider some $l \in L$ at which this occurs. By Algorithm \ref{alg:CD}, the previous iterate $\beta^{l-1}$ has a support $S$,
%So for the chosen $l$, 
which implies $\|\beta^{l-1}\|_0 - \|\beta^{l}\|_0 \geq 1$. Moreover, from the definition of the spacer step we have $f(\beta^{l}) \leq f(\beta^{l-1})$. Therefore, we get $$F(\beta^{l-1}) - F(\beta^{l}) = \underbrace{f(\beta^{l-1}) - f(\beta^{l}) }_{\geq 0} + \lambda_0(\underbrace{\|\beta^{l-1}\|_0 - \|\beta^l\|_0}_{\geq 1} ) \geq \lambda_0.$$ Thus, every time the event $\text{Supp}(\beta^{l}) \subsetneq S$ occurs, the objective $F$ decreases by at least $\lambda_0$. This contradicts the fact that $F$ is lower bounded by $0$, which establishes the result.

{\bf{Part 2.)}}
% For any $l \in L$, the step at $l-1$ is a non-spacer step (which leads to an iterate with support $S$) and we always have $f(\beta^{l-1}) \geq f(\beta^{l})$. Thus, for $l \geq N$, the sequence $\{ \beta^l \}_{l \in L}$ is that obtained from running cyclic CD to minimize $f(\beta_S)$, where intermediate steps (between any two consecutive $l$'s in $L$) that decrease the objective are allowed. Such an algorithm is guaranteed to have stationary limit points (following the standard proof of convergence for cyclic CD \cite{bertsekas2016nonlinear}). Finally, Lemma \ref{lemma:bdsuppsize} and Assumption \ref{assumption:Xlinindp} (only needed for the $(L_0)$ and $(L_0 L_1)$ problems) imply that $\beta_{S} \mapsto f(\beta_S)$ is strongly convex and hence the algorithm is guaranteed to converge to a unique limit point $\beta^{*}$ (the minimizer of $f(\beta_S)$).
The proof follows the standard steps for proving the convergence of cyclic CD (e.g., as in \cite{bertsekas2016nonlinear}) --- we provide the proof for completeness. First, we introduce some additional notation for the proof. Fix some $l \in L$. By Algorithm~1, $\beta^{l-1}$ has a support $S$ which we assume (without loss of generality) to be $S = \{1,2,\dots, J\}$. We recall that to obtain $\beta^{l}$ from $\beta^{l-1}$, Algorithm~1 performs a spacer step---i.e, it starts from $\beta^{l-1}$ and updates every coordinate in $S$ via $T(\cdot,0,\lambda_1,\lambda_2)$. We denote the intermediate iterates generated sequentially   by the spacer step as: $\beta^{l,1}, \beta^{l,2}, \ldots, \beta^{l,J}$ where $\beta^{l,J} = \beta^{l}$. 

Since the support $S$ occurs infinitely often, we consider the infinite sequence $\{ \beta^{l,1} \}_{l \in L}$ (i.e., the sequence of intermediate spacer steps where coordinate $1$ is updated). By Lemma \ref{lemma:boundedness}, $\{ \beta^{l,1} \}_{l \in L}$ is bounded and therefore, there exists a further subsequence $\{ \beta^{l',1} \}_{l' \in L'}$ that converges to a limit point $\beta^{*}$ (say). We assume that for every $l' \in L'$ we have $l' \geq N$, which implies that $\beta^{l'-1}$, $\beta^{l',1}, \beta^{l',2}, \dots, \beta^{l',J}$ all have the support $S$ (this follows from Part 1 of this lemma). Next, we will show that $\beta^{l,2}$ also converges to $\beta^{*}$.

Fix some $l' \in L'$. Then, we have:
\begin{equation}\label{eqn-suff-decrease1}
F(\beta^{l',1}) - F(\beta^{l',2}) = f(\beta^{l',1}_S) - f(\beta^{l',2}_S)   \geq \frac{1+2\lambda_2}{2} (\beta^{l',1}_2 - \beta^{l',2}_2)^2
\end{equation}
where the last inequality follows by replacing $\beta^{l',2}_2$ with the expression given by the thresholding map in (\ref{eq:thresholding}) and simplifying. By Lemma \ref{lemma:descent},  $\{ F(\beta^k) \}$ converges; so taking the limit as $l' \to \infty$ in~\eqref{eqn-suff-decrease1} we get
%$0 \geq \lim_{l' \to \infty} (\beta^{l',1}_2 - \beta^{l',2}_2)^2$; i.e., 
$ \beta^{l',1}_2 - \beta^{l',2}_2 \rightarrow 0$ as $l' \to \infty$. Since $\beta^{l',1}  \to \beta^{*}$, we conclude that $\beta^{l',2}_2 \to \beta^{*}_2$. Therefore, $\beta^{l',2} \to \beta^{*}$. The same argument applies to $\beta^{l',i}$ and $\beta^{l',i+1}$ for every $i \in \{2,3,\dots,J-1\}$. Therefore, we conclude that for every $i \in \{1,2,\dots,J\}$ we have $\beta^{l',i} \to \beta^{*}$. 

Let $k', l' \in L'$ be such that $k' > l'$, then $f(\beta^{k'}) \leq f(\beta^{l',1}) \leq f(\beta_1, \beta^{l'}_2, \beta^{l'}_3, \dots )$ for any $\beta_1 \in \mathbb{R}$. As $k', l' \to \infty$, we get $f(\beta^{*}) \leq f(\beta_1, \beta^{*}_2, \beta^{*}_3, \dots )$ for any $\beta_1$. The same result applies to all other coordinates in $j$, which implies that $0 \in \partial f(\beta^{*}_1, \beta^{*}_2, \beta^{*}_2, \dots)$ (where, $\partial f(\cdot)$ denotes subgradient) and consequently $\beta^{*}$ is a stationary solution for $\min_{\beta_S} f(\beta_S)$.
% Now, by the definition of the spacer step, for every $i \in [J]$ and every $l' \in L^{'}$; we have:
% \begin{align}
% \beta^{l',i}_i =  T(\widetilde{\beta}_i, 0, \lambda_1, \lambda_2) = \sign(\widetilde{\beta}^{l'}_i) \frac{|\widetilde{\beta}^{l'}_i| - \lambda_1}{1 + 2\lambda_2}
% \end{align}
% where, $ \widetilde{\beta}^{l'}_i = \langle y - \sum_{j \neq i} X_j \beta^{l',i}_j , X_i \rangle $ and $|\widetilde{\beta}^{l'}_i| \geq \lambda_1$. By taking $l' \to \infty$ in the above, 
%we get $\widetilde{\beta}^{*}_i = \lim_{l' \to \infty} \widetilde{\beta}^{l'}_i = \langle y - \sum_{j \neq i} X_j \beta^{*}_j , X_i \rangle$ where $|\widetilde{\beta}^{*}_i| \geq \lambda_1$.
% Therefore, for every $i \in [J]$, we have
% \begin{align}
% \beta^{*}_i = \lim_{l' \to \infty} \beta^{l',i}_i = 
% \begin{cases}
%  \sign(\widetilde{\beta}^{*}_i) \frac{|\widetilde{\beta}^{*}_i| - \lambda_1}{1 + 2\lambda_2} & \text{if} \ \lambda_1 > 0 \\
% \frac{\widetilde{\beta}^{*}_i}{1 + 2\lambda_2}  & \text{if} \  \lambda_1 = 0
% \end{cases}
% \label{eq:betaistar}
% \end{align}
% \begin{comment}
% \lim_{l' \to \infty} \sign(\widetilde{\beta}^{l'}_i) \frac{|\widetilde{\beta}^{l'}_i| - \lambda_1}{1 + 2\lambda_2}
% \end{comment}
% where we used the continuity of the $\sign$ function at $\widetilde{\beta}^{*}_i$ (note, $|\widetilde{\beta}^{*}_i| \geq \lambda_1$) in the case of $\lambda_1 >0$. 
Finally, Lemma \ref{lemma:bdsuppsize} and Assumption \ref{assumption:Xlinindp} for the $(L_0)$ and $(L_0 L_1)$ problems imply that $\beta_{S} \mapsto f(\beta_S)$ is strongly convex and has a unique minimizer --- hence $\beta^{*}$ is unique.

{\bf{Part 3.)}} By Part 2, there exists a subsequence $\{ \beta^{l'} \}_{l' \in L'}$ converging to $\beta^{*}$. By Part 1, for every $l' \geq N$ we have
$F(\beta^{l'}) = f(\beta^{l'}) + \lambda_0 |S|$. Taking $l' \to \infty$ and using the continuity of $f(\beta)$ we get:
\begin{align*}
\lim_{l' \to \infty} F(\beta^{l'}) = f(\beta^{*}) + \lambda_0 |S|.
\end{align*}
Now, consider any subsequence $\{ \beta^{k'} \}_{k' \in K'}$, where $K' \subseteq \{0,1,2,\dots\}$, such that the non-spacer steps in $K'$ have a support $S$. We will establish convergence of $\{ \beta^{k'} \}_{k' \in K'}$ via the method of contradiction.
To this end, suppose $\{ \beta^{k'} \}_{k' \in K'}$ has a limit point $\widehat{\beta}$ which is not equal to $\beta^{*}$. Then, there exists a subsequence 
$\{ \beta^{k''} \}_{k'' \in K''}$ (with $K'' \subseteq K'$) which converges to $\widehat{\beta}$. Then, for every $k'' \geq N$, we have $F(\beta^{k''}) = f(\beta^{k''}) + \lambda_0 |S|$. Taking the limit as $k'' \to \infty$ we get:
\begin{align*}
\lim_{k'' \to \infty} F(\beta^{k''}) = f(\widehat{\beta}) + \lambda_0 |S|.
\end{align*}
From Lemma \ref{lemma:descent}, we have $\lim_{l' \to \infty} F(\beta^{l'}) = \lim_{k'' \to \infty} F(\beta^{k''})$. This implies $f(\beta^{*}) = f(\widehat{\beta})$ and in particular, $f(\beta^{*}_S) = f(\widehat{\beta}_S)$ (since the supports of both limits points are a subset of $S$). 
%\textcolor{red}{But $\beta^{*}$ is the unique minimizer of $f(\beta_S)$ (from statement 2), which leads to a contradiction.-- this is confusing, because $\beta_{S}$ and $\beta^*$ have different dimensions. I think you mean that $\widehat{\beta}$ and $\beta^*$ have the same support, hence by strong convexity they should be the same??} Thus, $\{ \beta^{k'} \}$ has a unique limit point %COMMENT: Correct, thanks!
%$\beta^{*}$ as its only limit point, which along with boundedness (from Lemma \ref{lemma:boundedness}), implies that 
However, by Part 2, we know that $\beta_S^{*}$ is the unique minimizer of $\min_{\beta_{S}} f(\beta_S)$---which leads to a contradiction. Hence, we conclude that 
$ \beta^{k'} \to \beta^{*}$ as $k' \rightarrow \infty$.

{\bf{Part 4.)}} Let $l_1$ and $l_2$ be the indices of any two consecutive spacer steps generated by the support $S$. Recall from Algorithm \ref{alg:CD} that $C \geq 1$; and the support $S$ must appear in $C p$ non-spacer steps between $l_1$ and $l_2$.
Fix some $i \in S$ . We will show that there exists a non-spacer step with index $k'$ such that $l_1<k'<l_2$, $\text{Supp}(\beta^{k'}) = S$, and $|\beta^{k'}_i| \geq \sqrt{\frac{2\lambda_0}{1+2\lambda_2}}$. We proceed by contradiction. To this end, suppose that such an index does not exist --- i.e., every non-spacer step that updates coordinate $i$ thresholds it to $0$.
%\textcolor{red}{Denote by $k_1$ the index of the first non-spacer step that updates coordinate $i$ between $l_1$ and $l_2$.}
%\textcolor{red}{RM: rewording the above to: Let $k_1$ denote the iteration index of the first non-spacer step between $l_1$ and $l_2$ that updates coordinate $i$.???} COMMENT: HH: yes your version is clearer thanks!
Let $k_1$ denote the iteration index of the first non-spacer step between $l_1$ and $l_2$ that updates coordinate $i$.
In the $p$ coordinate updates after $l_1$, coordinate $i$ must be updated once, which implies $k_1 - l_1 \leq p$.
Since at iteration $k_1$, coordinate $i$ is set to $0$, the support $S$ can appear at most $p-1$ times between $l_1$ and $k_1$. Moreover, between $k_1$ and $l_2$, $S$ appears $0$ times --- this is because, coordinate $i$ never gets thresholded to a non-zero value by a non-spacer step. Therefore, $S$ appears for at most $p-1$ times between $l_1$ and $l_2$ -- this contradicts the fact that $l_2$ is the index after which $S$ appears in $Cp$ non-spacer steps. Therefore, we conclude that there exists an index $k'$ such that $l_1<k'<l_2$, $\text{Supp}(\beta^{k'}) = S$, and $|\beta^{k'}_i| \geq \sqrt{\frac{2\lambda_0}{1+2\lambda_2}}$.
%\textcolor{red}{RM: Is this  $|\beta^{k}_i| $ or  $|\beta^{k'}_i|$??}
%\textcolor{red}{RM: Ending is abrupt. What have we achieved here? please mention; and then we will go to the next paragraph} \textcolor{blue}{HH: Fixed}

%\textcolor{red}{RM: in the proof of Part 4, the indices/iteration indices and coordinate indices are not clear.} COMMENT: HH: I think this is fixed now

Let us now fix some $i \in S$. 
By considering the infinite sequence of spacer steps generated by $S$ and applying the result we proved above to every two consecutive spacer steps, we can see that there exists an infinite subsequence $\{ \beta^{k'} \}$ of non-spacer iterates where, for every $k'$, we have $\text{Supp}(\beta^{k'}) = S$ and $|\beta^{k'}_i| \geq \sqrt{\frac{2\lambda_0}{1+2\lambda_2}}$. By Part 3 of this lemma, $\{ \beta^{k'} \}$ converges to the stationary solution $\beta^{*}$. Taking the limit $k \rightarrow \infty$ in inequality: $|\beta^{k}_i| \geq \sqrt{\frac{2\lambda_0}{1+2\lambda_2}}$, we conclude that $|\beta^{*}_i| \geq \sqrt{\frac{2\lambda_0}{1+2\lambda_2}}$.
\end{proof}
\end{lemma}

Lemma~\ref{lemma:infoften} shows that the support corresponding to any limit point of $\{\beta^k\}$ appears infinitely often.
%%%in the sequence $\{\beta^k\}$.
\smallskip
\begin{lemma}
\label{lemma:infoften}
%\textcolor{red}{Let be $\{ \beta^k \}_{k\geq0}$ be the sequence of iterates generated by Algorithm \ref{alg:CD}. Suppose that the initial conditions and Assumption \ref{assumption:Xlinindp} hold for the $(L_0)$ and $(L_0 L_1)$ problems.} 
Let $B$ be a limit point of $\{\beta^k\}$ with $\text{Supp}(B) = S$, then $\text{Supp}(\beta^{k}) = S$ for infinitely many $k$.
%Suppose there exists a limit point $B$ of $\{\beta^{k}\}$ with $\text{Supp}(B) = S$. Then, $S$ appears infinitely often, i.e. $\text{Supp}(\beta^{k}) = S$ for infinitely many $k$. 
\end{lemma}
\begin{proof}
We prove this result by using a contradiction argument. To this end, suppose support $S$ occurs only finitely many times. Since there are only finitely many supports, there is a
support $S'$ with  $S'\neq S$; and a
subsequence $\{ \beta^{k'} \}$ of $\{\beta^k\}$ which satisfies: $\text{Supp}(\beta^{k'}) = S'$ for all $k'$; 
and $\beta^{k'} \rightarrow B$ as $k'\rightarrow \infty$. However, this is not possible 
by Part 3 of Lemma~\ref{lemma:spacer}.
\end{proof}

\begin{comment}
In the next lemma, we give a characterization of the limit points generated by Algorithm \ref{alg:CD}.
\begin{lemma}
\label{lemma:stationary}
Let $\{\beta^k\}$ be the sequence of iterates generated by Algorithm (\ref{alg:CD}).  Suppose that the initial conditions and Assumption (\ref{assumption:Xlinindp}) hold for the $(L_0)$ and $(L_0 L_1)$ problems. Every limit point $B$ of  $\{\beta^{k}\}$ is a stationary solution and for every $i$ in $\text{Supp}(B)$ we have $|B_i| \geq \sqrt{\frac{2\lambda_0}{1+2\lambda_2}}$.
\label{eq:b2j}
\end{lemma}
\begin{proof}
By the same argument used in the proof of Lemma (\ref{lemma:infoften}), there exists an infinite subsequence $\{\beta^{k'}\}$ with support $\text{Supp}(B)$ that converges to $B$. Lemma (\ref{lemma:spacer}) then implies that $B$ is a stationary solution satisfying $|B_i| \geq \sqrt{\frac{2\lambda_0}{1+2\lambda_2}}$ for every $i$ in $\text{Supp}(B)$.
\end{proof}
\end{comment}

Lemma~\ref{lemma:xj} is technical and will be needed in the proof of convergence of Algorithm~1.
\smallskip
\begin{lemma}
\label{lemma:xj}
%\textcolor{red}{$\{\beta^k\}$ be the sequence of iterates generated by Algorithm \ref{alg:CD}.  Suppose that the initial conditions and Assumption \ref{assumption:Xlinindp} hold for the $(L_0)$ and $(L_0 L_1)$ problems.} 
Let $B^{(1)}$ and $B^{(2)}$ be two limit points of the sequence $\{\beta^{k}\}$, with supports $S_1$ and $S_2$, respectively. Suppose that $S_2 = S_1 \cup \{j\} $ for some $j \notin S_1$. Then, exactly one of the following holds:
\begin{enumerate}
    \item If there exists an $i \in S_1$ such that $\langle X_i, X_j \rangle \neq 0$, then $|B^{(2)}_j| > \sqrt{\frac{2\lambda_0}{1+2\lambda_2}}$.
    %then there exists a $\delta_{S_1,S_2}>0$ \textcolor{red}{RM: what does this depend upon, in words}  such that $|B^{(2)}_j| = \sqrt{\frac{2\lambda_0}{1+2\lambda_2} + \delta_{S_1,S_2}}$.
    
    \item Otherwise, if $\langle X_i, X_j \rangle = 0$ for all $i \in S_1$, then $|B^{(2)}_j| = \sqrt{2\lambda_0 \over 1+2\lambda_2}$. Furthermore, for any $\beta \in \mathbb{R}^p$ with $\text{Supp}(\beta) = S_2$, we have $|T(\widetilde{\beta}_j, \lambda_0, \lambda_1, \lambda_2)| = \sqrt{2\lambda_0 \over 1+2\lambda_2}$.
\end{enumerate}
\begin{proof}
We first derive an useful expression for $B^{(2)}_j$, which will help us establish parts 1 and 2 of this lemma. By Lemma~\ref{lemma:descent}, $F(\beta^{k})$ converges to a finite non-negative limit $F^{*}$. Lemmas \ref{lemma:spacer}  and \ref{lemma:infoften} imply that there is a subsequence $\{\beta^{k'}\}_{k' \in K'}$ that converges to $B^{(1)}$ and satisfies $\text{Supp}(\beta^{k'}) = S_1$ for every $k'$. 
%Thus, $F(\beta^{k'}) = f(\beta^{k'}) + |S_1|$ for every $k'$. Taking the limit as $k' \to \infty$ on 
As $k\to \infty$, we get: $F^{*} = f(B^{(1)}) + |S_1| = F(B^{(1)})$. Similarly, for $B^{(2)}$ 
%to conclude that 
we have $F^{*} = F(B^{(2)})$. Therefore,  $F(B^{(1)}) = F(B^{(2)})$, which is equivalent to 
$$
f(B_{S_1}^{(1)}) + \lambda_0 \|B_{S_1}^{(1)}\|_0  = f(B_{S_2}^{(2)}) + \lambda_0 \|B_{S_2}^{(2)}\|_0.
$$
Since $\|B_{S_2}^{(2)}\|_0  = \|B_{S_1}^{(1)}\|_0 + 1$, we can simplify the above to obtain:
\begin{equation}\label{equation:difflambda}
f(B_{S_1}^{(1)}) - f(B_{S_2}^{(2)}) = \lambda_0.
\end{equation}
The term $f(B_{S_2}^{(2)})$ can be rewritten as follows (using elementary algebraic manipulations)
\begin{align}
f(B_{S_2}^{(2)}) & = \frac{1}{2} \|y - X_{S_2} B^{(2)}_{S_2}\|^2 + \lambda_{1} \|B_{S_2}^{(2)}\|_1 + \lambda_{2} \|B_{S_2}^{(2)}\|^2_2 \nonumber \\ & = \frac{1}{2} \|y - X_{S_1} B^{(2)}_{S_1} - X_j B^{(2)}_{j}\|^2  + \lambda_{1} \|B_{S_1}^{(2)}\|_1 +  \lambda_{1} |B^{(2)}_{j}| +  \lambda_{2} \|B_{S_1}^{(2)}\|^2_2  + \lambda_{2} (B^{(2)}_{j})^2  \nonumber \\
& = \Big( \frac{1}{2} \|y - X_{S_1} B^{(2)}_{S_1}\|^2 + \lambda_{1} \|B_{S_1}^{(2)}\|_1+\lambda_{2} \|B_{S_1}^{(2)}\|^2_2 \Big) \nonumber \\ &  - \langle y - X_{S_1} B^{(2)}_{S_1}, X_j \rangle B^{(2)}_{j} + \frac{1}{2} \|X_j\|^2 {B^{(2)}_{j}}^2 +  \lambda_{1} |B^{(2)}_{j}| + \lambda_{2} (B^{(2)}_{j})^2  \nonumber \\ 
& \label{equation:xj} = f(B_{S_1}^{(2)}) - \langle y - X_{S_1} B^{(2)}_{S_1}, X_j \rangle B^{(2)}_{j} + \frac{1}{2} \|X_j\|^2 (B^{(2)}_{j})^2 +  \lambda_{1} |B^{(2)}_{j}| + \lambda_{2} (B^{(2)}_{j})^2.
\end{align}
From Lemma \ref{lemma:spacer} we know that $B^{(2)}$ is a stationary solution. Using the characterization of stationary solutions in (\ref{eq:stationary}) and rearranging the terms, we get: 
\begin{equation}
\begin{myarray}[1.5]{rcl}
%%\begin{cases}
\langle y - X_{S_1} B^{(2)}_{S_1}, X_j \rangle &=& (1 + 2\lambda_2) B^{(2)}_{j}  + \lambda_1 \text{sign}(\langle y - X_{S_1} B^{(2)}_{S_1}, X_j \rangle) \\
|\langle y - X_{S_1} B^{(2)}_{S_1}, X_j \rangle| &>& \lambda_1.
%%\end{cases}
\end{myarray}
\end{equation}
Multiplying the first equation in the above by $B^{(2)}_{j}$ and using that the fact $\langle y - X_{S_1} B^{(2)}_{S_1}, X_j \rangle$ and $B^{(2)}_{j}$ have the same sign (which is evident from the system above), we arrive at
\begin{equation}
	\langle y - X_{S_1} B^{(2)}_{S_1}, X_j \rangle B^{(2)}_{j} = (1+2\lambda_{2}) (B^{(2)}_{j})^2 + \lambda_{1} |B^{(2)}_{j}|.
\end{equation}
Plugging in the above expression in the second term on the r.h.s of (\ref{equation:xj}) and using the fact that $\|X_j\|^2=1$ we get
\begin{equation}\label{line-diff-obj-1}
f(B_{S_2}^{(2)}) = f(B_{S_1}^{(2)}) - \frac{1+2\lambda_2}{2} (B^{(2)}_{j})^2.
\end{equation}
Substituting~\eqref{line-diff-obj-1} into equation (\ref{equation:difflambda}) and rearranging terms, we arrive at
\begin{equation}\label{equation:bj2}
|B^{(2)}_{j}| = \sqrt{ {2\lambda_0 \over 1+2\lambda_2} + {2 \over 1+2\lambda_2} \Big( f(B_{S_1}^{(2)}) - f(B_{S_1}^{(1)})  \Big) }.
\end{equation}
\noindent \textbf{Part 1.)} We consider Part 1, where there exists an $i \in S_1$ such that $\langle X_i, X_j \rangle \neq 0$. By Lemma \ref{lemma:spacer} we have that $B^{(1)}$ is a stationary solution. 
%\textcolor{red}{RM: why do you have $\nabla$ when $f$ is not differentiable? consequently $\nabla f(B_{S_1}^{(1)}) = 0$. -- HH: I meant the grad on the support. RM: I think we can just drop the part in red. HH: Sounds good}.
Thus, $B_{S_1}^{(1)} \in \argmin_{\beta_{S_1}}f(\beta_{S_1})$, and the following holds
\begin{equation}\label{equation:ineqb1opt}
f(B^{(1)}_{S_1}) \leq f(B^{(2)}_{S_1}).  
\end{equation}
We will show the inequality above is strict. To this end, suppose that~\eqref{equation:ineqb1opt} holds with equality. Lemma \ref{lemma:infoften} implies that $S_1$ appears in the sequence of iterates. But the function $f(\beta_{S_1})$ is strongly convex (this is trivial for ($L_0 L_2$) and holds due to Assumption \ref{assumption:Xlinindp} and Lemma \ref{lemma:bdsuppsize} for the ($L_0$) and ($L_0 L_1$) problems). Thus, $B_{S_1}^{(1)}$ is the  unique minimizer of $f(\beta_{S_1})$. Therefore, it must be the case that $B^{(1)}_{S_1} = B^{(2)}_{S_1}$, and in particular $B^{(1)}_i = B^{(2)}_i.$
%\textcolor{red}{RM: what is $i$ here? $i \in S_1$? HH: $i$ is the index defined in the first sentence after Part 1.) Oh right! ok -- thanks!}
By the characterization of stationary solutions in (\ref{eq:stationary}) we have:
\begin{equation}\label{equality-of-signs}
\begin{aligned}
\text{sign}(\langle y - X_{S_1 \backslash \{ i \} } B^{(1)}_{S_1 \backslash \{ i \}}, X_i \rangle) \overbrace{\frac{|\langle y - X_{S_1 \backslash \{ i \} } B^{(1)}_{S_1 \backslash \{ i \}}, X_i \rangle| - \lambda_1}{1+2\lambda_2} }^{\geq 0} \\
= \text{sign}(\langle y - X_{S_2 \backslash \{ i \} } B^{(2)}_{S_2 \backslash \{ i \}}, X_i \rangle) \underbrace{\frac{|\langle y - X_{S_2 \backslash \{ i \} } B^{(2)}_{S_2 \backslash \{ i \}}, X_i \rangle| - \lambda_1}{1+2\lambda_2} }_{\geq 0}
\end{aligned}
\end{equation}
Observing that the two sign terms in~\eqref{equality-of-signs} are equal, we can simplify the above to:
\begin{equation}\label{equality-signs-1}
\begin{aligned}
&\langle y - X_{S_1 \backslash \{ i \} } B^{(1)}_{S_1 \backslash \{ i \}}, X_i \rangle &=& \langle y - X_{S_2 \backslash \{ i \} } B^{(2)}_{S_2 \backslash \{ i \}}, X_i \rangle \\
%%&\langle y - X_{S_1 \backslash \{ i \} } B^{(1)}_{S_1 \backslash \{ i \}}, X_i \rangle 
&&=& \langle y - X_j B^{(2)}_j - X_{S_1 \backslash \{ i \} } B^{(2)}_{S_1 \backslash \{ i \}}, X_i \rangle
\end{aligned}
\end{equation}
where, the second line in~\eqref{equality-signs-1} follows by noting $X_{S_2 \backslash \{ i \} } B^{(2)}_{S_2 \backslash \{ i \}} = X_j B^{(2)}_j + X_{S_1 \backslash \{ i \} } B^{(2)}_{S_1 \backslash \{ i \}}.$
Substituting $B^{(2)}_{S_1} = B^{(1)}_{S_1}$ in~\eqref{equality-signs-1} and simplifying, we get
%\begin{align}
$\langle X_j, X_i \rangle = 0$,
%\end{align}
which contradicts the assumption in Part 1.
%that $\langle X_i, X_j \rangle \neq 0$. 
Thus, we have established that inequality (\ref{equation:ineqb1opt}) is strict. Using this result in~(\ref{equation:bj2}), we conclude that:
$|B^{(2)}_{j}| > \sqrt{\frac{2\lambda_0}{1+2\lambda_2}}$.
%$|B^{(2)}_{j}| = \sqrt{\frac{2\lambda_0}{1+2\lambda_2} + \delta_{S_1,S_2}}$ with 
%statement in the proof with 
%$\delta_{S_1,S_2}={2 \over 1+2\lambda_2} \left( f(B_{S_1}^{(2)}) - f(B_{S_1}^{(1)})  \right) >0$.
%that there exists a positive constant $\delta_{S_1,S_2}$ such that $|B^{(2)}_{j}| = \sqrt{\frac{2\lambda_0}{1+2\lambda_2} + \delta_{S_1,S_2}}$.
%\textcolor{red}{RM: Should we keep this as $|B^{(2)}_{j}| = \sqrt{2\lambda_0 \over 1+2\lambda_2} + \delta_{S_1,S_2}$ or $|B^{(2)}_{j}| = \sqrt{\frac{2\lambda_0}{1+2\lambda_2} + \delta_{S_1,S_2}}$?? HH: I think both should be fine but the current version is the same as the statement of the lemma}

\textbf{Part 2.)}~We now consider the case where $\langle X_i, X_j \rangle = 0$ for all $i \in S_1$. 
In this case, the optimization problem $\min_{\beta_{S_2}} f(\beta_{S_2})$ separates into
optimization w.r.t the variables $\beta_{S_1}$ and 
$\beta_{j}$. % COMMENT: HH: should we say independent optimization problems?...
Note that $B^{(2)}_{S_1}$ and $B^{(1)}_{S_1}$ are both minimizers of 
$\min_{\beta_{S_1}} f(\beta_{S_1})$; and hence $f(B_{S_1}^{(2)}) = f(B_{S_1}^{(1)})$. Thus, from~\eqref{equation:bj2} we get $|B^{(2)}_{j}| =\sqrt{2\lambda_0 \over 1+2\lambda_2}$. Finally, we note that for any $\beta \in \mathbb{R}^p$ such that $\text{Supp}(\beta) = S_2$, we have that $T(\widetilde{\beta}_j, \lambda_0, \lambda_1, \lambda_2) = B^{(2)}_{j}$. This completes the proof.
\end{proof}
\end{lemma}

The following establishes a lower bound on the decrease in objective value, when a non-zero coordinate is set to zero during Algorithm~1.

\smallskip

\begin{lemma} \label{lemma:Fdiffdrop}
Let $\beta^{k}$ be an iterate of
Algorithm~1 
%has obtained an iterate $\beta^{k}$ 
with $\beta^{k}_j \neq 0$ for some $j \in [p]$. Let $\beta^{k+1}$ correspond to a non-spacer step which updates coordinate $j$ to $0$, i.e., $\beta^{k+1}_j = 0$. Then, the following holds:
\begin{equation}
F(\beta^{k}) - F(\beta^{k + 1})   \geq \frac{1+2\lambda_2}{2} \left( |\beta^{k}_j| - \sqrt{\frac{2\lambda_0}{1+2\lambda_2}} \right)^2.
\end{equation}
\end{lemma}
\begin{proof}
$F(\beta^{k}) - F(\beta^{k + 1})$ can be simplified by noting that $\beta^{k}_i = \beta^{k + 1}_i$ for all $i \neq j$ and $\beta^{k + 1}_j = 0$:
%\vspace{-0.8cm}
\begin{align}
F(\beta^{k}) - F(\beta^{k + 1}) & = - \widetilde{\beta}^{k}_j \beta^{k}_j + \frac{1 + 2\lambda_2}{2} (\beta^{k}_j)^2 + \lambda_0 + \lambda_1 |\beta^{k}_j| \nonumber \\
								 & \geq - |\widetilde{\beta}^{k}_j| |\beta^{k}_j| + \frac{1 + 2\lambda_2}{2} (\beta^{k}_j)^2 + \lambda_0 + \lambda_1 |\beta^{k}_j| \nonumber  \\
								  & \geq - |\beta^{k}_j| (|\widetilde{\beta}^{k}_j| - \lambda_1 ) + \frac{1 + 2\lambda_2}{2} (\beta^{k}_j)^2 + \lambda_0, \label{eq:diffFsconv}
\end{align}
where $\widetilde{\beta}^{k}_j = \langle y - \sum_{i \neq j} X_j \beta^{k}_i , X_j \rangle  $. Since $\beta^{k+1}$ is a non-spacer step which sets coordinate $j$ to $0$, the definition of the thresholding operator (\ref{eq:thresholding}) implies $|\widetilde{\beta}^{k}_j| - \lambda_1 < \sqrt{2 \lambda_0 (1+2\lambda_2)}$. Plugging this bound into (\ref{eq:diffFsconv}) and factorizing, we arrive to the result of the lemma.
\end{proof}

%%%%%%%~~~~~~~~~~~~~~~~~~~~~~~

\paragraph{Proof of Theorem~\ref{theorem:convergence}} Finally, we present the proof of Theorem \ref{theorem:convergence} below.
\begin{proof}:
%By Lemma \ref{lemma:boundedness}, $\{ \beta^{k} \}$ is bounded so it must have at least one limit point. 
Let $B$ be a limit point of $\{ \beta^{k} \}$ with the largest support size and denote its support by $S$.  We will show that $ \beta^{k}  \to B$ as $k \rightarrow \infty$.
%\textcolor{red}{RM: do not remove! We will show that $ \beta^{k}  \to B$ as $k \rightarrow \infty$.}
%\textcolor{red}{RM: HH to fix. -- HH done fixing this.}

\textbf{Part 1.)} By Lemma \ref{lemma:infoften}, there is a subsequence $\{ \beta^{r} \}_{r \in R}$ of $\{\beta^k\}$ 
%where $R = \{k \ | \ \text{Supp}(\beta^k) = S \}$, 
which satisfies: $\text{Supp}(\beta^r) = S$ for all $r$
and $\beta^{r} \to B$ (as $r \rightarrow \infty$). 
By Lemma~\ref{lemma:spacer}, there exists an integer $N$ such that for every $r \geq N$, if $r+1$ is a spacer step then $\text{Supp}(\beta^{r+1}) = \text{Supp}(\beta^{r})$. In what follows, we assume that $r \geq N$. Let $j$ be any element in $S$. We will show that there exists an integer $N_j$ such that for every $r \geq N_j$, we have $j \in \text{Supp}(\beta^{r+1})$. We show this by contradiction. To this end, let $j \notin \text{Supp}(\beta^{r+1})$ for infinitely many values of $r$. Hence, there is a further subsequence $\{ \beta^{r'} \}_{r' \in R'}$ of $\{ \beta^{r} \}_{r \in R}$
%\to B$ 
(with $R' \subseteq R$) such that
%for every $r'\geq N$ 
$\text{Supp}(\beta^{r'+1}) = S \setminus \{j\}$.
%%%and $r' \geq N$.
For every $r' \in R'$, note that $r' + 1$ is a non-spacer step (since $r' \geq N$). Therefore, applying Lemma~\ref{lemma:Fdiffdrop} with $k = r'$, we get:
\begin{comment}
\begin{align}
F(\beta^{r'}) - F(\beta^{r' + 1}) & = - \widetilde{\beta}^{r'}_j \beta^{r'}_j + \frac{1 + 2\lambda_2}{2} (\beta^{r'}_j)^2 + \lambda_0 + \lambda_1 |\beta^{r'}_j| \nonumber \\
								 & \geq - |\widetilde{\beta}^{r'}_j| |\beta^{r'}_j| + \frac{1 + 2\lambda_2}{2} (\beta^{r'}_j)^2 + \lambda_0 + \lambda_1 |\beta^{r'}_j| \nonumber  \\
								  & \geq - |\beta^{r'}_j| (|\widetilde{\beta}^{r'}_j| - \lambda_1 ) + \frac{1 + 2\lambda_2}{2} (\beta^{r'}_j)^2 + \lambda_0 \label{eq:diffFsconv}
\end{align}
For every $r' \in R'$, it must be the case that $r' + 1$ is a non-spacer step (since $r' \geq N$). Thus, the thresholding operator defined in (\ref{eq:thresholding}) sets $\beta^{r' + 1}_j = 0$, which implies $|\widetilde{\beta}^{r'}_j| - \lambda_1 \leq \sqrt{2 \lambda_0 (1+2\lambda_2)}$. Using this bound into (\ref{eq:diffFsconv}) and simplifying, we get
\end{comment}
\begin{equation}\label{suff-decr-thm-2}
F(\beta^{r'}) - F(\beta^{r' + 1})   \geq \frac{1+2\lambda_2}{2} \left( |\beta^{r'}_j| - \sqrt{\frac{2\lambda_0}{1+2\lambda_2}} \right)^2.
\end{equation}
Taking $r' \to \infty$ in~\eqref{suff-decr-thm-2} and using the convergence of 
$\{F(\beta^k)\}$ (by Lemma \ref{lemma:descent}), we conclude
\begin{equation}\label{eq:bjexact}
|B_j| = \lim_{r' \to \infty} |\beta^{r'}_j| = \sqrt{\frac{2\lambda_0}{1+2\lambda_2}}. 
\end{equation}
Since $\text{Supp}(\beta^{r' + 1}) = S \setminus \{j\}$ for every $r'$, Lemma \ref{lemma:spacer} implies that $\{ \beta^{r' + 1} \}$ converges to a limit point, which we denote by $\widehat{B}$. If $\langle X_i, X_j \rangle = 0$ for all $i \in S \setminus \{j\}$, then Lemma \ref{lemma:xj} (part 2) implies $j \in \text{Supp}(\beta^{r'+1})$, which contradicts the definition of $\{ \beta^{r'} \}_{r' \in R'}$. Thus, it must be the case that there exists an index $i \in S \setminus \{j\}$ such that $\langle X_i, X_j \rangle \neq 0$. Applying Lemma \ref{lemma:xj} (part 1) to $B$ and $\widehat{B}$ we have that $|B_j| > \sqrt{\frac{2\lambda_0}{1+2\lambda_2}}$ --- this contradicts (\ref{eq:bjexact}). Therefore, there exists an integer $N_j$ such that for every $r \geq N_j$, we have $j \in \text{Supp}(\beta^{r+1})$.

{The above argument says that no $j$ in the support of $B$ can be dropped infinitely often in the sequence $\{\beta_{k}\}$. Since $S$ has the largest support size, no coordinate can be added to $S$ infinitely often in the sequence $\{\beta_{k}\}$. This concludes the proof of Part 1.}

%By applying the above argument for every $j \in S$ we conclude that for every $r \geq \max_{j \in S} N_j$, we have $\text{Supp}(\beta^{r+1}) \supseteq S$. Moreover, there exists an integer $N'$ such that for $r \geq N'$, we have $\text{Supp}(\beta^{r+1}) \subseteq S$ (otherwise, there will be a limit point with a support size greater than that of $B$). Thus, for $r \geq \max \{ N' , \max_{j \in S} N_j \}$, we have that $\text{Supp}(\beta^{r+1}) = \text{Supp}(\beta^{r})$. By induction on $r$, we conclude that for any $r \geq \max \{ N' , \max_{j \in S} N_j \}$, $\text{Supp}(\beta^{r+i}) = \text{Supp}(\beta^{r})$ for every $i \in \mathbb{N}$, leading to statement 1 of this theorem.

\textbf{Part 2.)} Finally, we show that the limit of $\{ \beta^k \}$ is a CW minimum.
%-- have we shown that $\beta^k$ converges? HH: This is a direct implication of Part 1 and lemma \ref{lemma:spacer} - I will make it clear. -- HH: Fixed this.} 
To this end, note that the results of Part 1 (above) and Lemma \ref{lemma:spacer} (Parts 3 and 4) imply that $\beta^k$ converges to the limit $B$, which satisfies $\text{Supp}(B) = S$, and for every $i \in S$, we have:
\begin{equation}
B_i = \sign(\widetilde{B}_i) \frac{|\widetilde{B}_i| - \lambda_1}{1+2\lambda_2} \ \text{ and } \ |B_i| \geq \sqrt{2\lambda_0 \over 1+2 \lambda_{2}}. \label{eq:betaik}
\end{equation}
Fix some $j \notin \text{Supp}(B)$ and let $\{ \beta^{k'} \}_{k' \in K'}$ be the sequence of non-spacer iterates at which coordinate $j$ is updated. For every $k'$ after support stabilization, the algorithm maintains:
\begin{align*}
\frac{|\widetilde{\beta}^{k'}_j| - \lambda_1}{1+2\lambda_2} < \sqrt{2\lambda_0 \over 1+2 \lambda_{2}}
\end{align*}
where $\widetilde{\beta}^{k'}_j = \langle y - \sum_{i \neq j} X_i \beta_i^{k'}, X_j \rangle$. Taking  $k' \to \infty$ in the above, we have:
\begin{align}
\frac{|\widetilde{B}_j| - \lambda_1}{1+2\lambda_2} \leq \sqrt{2\lambda_0 \over 1+2 \lambda_{2}}. \label{eq:convergenceoutsupp}
\end{align}
(\ref{eq:betaik}) and~(\ref{eq:convergenceoutsupp}) together imply that $B$ is a CW minimum (by definition).
\end{proof}

\subsection{Proof of Theorem \ref{theorem:convergencerate}}
\begin{proof}
By Theorem~\ref{theorem:convergence}, we have $\beta^K \to B$, and there exists an integer $M$ such that for all $K \geq M$, we have $\text{Supp}(\beta^{K}) = S$ and $\text{Supp}(B) = S$. Therefore, there exists an integer $N \geq M$ such that for $K \geq N$, $\sign(\beta^{K}_i) = \sign(B_i)$ for every $i \in S$.	For $K \geq N$, it can be readily seen that the iterates $\beta^{K}$ are the same as those generated by minimizing the following objective
    \begin{align}
        g(\beta_S) = \frac{1}{2} \| y - X_S \beta_S \|^{2} + \lambda_1 \sum\limits_{i \in S, B_i > 0} \beta_i - \lambda_1 \sum\limits_{i \in S, B_i < 0} \beta_i + \lambda_2 \|\beta_S\|^{2},
    \end{align}
    using coordinate descent with step size $1 \over 1+2\lambda_2$ and starting from the initial solution $\beta^{N}$. The function $\beta_S \mapsto g(\beta_S)$ is continuously differentiable and its gradient is Lipschitz continuous with parameter $L = M_S + 2\lambda_2$. Moreover, it is strongly convex with strong-convexity parameter $\sigma_S = m_S + 2\lambda_2$. \cite{BeckConvergence} (see Theorem~3.9) 
    %%\textcolor{red}{RM: mention the number of the theorem.} 
has proven a linear rate of convergence for cyclic CD when applied to strongly convex and continuously differentiable functions.
%(see Theorem 3.9 in \cite{BeckConvergence}). 
Applying \cite{BeckConvergence}'s result in our context leads to the conclusion of the theorem.
%using the previously defined $\sigma_S$ and $L$, we arrive to the conclusion of this theorem. 
\end{proof}

\subsection{Proof of Theorem \ref{theorem:cd-kswaps}}
\begin{proof}
Before it terminates, Algorithm~2 leads to a sequence $\{\beta^{i}\}_{0}^{\ell}$ such that $F(\beta^{\ell}) < F(\beta^{\ell-1}) < \dots < F(\beta^{1})$.
Since $\beta^{\ell}, \beta^{\ell-1}, \dots, \beta^{1}$ are all outputs of Algorithm \ref{alg:CD}, they are all CW minima (by Theorem \ref{theorem:convergence}). Any CW minimum on a support $S$ is stationary for the problem: $\min_{\beta_S} f(\beta_S)$. By the convexity of $\beta_S \to f(\beta_S)$, all stationary solutions on support $S$ have the same objective (since they all correspond to the minimum of $\min_{\beta_S} f(\beta_S)$).   
Thus, we have $\text{Supp}(\beta^{i}) \neq \text{Supp}(\beta^{j})$ for any $1 \leq i, j \leq \ell$ such that $i \neq j$.
Therefore, a support can appear at most once during the course of Algorithm~2. Since the number of possible supports is finite, we conclude that Algorithm~2 terminates in a finite number of iterations. Finally, we note that Algorithm~2 terminates iff there is no feasible solution $\widehat{\beta}$ for (\ref{eq:psi}) satisfying $F(\widehat{\beta}) < F(\beta^{\ell})$. This implies that $\beta^{\ell}$ is a minimizer of (\ref{eq:psi}) and thus a PSI($k$) minimum (by Definition \ref{def:psi}).
\end{proof}

\subsection{Proof of Lemma \ref{lemma:nextlambda}}
\begin{proof}
Let us consider the case where, $\lambda_0^{i+1} < M^i$. 
It follows from~\eqref{eq:Mi} that:
%By writing $M^i$ explicitly and rearranging we have 
\begin{align}
  \max_{j \in S^c}~\frac{ \left|~ \left|\langle r , X_j \rangle \right|  - \lambda_1  \right|}{1+2\lambda_2} >  \sqrt{\frac{2 \lambda_0^{i+1}}{1 + 2\lambda_2}},
\end{align}
which implies that $\beta^{(i)}$ is not a CW minimum for the given $\lambda_0^{i+1}$ 
(see~\eqref{lemma:CW-eqn}). By Theorem \ref{theorem:convergence}, Algorithm~\ref{alg:CD}  converges to a CW minimum. Therefore, 
Algorithm~\ref{alg:CD} initialized with $\beta^{(i)}$ leads to $\beta^{(i+1)} \neq \beta^{(i)}$.

We now consider the case where, $\lambda^{i+1}_0 \in (M^i , \lambda^i_0 ]$. Then~\eqref{eq:Mi} implies
\begin{align}
\label{eq:gridoutside}
\max_{j \in S^c} ~ \frac{ |~ |\langle r , X_j \rangle|  - \lambda_1 |}{1+2\lambda_2} < \sqrt{\frac{2 \lambda_0^{i+1}}{1 + 2\lambda_2}} \leq \sqrt{\frac{2 \lambda_0^{i}}{1 + 2\lambda_2}}.
\end{align}
Also, since $\beta^{(i)}$ is a CW minimum for  $\lambda  = \lambda^i_0$, we have for every $j \in S$
\begin{equation}
\label{eq:gridinside}
|\beta^{(i)}_j| \geq \sqrt{\frac{2 \lambda_0^{i}}{1 + 2\lambda_2}} \geq \sqrt{\frac{2 \lambda_0^{i+1}}{1 + 2\lambda_2}},
\end{equation}
where, the second inequality follows from $\lambda_0^{i+1} \leq \lambda_0^{i}$. The condition %$\beta^{(i)}$ for 
$\nabla_{S} f(\beta^{(i)}_S) =0$ along with inequalities (\ref{eq:gridoutside}) and (\ref{eq:gridinside}) imply that $\beta^{(i)}$ is a CW minimum for Problem~\eqref{problem:main} at $\lambda_0 = \lambda^{i+1}_0$. Therefore, $\beta^{(i)}$ is a fixed point for Algorithm~\ref{alg:CD}.
%when initialized with $\beta^{(i)}$ outputs $\beta^{(i)}$ is a fixed point for . 
\end{proof}

%%%%%%%~~~~~~~~~~~~~~~~~~~~~~~

\section{Appendix: Oracle Tuning}

\subsection{Statistical Performance for Varying Number of Samples}
\begin{figure}[H]
\centering
Exponential Correlation, $\rho = 0.9$, $p = 1000$, $k^{\dagger} = 20$, SNR $=5$
\includegraphics[scale=0.47]{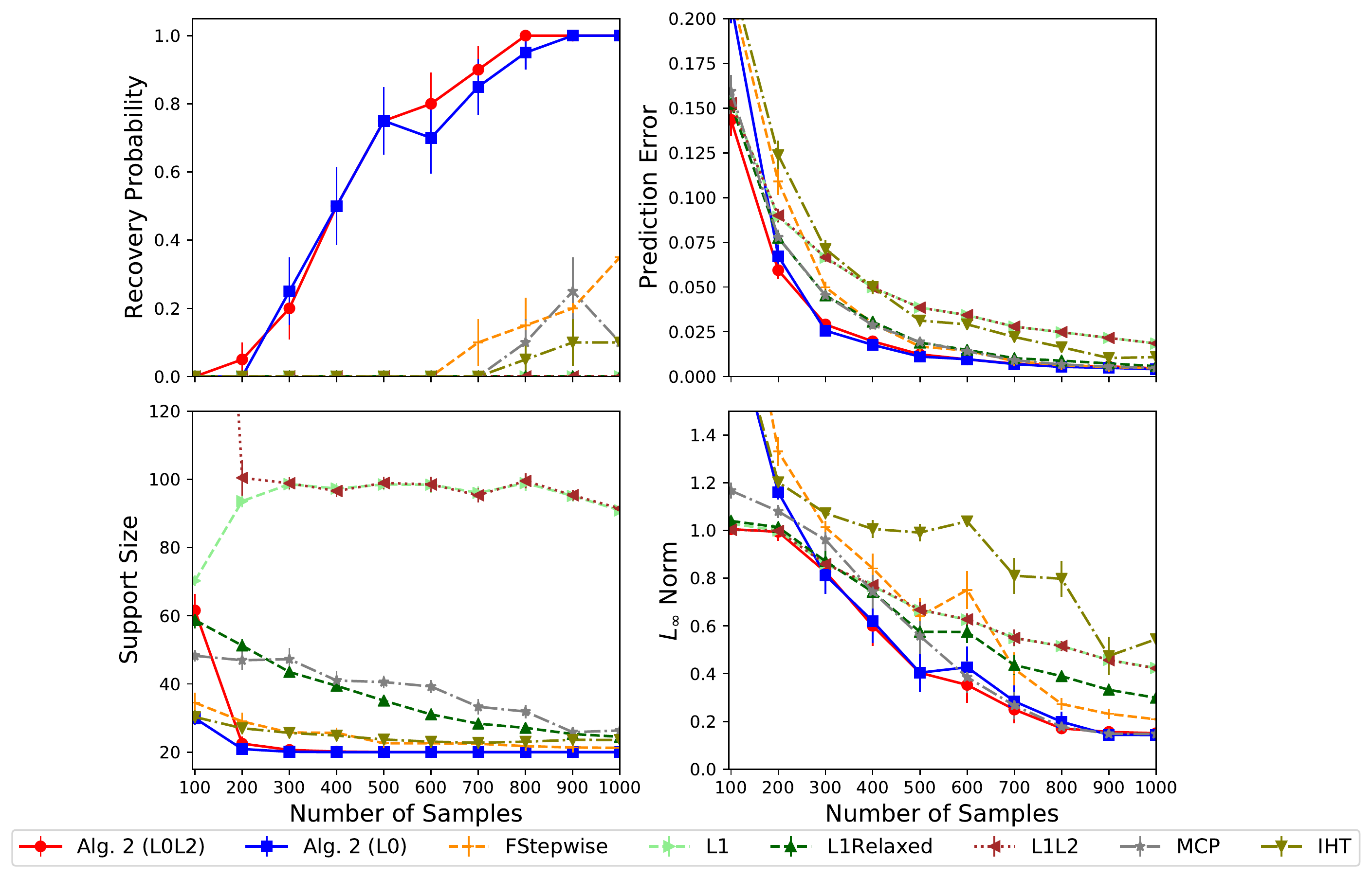}
\includegraphics[scale=0.47]{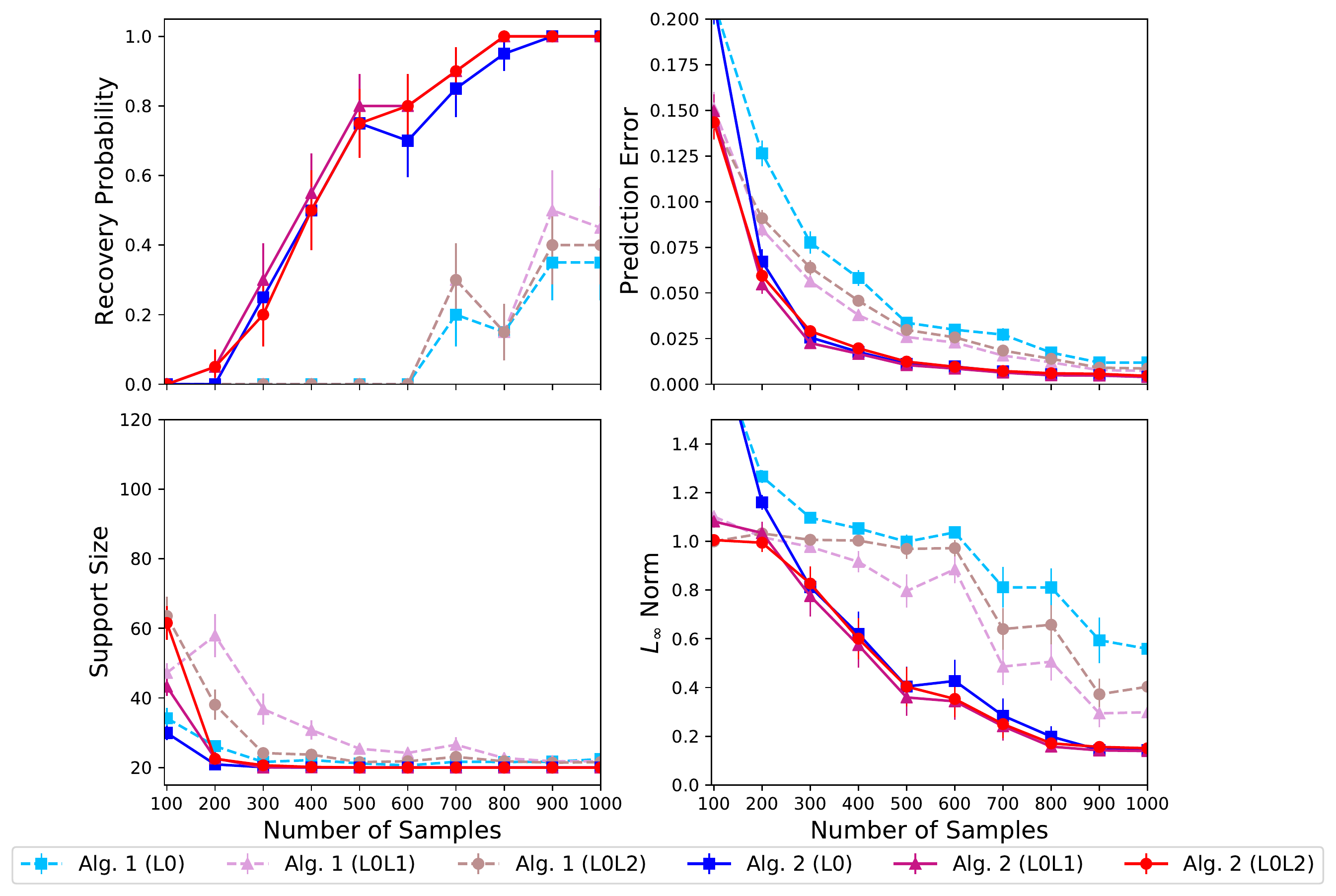}
\caption{{\small{Performance measures as the number of samples $n$ varies between $100$ and $1000$. The top figure compares Algorithm 2 ($L_0$), Algorithm 2 $(L_0 L_2)$, and other state-of-the-art algorithms. The bottom figure compares all of our proposed algorithms.} }}
\label{fig:NSweep-Exp09-Oracle}
\end{figure}

\begin{figure}[H]
\centering
{\sf {Exponential Correlation, $\rho = 0.5$, $p = 1000$, $k^{\dagger} = 20$, SNR $=5$ }}
\includegraphics[scale=0.47]{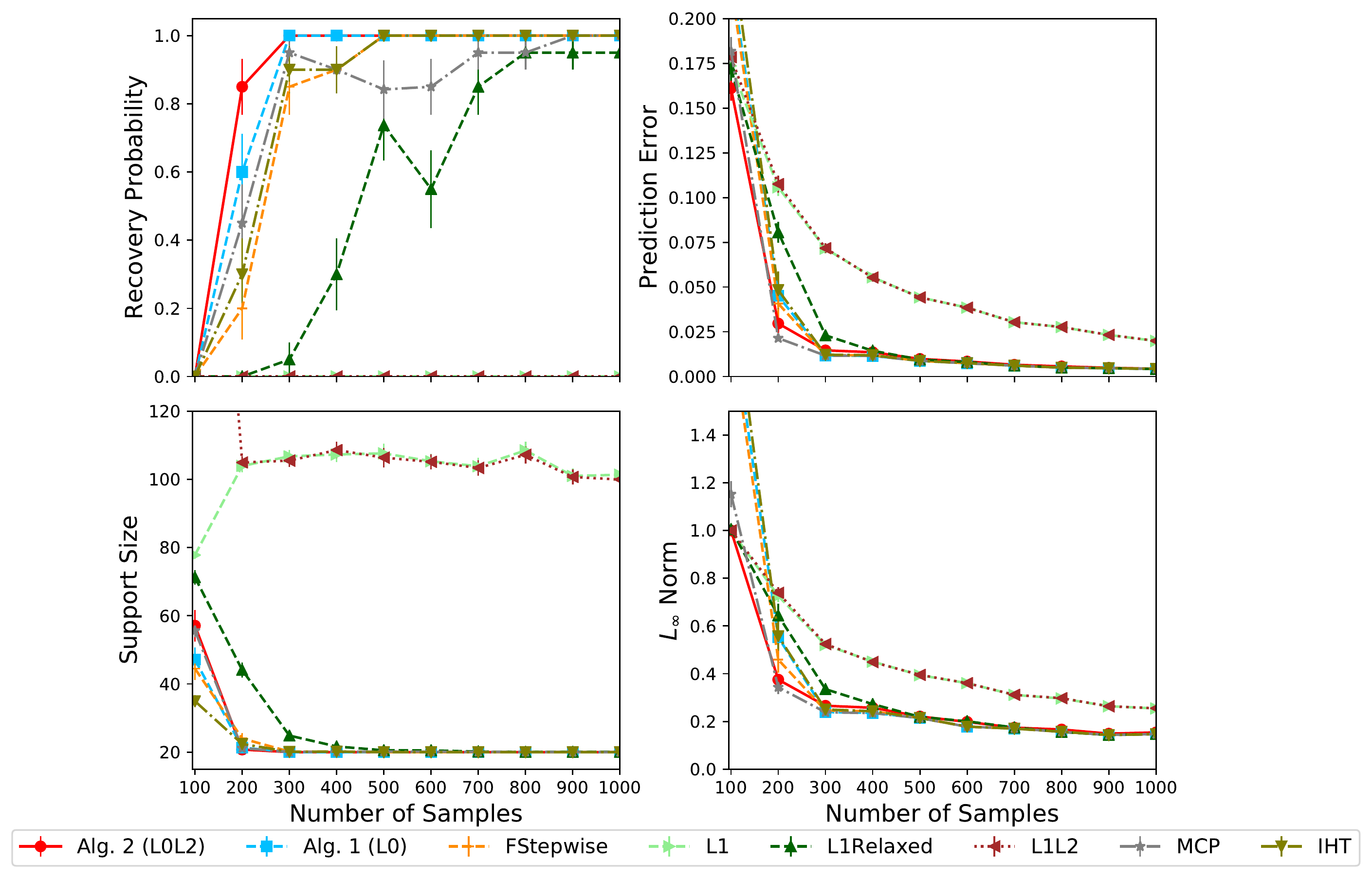}
\caption{{\small{Performance measures as the number of samples $n$ varies between $100$ and $1000$. The figure compares  Algorithm 1 ($L_0$), Algorithm 2 $(L_0 L_2)$, and other state-of-the-art algorithms.} }}
\label{fig:NSweep-Exp05-Oracle}
\end{figure}

\subsection{Statistical Performance for Varying SNR}

\begin{figure}[H]
\centering
Constant Correlation, $\rho = 0.4$, $n=1000$, $p = 2000$, $k^{\dagger} = 50$
\includegraphics[scale=0.47]{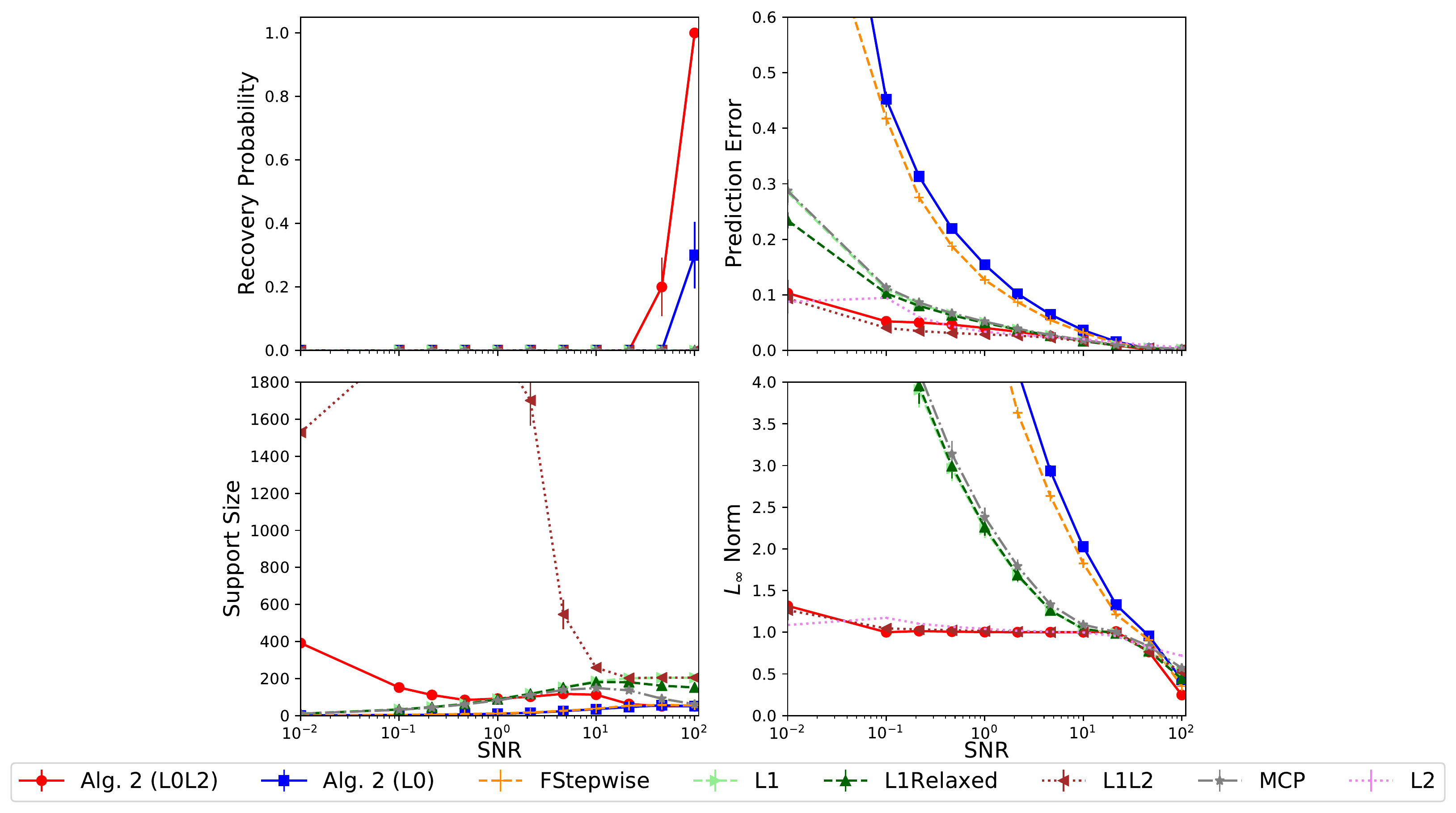}
\caption{Performance measures as the signal-to-noise ratio (SNR) is varied between 0.01 and 100. The  figure compares two of our methods and other state-of-the-art algorithms. }
\label{fig:SNRSweep-C-Oracle}
\end{figure}
\begin{figure}[H] 
\centering
Exponential Correlation, $\rho = 0.5$, $n=1000$, $p = 5000$, $k^{\dagger} = 50$
\includegraphics[scale=0.47]{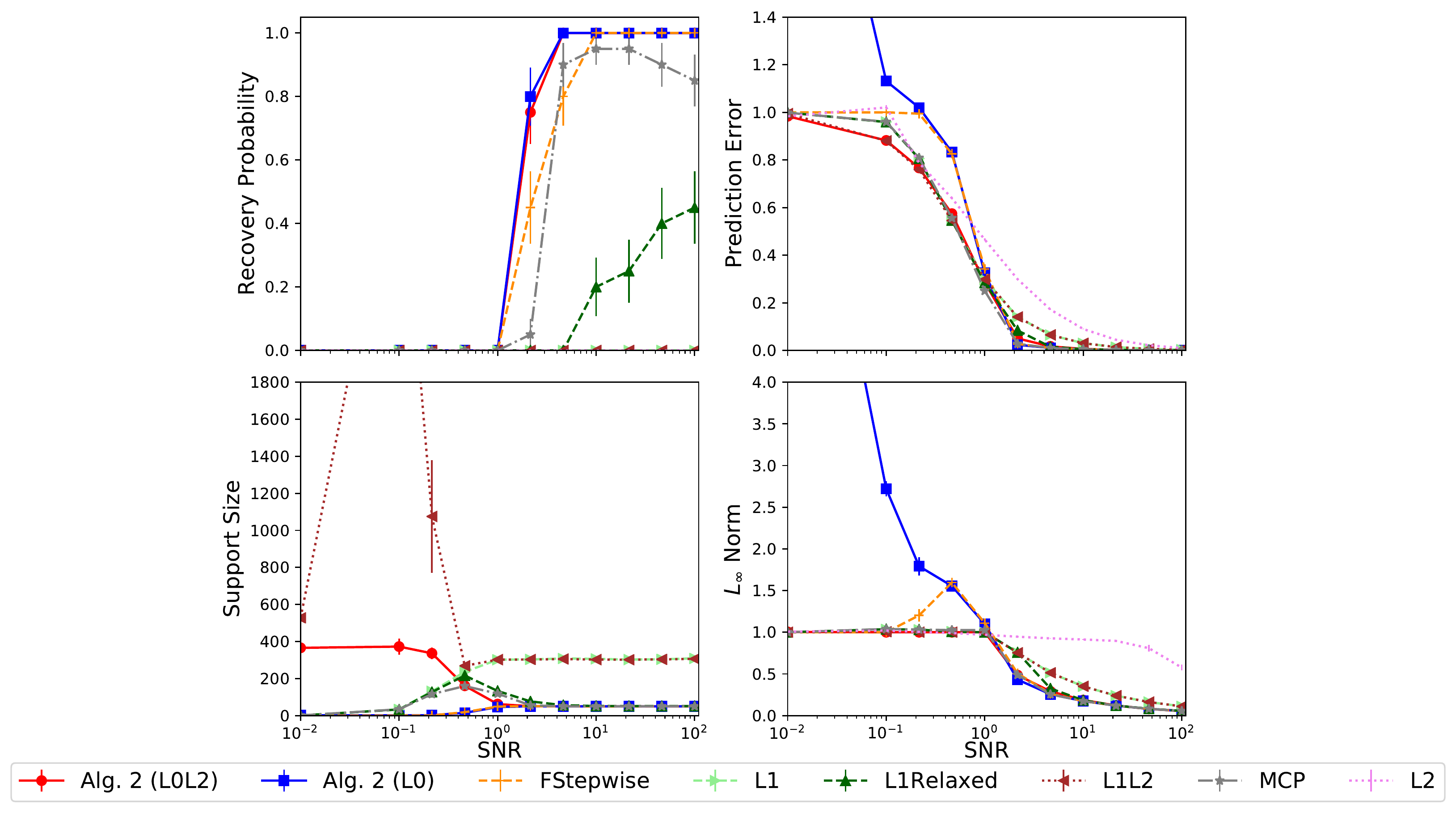}
\caption{Performance measures as the signal-to-noise ratio (SNR) is varied between 0.01 and 100. The  figure compares two of our methods and other state-of-the-art algorithms. }
\label{fig:SNRSweep-Exp-Oracle}
\end{figure}

\newpage
\section{Appendix: Random Design Tuning}
We use the same definition of Relative Risk (RR) as in \cite{onbestsubset}:  given an estimator $\hat{\beta}$,  RR$ = \frac{(\hat{\beta} - \beta^{\dagger})^{T} \Sigma (\hat{\beta} - \beta^{\dagger}) }{(\beta^{\dagger})^T \Sigma \beta^{\dagger}}$.
\vspace{-0.7cm} %%%
\subsection{Statistical Performance for Varying Number of Samples}
\begin{figure}[H]
\centering
Exponential Correlation, $\rho = 0.9$, $p = 1000$, $k^{\dagger} = 20$, SNR $=5$
\includegraphics[scale=0.47]{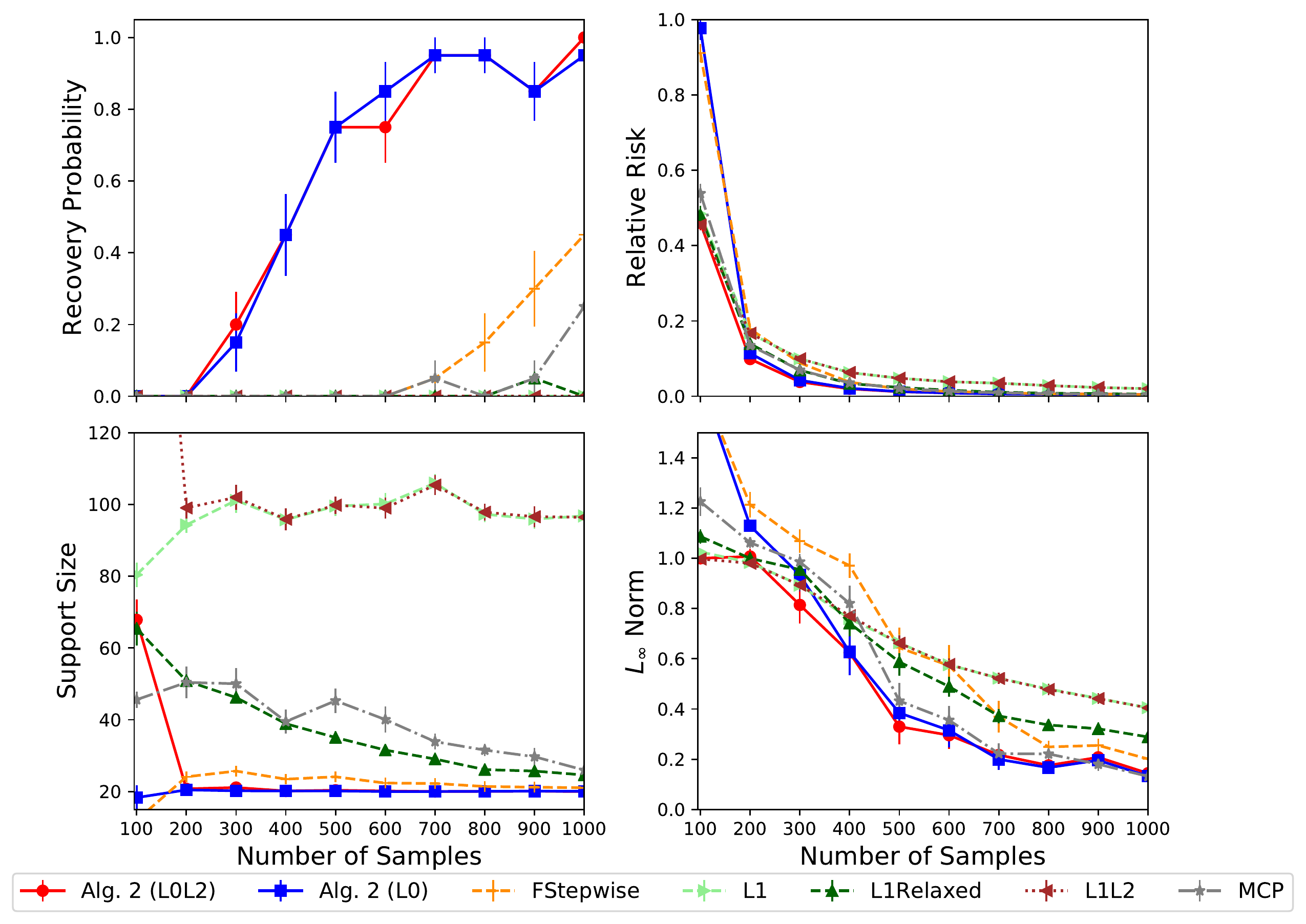}
\includegraphics[scale=0.47]{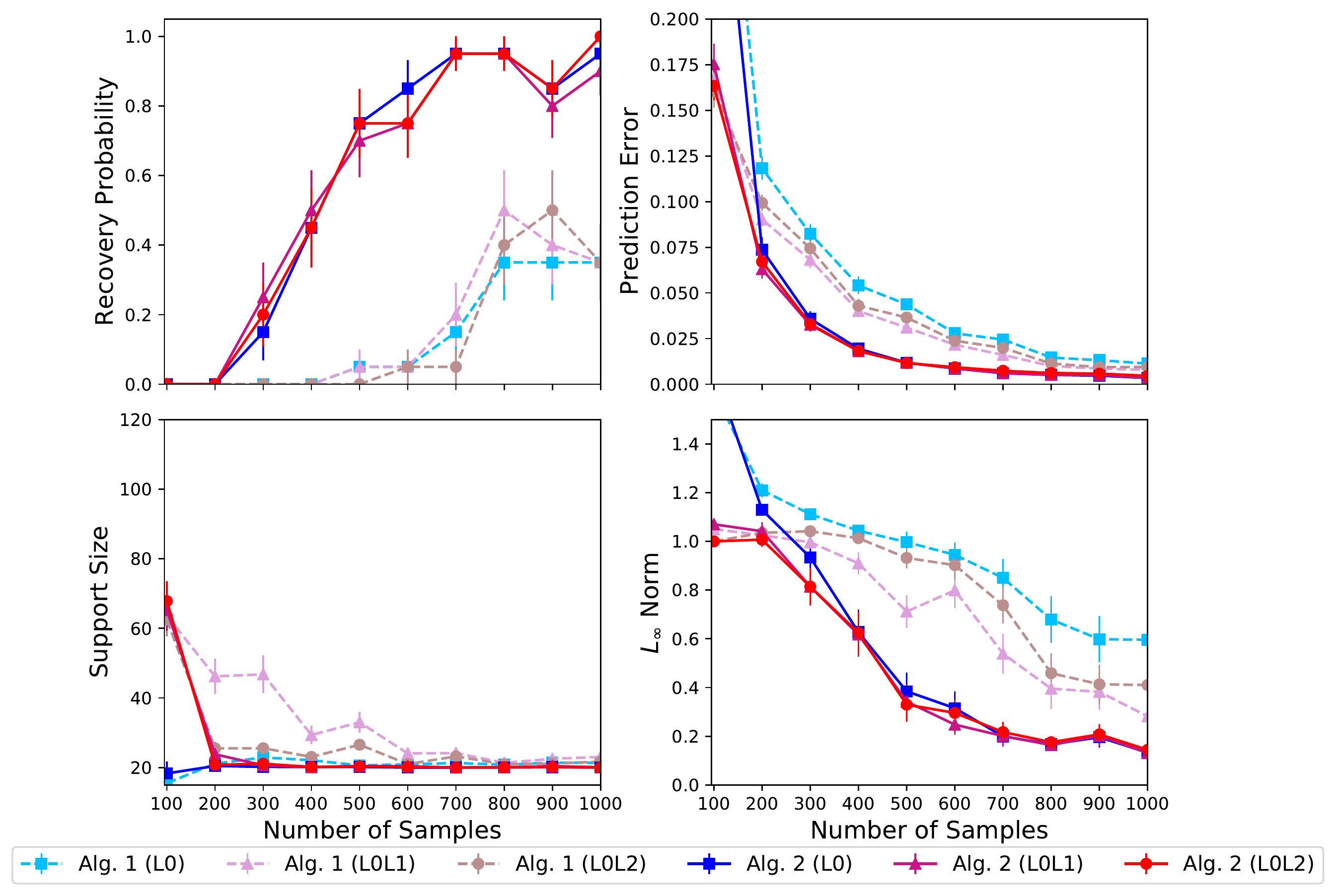}
\caption{{\small{Performance measures as the number of samples $n$ varies between $100$ and $1000$. The top figure compares Algorithm 2 ($L_0$), Algorithm 2 $(L_0 L_2)$, and other state-of-the-art algorithms. The bottom figure compares all of our proposed algorithms.} }}
\label{fig:NSweep-Exp09-Random}
\end{figure}

\begin{figure}[H]
\centering
{\sf {Exponential Correlation, $\rho = 0.5$, $p = 1000$, $k^{\dagger} = 20$, SNR $=5$ }}
\includegraphics[scale=0.47]{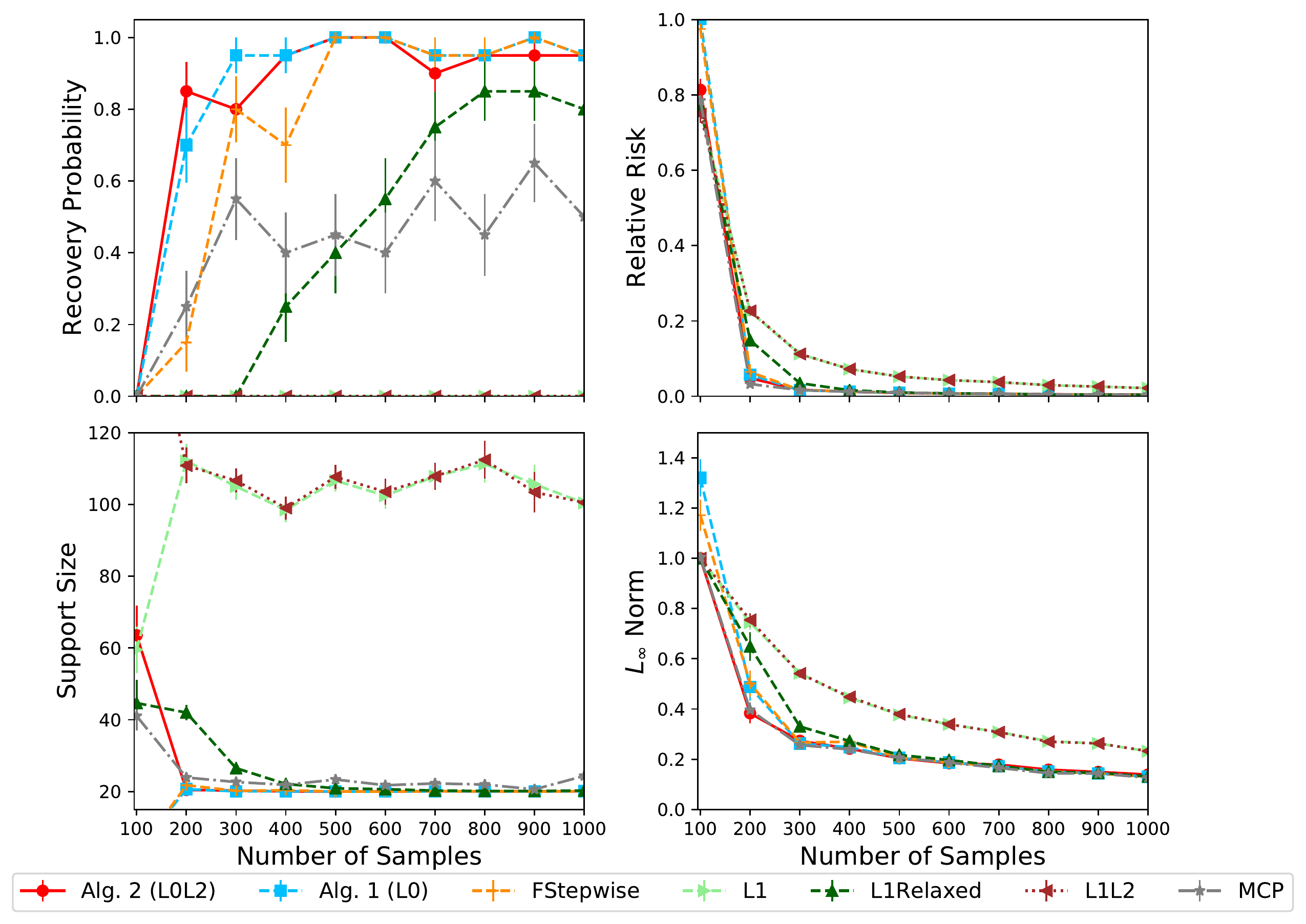}
\caption{{\small{Performance measures as the signal-to-noise ratio (SNR) is varied between 0.01 and 100. The figure compares two of our methods and other state-of-the-art algorithms.} }}
\label{fig:NSweep-Exp05-Random}
\end{figure}

\subsection{Statistical Performance for Varying SNR}
\begin{figure}[H]
\centering
Constant Correlation, $\rho = 0.4$, $n=1000$, $p = 2000$, $k^{\dagger} = 50$
\includegraphics[scale=0.47]{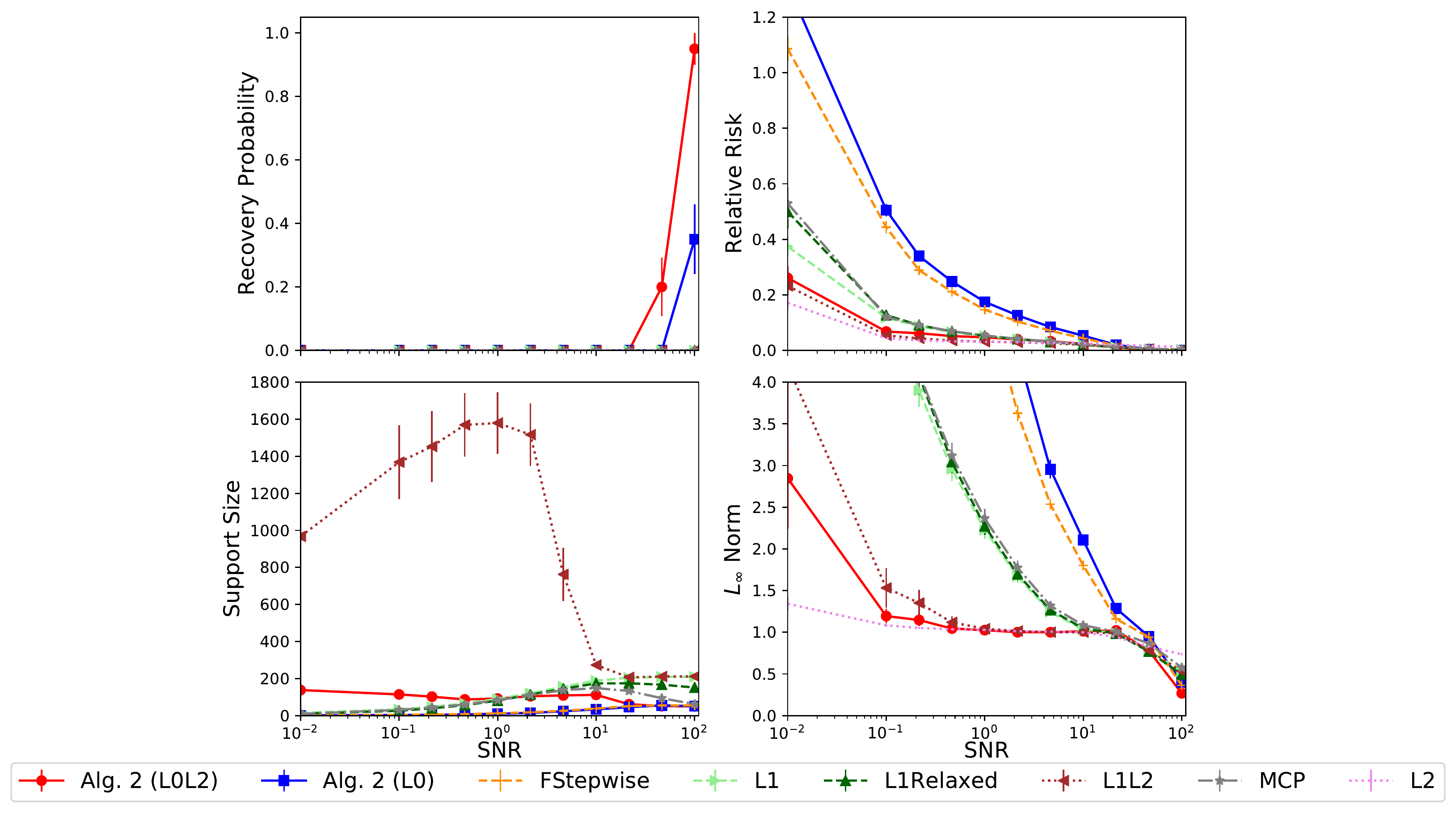}
\caption{Performance measures as the signal-to-noise ratio (SNR) is varied between 0.01 and 100. The  figure compares two of our methods and other state-of-the-art algorithms. }
\label{fig:SNRSweep-C-Random}
\end{figure}

\begin{figure}[H] 
\centering
Exponential Correlation, $\rho = 0.5$, $n=1000$, $p = 5000$, $k^{\dagger} = 50$
\includegraphics[scale=0.47]{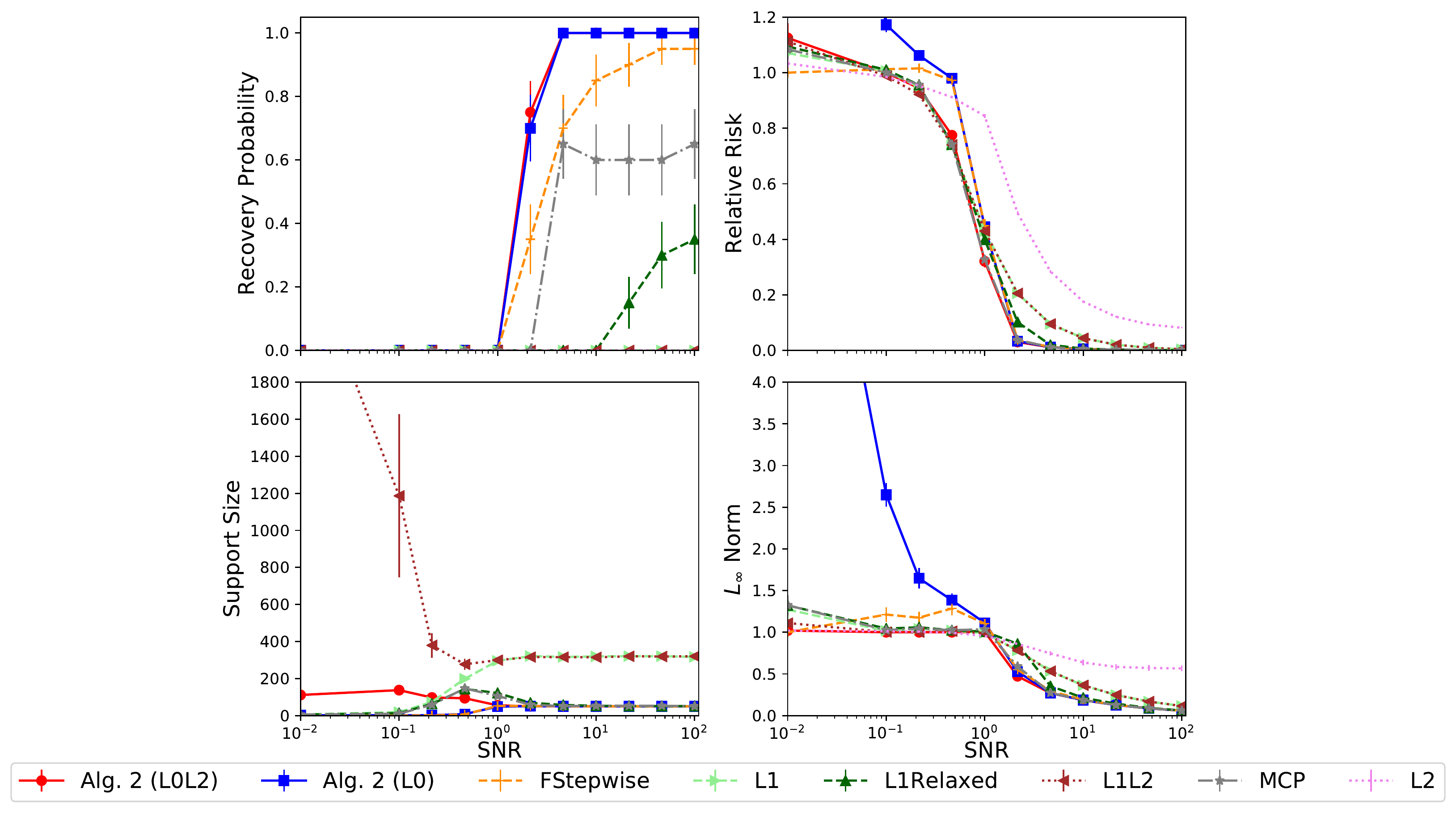}
\caption{Performance measures as the signal-to-noise ratio (SNR) is varied between 0.01 and 100. The  figure compares two of our methods and other state-of-the-art algorithms. }
\label{fig:SNRSweep-Exp-Random}
\end{figure}

\section{Comparisons among PSI$(k)$ and FSI$(k)$}
In Figure \ref{fig:CDKSwapsEvolution}, we show the evolution of solutions when running the FSI(5) variant of Algorithm~\ref{alg:mipcd} on the dataset of Section~\ref{swap-inescp-minima}. The algorithm starts from a CW minimum and iterates between running CD-PSI($1$) and finding a swap that improves the objective by solving optimization problem (\ref{mio:fsi}) using MIO. The PSI(1) minima obtained from running CD-PSI($1$) are marked by red circles and the results obtained by solving problem (\ref{mio:fsi}) using MIO are denoted by blue squares.

\begin{figure}[H]
\centering
\scalebox{0.99}{\includegraphics[height = .25\textheight, width=\textwidth]{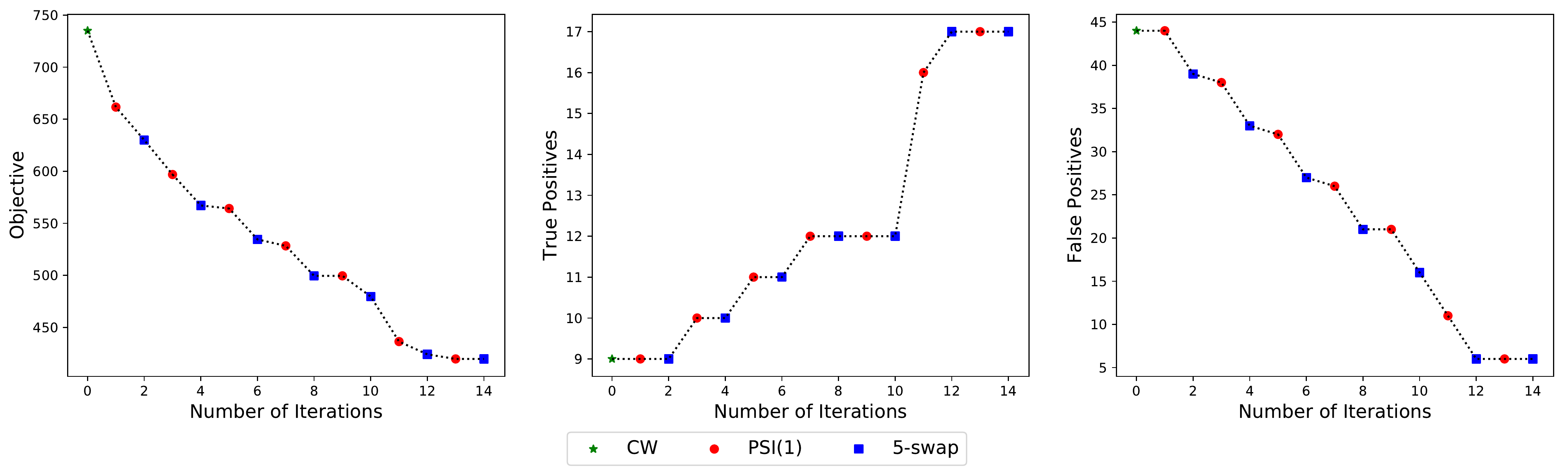}}
\caption{Evolution of solutions during the course of a variant of Algorithm~2 where we successively run CD-PSI(1) and solve combinatorial problem~\eqref{mio:fsi} to generate an FSI($5$) minimum. }
\label{fig:CDKSwapsEvolution}
\end{figure}

Figure \ref{fig:CDKSwapsEvolution} shows that running CD-PSI($1$) on top of the solutions obtained by MIO leads to important gains in terms of better objective values, for most of the cases. This also confirms our intuition that  MIO can lead to solutions that are not available via PSI(1). From the plot of true positives and false positives, we can see that CD-PSI($1$) improves the solution by increasing the true positives whereas MIO improves the solution by removing false positives. This observation confirms the behavior we noticed in Figure 6, where PSI(1) minima were successful in obtaining a good number of true positives, but suffered in terms of false positives.

\end{appendix}

% \clearpage
% \bibliographystyle{plainnat_my}
% \small{\bibliography{ref}}

\end{document}